\documentclass{article} 
\usepackage{natbib}
\usepackage[final]{neurips_2019}
\setcitestyle{authoryear}

\usepackage[utf8]{inputenc}
\usepackage{graphicx}

\usepackage{tikz}

\usepackage{amsthm, amsmath, amssymb, amsfonts}
\usepackage{enumerate}
\usepackage{mathrsfs}
\usepackage{microtype}
\usepackage{bbm}
\usepackage{mathtools}
\usepackage{dsfont}
\usepackage{hyperref}
\usepackage{url}
\usepackage{authblk}
\usepackage{pifont}
\usepackage[linewidth=1pt]{mdframed}
\usepackage{framed}
\usepackage{lipsum}
\usepackage{caption}

\usepackage{apptools}
\usepackage{subcaption}
\usepackage{wrapfig}
\usepackage{aliascnt}
\usepackage{algorithm, tabularx}
\usepackage{algorithmic}
\usepackage{tabu}
\usepackage{boldline,multirow}
\usetikzlibrary{calc,trees,positioning,arrows,chains,shapes.geometric,decorations.pathreplacing,decorations.pathmorphing,shapes,matrix,shapes.symbols, math}

\usepackage{csquotes}

\usepackage{shorthands}

\makeatletter
\patchcmd{\maketitle}
{\def\@makefnmark}
{\def\@makefnmark{}\def\useless@macro}
{}{}
\makeatother

\title{Quant GANs: \\Deep Generation of Financial Time Series}
\author[ ]{Magnus Wiese\textsuperscript{1, 2,}\thanks{Corresponding author: \texttt{quant.gans@gmail.com}}}
\author[2]{Robert Knobloch}
\author[1, 2]{Ralf Korn}
\author[1, 2]{\\ \vspace{1em}\hspace{2.5em} Peter Kretschmer}
\affil[1]{TU Kaiserslautern, Gottlieb-Daimler-Straße 48, 67663 Kaiserslautern, Germany
}
\affil[2]{Fraunhofer ITWM, Fraunhofer-Platz 1, 67663 Kaiserslautern, Germany
}
\date{\today}

\newcommand{\parencite}[2][]{\cite[#1]{#2}}
\newcommand{\autocite}[2][]{\cite[#1]{#2}}

\begin{document}

\maketitle
\begin{abstract}
Modeling financial time series by stochastic processes is a challenging task and a central area of research in financial mathematics. As an alternative, we introduce Quant GANs, a data-driven model which is inspired by the recent success of generative adversarial networks (GANs). Quant GANs consist of a generator and discriminator function, which utilize temporal convolutional networks (TCNs) and thereby achieve to capture long-range dependencies such as the presence of volatility clusters. The generator function is explicitly constructed such that the induced stochastic process allows a transition to its risk-neutral distribution. Our numerical results highlight that distributional properties for small and large lags are in an excellent agreement and dependence properties such as volatility clusters, leverage effects, and serial autocorrelations can be generated by the generator function of Quant GANs, demonstrably in high fidelity.
\end{abstract}

\section{Introduction}
Since the ground-breaking results of AlexNet \parencite{alexnet} at the ImageNet competition, neural networks (NNs) excel in various areas ranging from generating realistic audio waves \parencite{wavenet} to surpassing human-level performance on ImageNet \parencite{prelu_he} or beating the world champion in the game Go \parencite{alphago}. While NNs have already become standard tools in image analysis, the application in finance is still in early stages. Citing a few, \cite{deep_hedging, Buehler2019b} use NNs to hedge large portfolios of derivatives, \cite{deep_optimal_stopping} solve optimal stopping problems and \cite{deep_fraud_detection, Schreyer2019b, Schreyer2019a} detect anomalies in accounting data and show how auditing firms can be attacked by DeepFakes.

In this article, we consider the problem of approximating a realistic asset price simulator by using NNs and adversarial training techniques. Such a path simulator is useful as it can be used to extend and enrich limited real-world datasets, which in turn can be used to fine-tune or robustify financial trading strategies. 

To approximate a data-driven path simulator we propose \emph{Quant GANs}. Quant GANs are based on the application of generative adversarial networks (GANs) \parencite{gans} and are located between pure data-based approaches such as historical simulation and model-driven methods such as Monte Carlo simulation assuming an underlying stock price model like the Black Scholes model \parencite{blackscholes}, the Heston stochastic volatility model \parencite{heston} or L\'evy process based modeling \parencite{tankov2003financial}. 

Using two different NNs as opponents is the fundamental principle of GANs. While one NN, the so-called generator, is responsible for the generation of stock price paths, the second one, the discriminator, has to judge whether the generated paths are synthetic or from the same underlying distribution as the data (i.e. the past prices). 

Various pitfalls exist when training GANs ranging from limited convergence when optimizing both networks simultaneously (cf. \cite{Arjovsky2017, mescheder_which_gan}) to extrapolation problems when using recurrent generation schemes. To address the latter issue we propose the use of \emph{temporal convolutional networks} (\emph{TCNs}), also known as \textit{WaveNets} \parencite{wavenet}, as the generator architecture. A TCN generator comes with multiple benefits: it is particularly suited for modeling long-range dependencies, allows for parallization and guarantees stationarity. 

\begin{figure}[H]
	\centering
	\begin{subfigure}[b]{0.5\textwidth}
		\includegraphics[width=\textwidth]{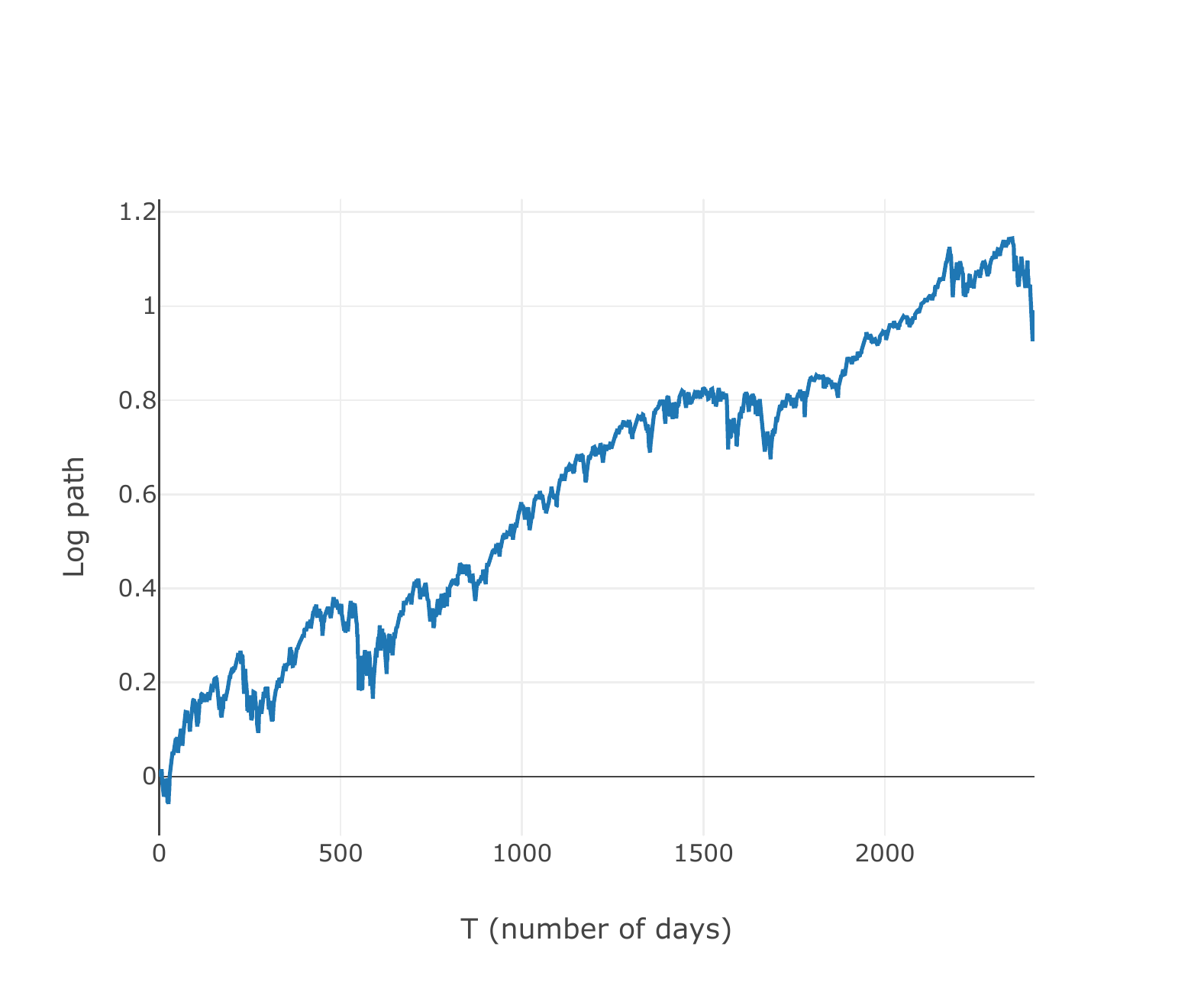}
		\caption{S\&P 500 log path}
		\label{subfig:spy_logrtn}
	\end{subfigure}
	\hspace{-2em}
	\begin{subfigure}[b]{0.5\textwidth}
		\includegraphics[width=\textwidth]{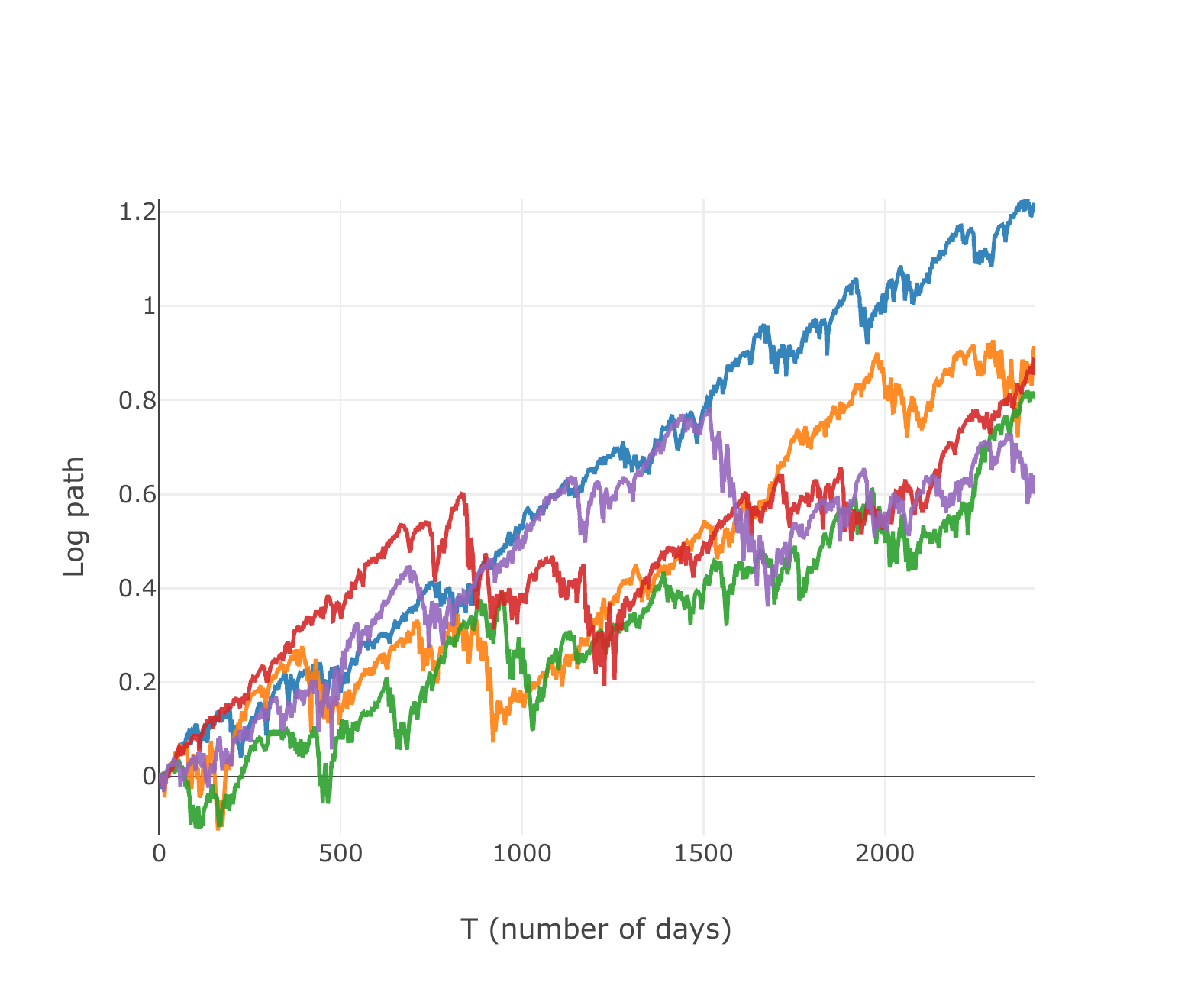}
		\caption{Generated log paths of a Quant GAN}
		\label{subfig:qgan_logrtn_sample}
	\end{subfigure}
	\caption{Comparison of the S\&P 500 index in log-space (a) with generated log paths of a calibrated Quant GAN (b).}
	\label{fig:comparison_sp500_qgan}
\end{figure}

One of our main contributions is the rigorous mathematical definition of TCNs for the first time in the literature. As this definition requires the use of heavy notation, we illustrate it by simple examples and graphical representations. This should help to make the definition accessible for mathematicians that are not familiar with the specialized language of neural nets. 

After introducing NNs (and in particular TCNs) in \hyperref[sec:neural_network_topologies]{Section~\ref*{sec:neural_network_topologies}} and GANs in \hyperref[sec:gans]{Section~\ref*{sec:gans}}, we  define the proposed generator architecture of Quant GANs, namely \textit{Stochastic Volatility Neural Networks (SVNNs)} in \hyperref[sec:svnns]{Section~\ref*{sec:svnns}}. In spirit of stochastic volatility models, the SVNN architecture consists of a volatility and drift TCN and an innovation NN. SVNNs are constructed such that the generated paths can be evaluated under their risk-neutral distribution and, as a special case, constrained to exhibit conditionally normal log returns.

\begin{figure}[htp]
	\begin{subfigure}[b]{0.48\textwidth}
		\includegraphics[width=\textwidth]{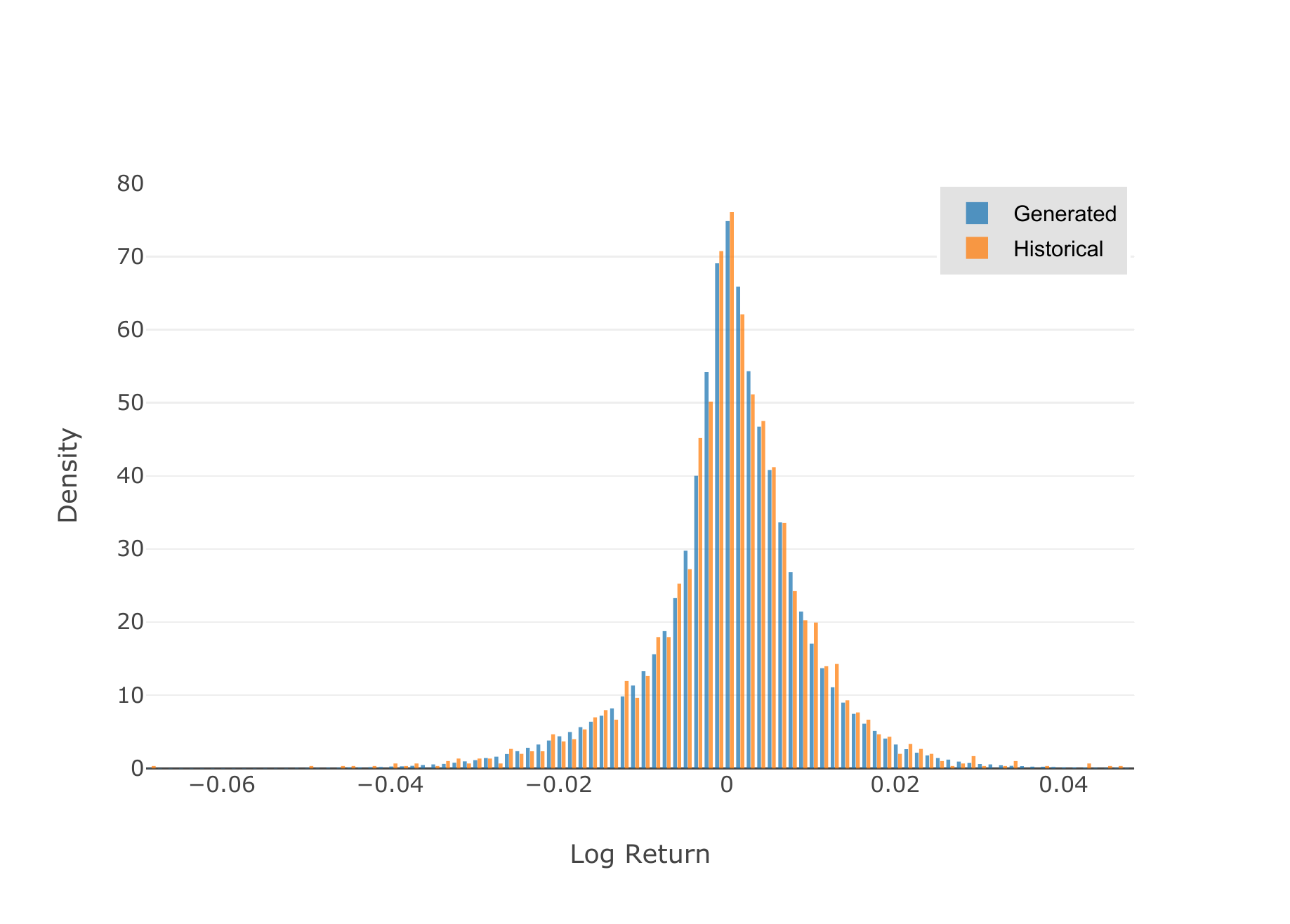}
		\caption{Histogram of daily log returns}
		\label{subfig:hist1_after_crisis}
	\end{subfigure}
	\begin{subfigure}[b]{0.48\textwidth}
		\includegraphics[width=\textwidth]{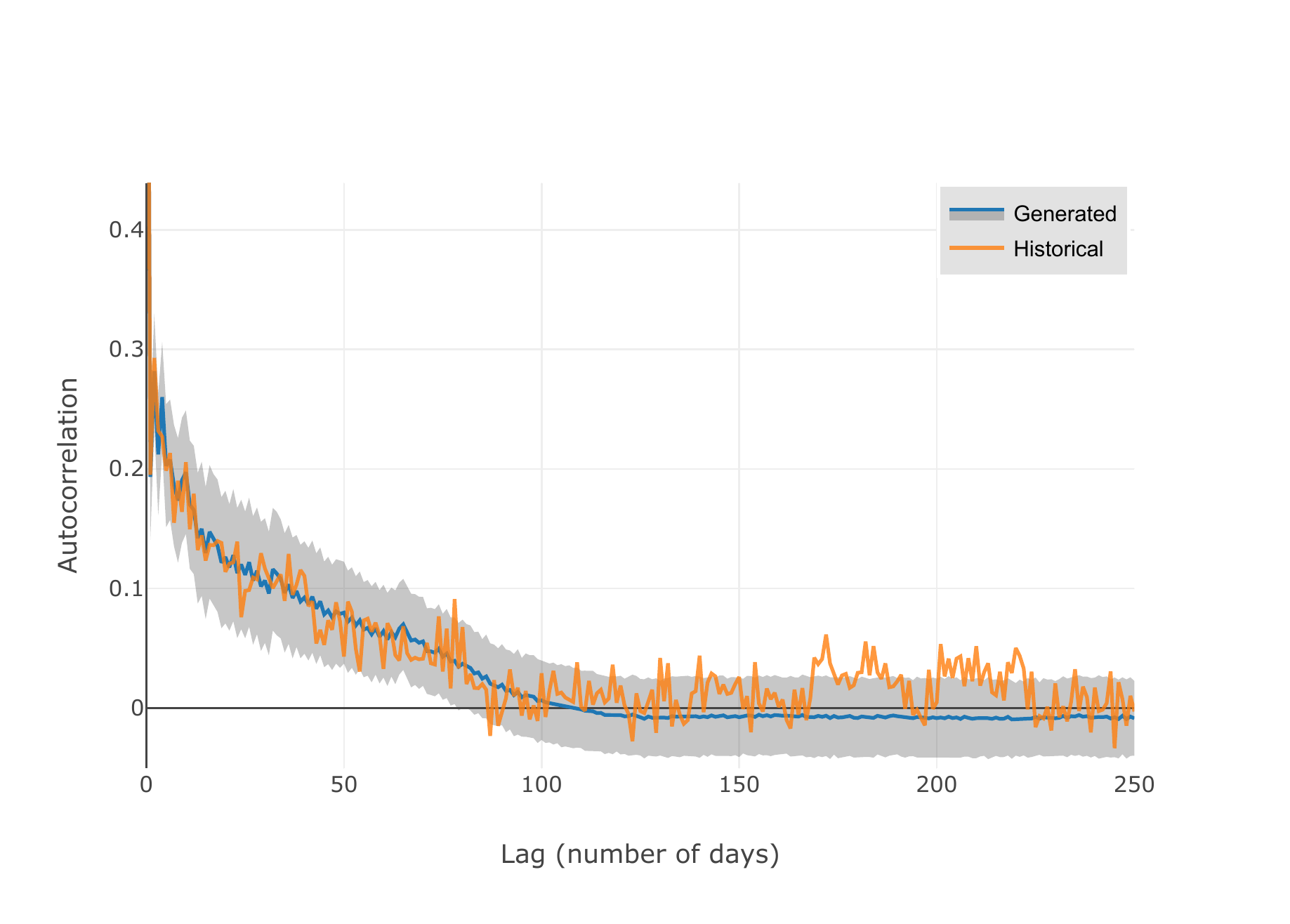}
		\caption{ACF of absolute log returns}
		\label{subfig:acf_abs}
	\end{subfigure}
	\caption{Figures (a) and (b) display a comparison of the generated Quant GAN (blue) and empirical S\&P 500 (orange) samples.}
	\label{fig:sample_properties}
\end{figure}

For SVNNs we prove theoretical $L^p$-space-related claims, which demonstrate that all moments of the process exist. Since financial time series are understood to exhibit heavy-tails, we show that the Lambert W transformation plays an essential role for generating heavy-tailed stock price returns by using methods for normalized Gaussian data. 

In \hyperref[sec:numerical]{Section~\ref*{sec:numerical}}, we present a numerical study applying Quant GANs to the S\&P 500 index from May 2009 - December 2018 (see \autoref{fig:comparison_sp500_qgan}). Our results demonstrate that SVNNs and TCNs outperform a typical {GARCH}-model with respect to distributional and dependence properties. \autoref{fig:sample_properties} illustrates a comparison of some sample properties of synthetic paths with the S\&P 500 index.

\section{Generative modeling of financial time series}
\label{sec:stylized_facts}
In this section, we briefly collect some facts and modeling issues of financial time series, which motivate our choice of NN types. Furthermore, we review existing literature related to (data-driven) modeling of financial time series. 

First, note that the performance of a stock over a certain period (as e.g. a day, a month or a year) is given by its relative return, either $R_t = (S_t - S_{t-1}) / S_{t-1}$ or its log return $R_t = \log(S_t) - \log(S_{t-1})$. Therefore, the generation of asset returns is the main objective of this paper. The characteristic properties of asset returns are well-studied and commonly known as \textit{stylized facts}. A list of the most important stylized facts includes (see e.g. \cite{stylized_facts_chakraborti, stylized_facts}):
\begin{itemize}
\item asset returns admit heavier tails than the normal distribution, 
\item the distribution of asset returns seems to be more peaked than the normal one,
\item asset returns admit phases of high activity and low activity in terms of price changes, an effect which is called \textit{volatility clustering},
\item the volatility of asset returns is negatively correlated with the return process, an effect named \textit{leverage effect},
\item empirical asset returns are considered to be uncorrelated but not independent.
\end{itemize}

A large literature of financial time series models exist ranging from a variety of discrete-time GARCH-inspired models \parencite{garch} to models in continuous-time such as \cite{blackscholes}, \cite{heston} and rough extensions \parencite{rough_heston}. However, the development and innovation of new models is difficult: it took 20 years to extend the Black-Scholes model, which assumes that asset prices can be described through geometric Brownian motions, to the more sophisticated Heston model, which accounted for stochastic volatility and the leverage effect. 

Pure data-driven modeling with NNs is a relatively new sub-field of research, which bears the potential of being able to model complicated statistical (perhaps unknown) dynamics. As a consequence, literature is rather sparse and can be summarized by four main manuscripts. \cite{Koshiyama2019} approximate a conditional model with GANs and define a score to select simulator / parameter candidates. \cite{deMeerPardo2019a} uses Wasserstein and relativistic GANs (cf. \cite{Gulrajani2017, Martineau2018}) and presents synthetic paths with volatility clusters and similar statistics. They also demonstrate that the use of synthetic paths can increase the accuracy of a trading strategy (cf. \cite{deMeerPardo2019b}). \cite{Takahashi2019} demonstrates that GANs can approximate various stylized facts. \cite{Wiese2019b} show that recurrent architectures can be used to generate equity option markets that satisfy static arbitrage constraints by using adversarial training and discrete local volatilities \parencite{BuehlerDLV}. 

However, \cite{deMeerPardo2019a} and \cite{Takahashi2019} both lack details on their NN architecture and do not state whether their proposed algorithm approximates the conditional or unconditional distribution. For all of the mentioned papers code is not available such that benchmarks are difficult to develop due to limited reproducibility.

\section{Neural network topologies}
\label{sec:neural_network_topologies}
In this section, we introduce the NN topologies that are essential to construct SVNNs. First, the {\it multilayer perceptron (MLP)} is defined. Afterward, we provide a formal definition of TCNs. 
\subsection{Multilayer perceptrons}
The MLP lies at the core of deep learning models. It is constructed by composing affine transformations with so-called \textit{activation functions}; non-linearities that are applied element-wise. \autoref{fig:nn_tikz} depicts the construction of an MLP with two hidden layers, where the input is three-dimensional and the output is one-dimensional. We begin with the activation function as the crucial ingredient and then define the MLP formally. 
\begin{definition}[Activation function]
	\label{def:activation_function}
	A function $\phi:\R\to\R$ that is Lipschitz continuous and monotonic is called {\it activation function}.
\end{definition}
\begin{remark}
	\autoref{def:activation_function} comprises a large class of functions used in the deep learning literature. Examples of activation functions in the sense of \autoref{def:activation_function} include $\tanh$, rectifier linear units  (\textit{ReLUs}) \parencite{relu_nair}, parametric ReLUs (\textit{PReLUs}) \parencite{prelu_he}, MaxOut \parencite{maxout} and a vast amount of other functions in the literature \cite{elu, selu, efficient_backprop, leaky_relu}. 
	Note that $\phi(0)=0$ is a desired property for the optimization of the networks parameters (cf. \cite{cs231n}) and is satisfied by the above activation functions. 
\end{remark}

\begin{definition}[Multilayer perceptron]
	\label{def:mlp}
	Let $L, N_0, \dots, N_{L +1} \in\N$, $\phi$ an activation function, $\Theta$ an Euclidean vector space and for any $ l \in \{1, \dots, L+1\} $ let $ \aff_l : \R^{N_{l-1}}  \rightarrow \R^{N_l}$ be an affine mapping. A function $f: \R^{N_0} \times \Theta \to \R^{N_{L+1}} $, defined by
	\begin{equation*}
	f( x, \theta) = \aff_{L+1} \circ f_{L} \circ \cdots \circ f_{1}( x),
	\end{equation*}
	where $\circ$ denotes the composition operator,
	\begin{equation*}
	f_l = \phi \circ \aff_{l} \quad \textrm{for all} \quad l\in\{1,\dots,L\}
	\end{equation*}
	and $\phi$ being applied component-wise, is called a {\it multilayer perceptron with $L$ hidden layers}. In this setting $N_0$ represents the \textit{input dimension}, $N_{L+1}$ the {\it output dimension}, $N_1,\dots,N_{L}$ the {\it hidden dimensions} and $a_{L+1}$ the \textit{output layer}. Furthermore, for any $l \in \{1,\dots,L+1\}$ the function $\aff_l$ takes the form $\aff_l: x \mapsto  W^{(l)} x +  b^{(l)}$ for some \textit{weight matrix} $ W^{(l)} \in \R^{N_l\times N_{l-1}}$ and \textit{bias} $ b^{(l)} \in  \R^{N_l}$. With this representation, the MLP's parameters are defined by
	\begin{equation*}
	\theta := \left( W^{(1)}, \dots,  W^{(L+1)},  b^{(1)}, \dots,  b^{(L+1)} \right) \in \Theta.
	\end{equation*}
\end{definition}
\begin{remark}
	\label{rem:nn_lip_prop}
	We call a function $f:\R^{d_0} \times \Theta \to \R^{d_1}$ with parameter space $\Theta$ a \textit{network}, if it is Lipschitz continuous.
\end{remark}
\def\layersep{2.5cm}
\def\numhidden{11}
\def\scalermlp{0.8}
\begin{figure}[h]
	\centering
\begin{tikzpicture}[shorten >=1pt,->,draw=black!50, node distance=\layersep]
    \tikzstyle{every pin edge}=[<-,shorten <=1pt]
    \tikzstyle{neuron}=[circle,fill=black!25,minimum size=12pt,inner sep=0pt, draw=black!80]
    \tikzstyle{input neuron}=[neuron, fill=blue!50];
    \tikzstyle{output neuron}=[neuron, fill=red!50];
    \tikzstyle{hidden1 neuron}=[neuron, fill=white!50];
    \tikzstyle{hidden2 neuron}=[neuron, fill=white!50];
    \tikzstyle{annot} = [text width=4em, text centered]

    \foreach \name / \y in {1,...,3}
    \pgfmathsetmacro{\ysc}{\y*\scalermlp}
        \node[input neuron] (I-\name) at (0,-\ysc) {};

    \foreach \name / \y in {1,...,\numhidden}
    \pgfmathsetmacro{\ysc}{\y*\scalermlp}
        \path[yshift=3.3cm]
            node[hidden1 neuron] (H1-\name) at (\layersep,-\ysc cm) {};
	
	\foreach \name / \y in {1,...,\numhidden}
	\pgfmathsetmacro{\ysc}{\y*\scalermlp}
	\path[yshift=3.3cm]
	node[hidden2 neuron] (H2-\name) at (2*\layersep,-\ysc cm) {};
	
    \node[output neuron, right of=H2-6] (O) {};

    \foreach \source in {1,...,3}
        \foreach \dest in {1,...,\numhidden}
            \path (I-\source) edge (H1-\dest);
	
	\foreach \source in {1,...,\numhidden}
		\foreach \dest in {1,...,\numhidden}
			\path (H1-\source) edge (H2-\dest);
	
    \foreach \source in {1,...,\numhidden}
        \path (H2-\source) edge (O);

    \node[annot,above of=H1-1, node distance=1cm] (hl1) {Hidden Layer};
    \node[annot,above of=H2-1, node distance=1cm] (hl2) {Hidden Layer};
    \node[annot,left of=hl1] {Input Layer};
    \node[annot,right of=hl2] {Output Layer};
\end{tikzpicture}
\caption{Two-layer MLP with a three-dimensional input space ($N_0=3$) and an one-dimensional output space ($N_3=1$). The hidden dimensions of layer one and two are eleven ($N_1=N_2=11$).}
\label{fig:nn_tikz}
\end{figure}
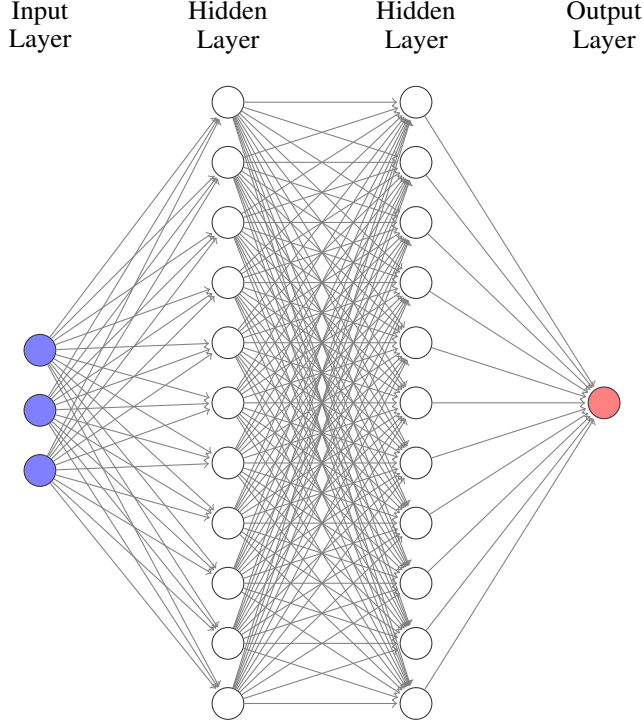
A well-known result that justifies the high applicability of MLPs is the so-called \textit{universal approximation theorem} for one-layer perceptrons with output dimension $d_1 = 1$, see for instance Theorem 1 and Theorem 2 in \cite{hornik} or Theorem 4.2 in \cite{deep_hedging}. Let us point out that the universal approximation theorem easily carries over to the case of MLPs with output dimension $d_1 > 1$ and more than one hidden layer, which corresponds to the situation under consideration in the present paper.

\subsection{Temporal convolutional networks}
As a result of volatility clusters being present in the market, the log return process is often decomposed into a stochastic volatility process and an innovations process. In order to model stochastic volatility, we propose the use of TCNs.

TCNs are convolutional architectures, which have recently shown to be competitive on many sequence-related modeling tasks \cite{bai_empirical}. In particular, empirical results suggest that TCNs are able to capture long-range dependencies in sequences more effectively than well-known recurrent architectures \parencite[Chapter 10]{deeplearningbook} such as the LSTM \parencite{lstm_original} or the GRU \parencite{gru}. One of the main advantages of TCNs compared to recurrent NNs is the absence of exponentially vanishing and exploding gradients through time \parencite{difficulty_of_training_rnns}, which is one of the main issues why RNNs are difficult to optimize. Although LSTMs address this issue by using gated activations, empirical studies show that TCNs perform better on supervised learning benchmarks \parencite{bai_empirical}. 

The construction of TCNs is simple. The crucial ingredient are so-called \textit{dilated causal convolutions}. Causal convolutions are convolutions, where the output only depends on past sequence elements. Dilated convolutions, also referred to as \textit{atrous} convolutions, are convolutions ``with holes''. \autoref{fig:tcn_no_dilation} illustrates a {\it Vanilla TCN} (cf. \autoref{def:vanilla_tcn}) with four hidden layers, a kernel size of two ($K=2$) and a dilation factor of one ($D=1$). \autoref{fig:tcn_dilation} depicts a TCN with a kernel size and a dilation factor equal to two ($K=D=2$). Note that as $D=2$, the dilation increases by a factor of two in every layer. Comparing both networks it becomes clear that the use of increasing dilations in each layer is essential to capture and model long-range dependencies. 

Below we define the TCN as well as related concepts formally. We first give the definition of the dilated causal convolutional operator, which is the basic building block of the {\it convolutional layer}. Composing several convolutional layers (together with activation functions) gives the Vanilla TCN. For the rest of this section, let $N_I, N_O, K, D, T \in \N$.

\begin{definition}[$*_D$ operator]
	\label{def:causal_convolution}
	Let $X \in \R^{N_I \times T}$ be an $N_I$-variate sequence of length $T $ and $W \in \R^{K\times N_I \times N_O}$ a tensor. Then for $t\in\{D(K-1)+1, \dots, T\}$ and $m \in \{1\dots, N_O\}$ the operator $*_D$, defined by
	\[
	\normalbrack{W *_D X}_{m,t} \coloneqq \sum_{i=1}^{K} \sum_{j =1}^{N_I} W_{i,j,m} \cdot X_{j,t-D(K-i)} ~,
	\]
	is called \textit{dilated causal convolutional operator} with \textit{dilation} $D$ and kernel size $K$.
\end{definition}

A visualization of the operator for different dilations and kernel sizes is given in \autoref{tikz:colvolution_operator}. For $K=D=1$ (see \autoref{tikz:colvolution_operator_a}), each element of the output sequence only depends on the element of the input sequence at the same time step. In the case of $K=2$ and $D=1$, each element of the output sequence originates from the elements of the input sequence at the same and previous time step. In the case of $K=D=2$, the distance between the elements that pass on the information is two.
\def\layersep{1cm}
\def\scaler{0.5}
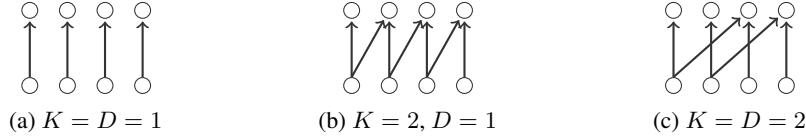
\begin{figure}[H]
	\begin{subfigure}{0.3\textwidth}
	\centering
	\begin{tikzpicture}[shorten >=1pt,->,draw=black!80, node distance=\layersep]
	\tikzstyle{every pin edge}=[<-,shorten <=1pt]
	\tikzstyle{neuron}=[circle,fill=black!25,minimum size=6pt,inner sep=0pt, draw=black!80]
	\tikzstyle{neuronI}=[neuron, fill=white!50];
	\tikzstyle{neuronII}=[neuron, fill=white!50];
	\tikzstyle{annot} = [text width=4em, text centered]
	
	\foreach \name / \y in {1,...,4}
	\pgfmathsetmacro{\ysc}{\y*\scaler}
	\node[neuronI] (I-\name) at (-\ysc, 0) {};
	
	\foreach \name / \y in {1,...,4}
	\pgfmathsetmacro{\ysc}{\y*\scaler}
	\path node[neuronII] (II-\name) at (-\ysc cm, \layersep) {};
	
	\path (I-1.north) edge [thick] (II-1);
	\path (I-2.north) edge [thick] (II-2);
	\path (I-3.north) edge [thick] (II-3);
	\path (I-4.north) edge [thick] (II-4);
	\end{tikzpicture}
	\caption{$K=D=1$}
	\label{tikz:colvolution_operator_a}
\end{subfigure}
\begin{subfigure}{0.3\textwidth}
	\centering
	\begin{tikzpicture}[shorten >=1pt,->,draw=black!80, node distance=\layersep]
	\tikzstyle{every pin edge}=[<-,shorten <=1pt]
	\tikzstyle{neuron}=[circle,fill=black!25,minimum size=6pt,inner sep=0pt, draw=black!80]
	\tikzstyle{neuronI}=[neuron, fill=white!50];
	\tikzstyle{neuronII}=[neuron, fill=white!50];
	\tikzstyle{annot} = [text width=4em, text centered]
	
	\foreach \name / \y in {1,...,4}
	\pgfmathsetmacro{\ysc}{\y*\scaler}
	\node[neuronI] (I-\name) at (-\ysc, 0) {};
	
	\foreach \name / \y in {1,...,4}
	\pgfmathsetmacro{\ysc}{\y*\scaler}
	\path node[neuronII] (II-\name) at (-\ysc cm, \layersep) {};
	
	\path (I-1.north) edge [thick] (II-1);
	\path (I-2.north) edge [thick] (II-1);
	\path (I-2.north) edge [thick] (II-2);
	\path (I-3.north) edge [thick] (II-2);
	\path (I-3.north) edge [thick] (II-3);
	\path (I-4.north) edge [thick] (II-3);
	\path (I-4.north) edge [thick] (II-4);
	\end{tikzpicture}
	\caption{$K=2$, $D=1$}
\end{subfigure}
\begin{subfigure}{0.3\textwidth}
	\centering
	\begin{tikzpicture}[shorten >=1pt,->,draw=black!80, node distance=\layersep]
	\tikzstyle{every pin edge}=[<-,shorten <=1pt]
	\tikzstyle{neuron}=[circle,fill=black!25,minimum size=6pt,inner sep=0pt, draw=black!80]
	\tikzstyle{neuronI}=[neuron, fill=white!50];
	\tikzstyle{neuronII}=[neuron, fill=white!50];
	\tikzstyle{annot} = [text width=4em, text centered]
	
	\foreach \name / \y in {1,...,4}
	\pgfmathsetmacro{\ysc}{\y*\scaler}
	\node[neuronI] (I-\name) at (-\ysc, 0) {};
	
	\foreach \name / \y in {1,...,4}
	\pgfmathsetmacro{\ysc}{\y*\scaler}
	\path node[neuronII] (II-\name) at (-\ysc cm, \layersep) {};
	
	\path (I-1.north) edge [thick] (II-1);
	\path (I-2.north) edge [thick] (II-2);
	\path (I-3.north) edge [thick] (II-1);
	\path (I-3.north) edge [thick] (II-3);
	\path (I-4.north) edge [thick] (II-2);
	\path (I-4.north) edge [thick] (II-4);
	\end{tikzpicture}
	\caption{$K=D=2$}
	\end{subfigure}
	\caption{Dilated causal convolutional operator for different dilations $D$ and kernel sizes $K$}
	\label{tikz:colvolution_operator}
\end{figure}

\begin{definition}[Causal convolutional layer]
	\label{def:ccl}
	Let $W$ be as in \autoref{def:causal_convolution} and $b \in \R^{N_O}$. A function 
	\begin{align*}
	\ccl: \R^{N_I \times T} &\to \R^{N_O \times (T-D(K-1))}
	\end{align*}
	defined for $ t\in\{D(K-1)+1, \dots, T\}$ and $m \in \{1,...,N_O\}$ by 
	\[ \ccl(X)_{m,t} := \normalbrack{W *_D X}_{m,t} + b_m \]
	is called \textit{causal convolutional layer with dilation $D$}.
\end{definition}
\begin{remark}
	The quadruple $(N_I, N_O, K, D)$ will be called the \textit{arguments} of a causal convolution $w$ and represent the \textit{input dimension}, \textit{output dimension}, \textit{kernel size}, and \textit{dilation}, respectively.
\end{remark}

\def\layersep{1cm}
\def\scaler{0.5}
\begin{figure}[H]
	\centering
	\begin{tikzpicture}[shorten >=1pt,->,draw=black!80, node distance=\layersep]
	\tikzstyle{every pin edge}=[<-,shorten <=1pt]
	\tikzstyle{neuron}=[circle,fill=black!25,minimum size=6pt,inner sep=0pt, draw=black!80]
	\tikzstyle{input neuron}=[neuron, fill=blue!50];
	\tikzstyle{output neuron}=[neuron, fill=red!50];
	\tikzstyle{hidden1 neuron}=[neuron, fill=white!50];
	\tikzstyle{hidden2 neuron}=[neuron, fill=white!50];
	\tikzstyle{annot} = [text width=4em, text centered]
	
	\foreach \name / \y in {1,...,16}
	\pgfmathsetmacro{\ysc}{\y*\scaler}
	\node[input neuron] (I-\name) at (-\ysc, 0) {};
	
	\foreach \name / \y in {1,...,16}
		\pgfmathsetmacro{\ysc}{\y*\scaler}
	\path node[hidden1 neuron] (H1-\name) at (-\ysc cm, \layersep) {};
	
	\foreach \name / \y in {1,...,16}
	\pgfmathsetmacro{\ysc}{\y*\scaler}
	\path node[hidden2 neuron] (H2-\name) at (-\ysc cm, 2*\layersep) {};
	
	\foreach \name / \y in {1,...,16}
	\pgfmathsetmacro{\ysc}{\y*\scaler}
	\path node[hidden2 neuron] (H3-\name) at (-\ysc cm, 3*\layersep) {};
	
	\foreach \name / \y in {1,...,16}
	\pgfmathsetmacro{\ysc}{\y*\scaler}
	\path node[hidden2 neuron] (H4-\name) at (-\ysc cm, 4*\layersep) {};
	
	\foreach \name / \y in {1,...,16}
	\pgfmathsetmacro{\ysc}{\y*\scaler}
	\path node[output neuron](O-\name) at (-\ysc cm, 5*\layersep) {};
	
%
%
%
%
%
	
	\path (H4-1.north) edge [thick] (O-1);
	
	\path (H3-1.north) edge [thick] (H4-1);
	\path (H3-2.north) edge [thick] (H4-1);
	
	\path (H2-1.north) edge [thick] (H3-1);
	\path (H2-2.north) edge [thick] (H3-1);
	\path (H2-2.north) edge [thick] (H3-2);
	\path (H2-3.north) edge [thick] (H3-2);
	
	\path (H1-1.north) edge [thick] (H2-1);
	\path (H1-2.north) edge [thick] (H2-1);
	\path (H1-2.north) edge [thick] (H2-2);
	\path (H1-3.north) edge [thick] (H2-2);
	\path (H1-3.north) edge [thick] (H2-3);
	\path (H1-4.north) edge [thick] (H2-3);
	
	\path (I-1.north) edge [thick] (H1-1);
	\path (I-2.north) edge [thick] (H1-1);
	\path (I-2.north) edge [thick] (H1-2);
	\path (I-3.north) edge [thick] (H1-2);
	\path (I-3.north) edge [thick] (H1-3);
	\path (I-4.north) edge [thick] (H1-3);
	\path (I-4.north) edge [thick] (H1-4);
	\path (I-5.north) edge [thick] (H1-4);
	

	\node[annot,right of=I-1, node distance=1.4cm, text width=3cm] (o1) {\footnotesize Input Layer};
	\node[annot,right of=H1-1, node distance=1.4cm, text width=3cm] (hl1) {\footnotesize Hidden Layer};
	\node[annot,right of=H2-1, node distance=1.4cm, text width=3cm] (hl2) {\footnotesize Hidden Layer};
	\node[annot,right of=H3-1, node distance=1.4cm, text width=3cm] (hl3) {\footnotesize Hidden Layer};
	\node[annot,right of=H4-1, node distance=1.4cm, text width=3cm] (hl4) {\footnotesize Hidden Layer};
	\node[annot,right of=O-1, node distance=1.4cm, text width=3cm] (o1) {\footnotesize Output Layer};
	\end{tikzpicture}
	\caption{Vanilla TCN with $4$ hidden layers, kernel size $K=2$ and dilation factor $D=1$ (cf. \cite{wavenet}).}
	\label{fig:tcn_no_dilation}
\end{figure}
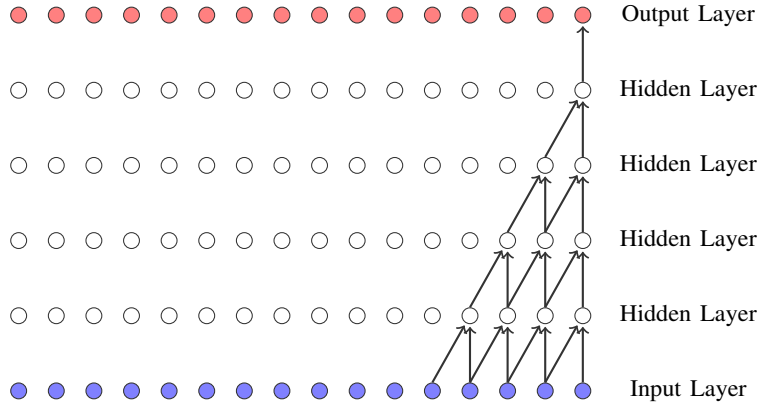

\begin{example}[$1 \times 1$ convolutional layer]
	\label{ex:1_times_1_conv}
	Let $X \in \R^{N_I \times T}$ be an $N_I$-variate sequence and $\ccl: \R^{N_I \times T} \to \R^{N_O \times T}$ a causal convolutional layer with arguments $(N_I, N_O, 1, 1)$. We call such a layer a $1\times1$ \textit{convolutional layer}.\footnote{Note that using a $1\times 1$ convolution is equivalent to applying an affine transformation along the time dimension of $X$.}
\end{example}

In the previous section we constructed MLPs by composing affine transformations with activation functions. The Vanilla TCN construction follows a similar approach: dilated causal convolutional layers are composed with activation functions. In order to allow for more expressive transformations, the TCN uses a \textit{block module} construction and thereby generalizes the Vanilla TCN. For completeness, both definitions are given below.

\begin{definition}[Block module]
	Let $S \in \N$. A function $\psi: \R^{N_I \times T} \to \R^{N_O \times (T-S)}$ that is Lipschitz continuous is called \textit{block module} with arguments $(N_I, N_O, S)$. 
\end{definition}

\begin{definition}[Temporal convolutional network]
	\label{def:tcn}
	\sloppy
	Let $T_0, L, N_0, \dots, N_{L+1} \in \N $.
	Moreover, for $l \in \{1,\dots, L\}$ let $S_l \in \N$ such that $\sum_{l=1}^{L} S_l \leq T_0 -1$. Hence, for $T_l \coloneqq T_{l-1} - S_l$ it holds
	\begin{equation}
	\label{eq:t_L_rfs}
	T_L = T_0 - \sum_{l=1}^{L} S_l \geq 1 ~.
	\end{equation}
	Furthermore, let $\psi_{l}: \R^{N_{l-1} \times T_{l-1}} \to \R^{N_{l}  \times T_{l}}$ for $l \in \{1,\dots,L\}$ represent block modules and ${\ccl:\R^{N_{L} \times T_L} \to \R^{N_{L+1} \times T_L}}$ a $1\times 1$ convolutional layer.
	A function $f: \R^{N_0 \times T_0} \times \Theta \to \R^{ N_{L+1} \times T_{L} } $, defined by 
	\begin{equation*}
	f(X, \theta) = w \circ \psi_{L} \circ \cdots \circ \psi_{1}(X) ~,
	\end{equation*}
	is called {\it temporal convolutional network} with $L$ hidden layers. The class of TCNs with $L$ hidden layers mapping from $\R^{d_0}$ to $\R^{d_1}$ will be denoted by $\tcn_{d_0, d_1, L}$ ($d_0 = N_0, d_1 = N_{L+1}$).
\end{definition}

\begin{definition}[Vanilla TCN]
	\label{def:vanilla_tcn}
	Let $f \in \tcn_{N_0, N_{L+1}, L}$ such that for all $l \in \{1,\dots, L\}$ each block module $\psi_l$ is defined as a composition of a causal convolutional layer $\ccl_l$ with arguments $(N_{l-1}, N_l, K_l, D_l)$ and an activation function $\phi$, i.e. $\psi_l = \phi \circ \ccl_l$. Then we call $f: \R^{N_0 \times T_0} \times \Theta \to \R^{N_{L+1} \times T_{L}} $ a \textit{Vanilla TCN}. Moreover, if $D_l = D^{l-1}$ for all $l \in \{1,\dots, L\}$, we call $f$ a \textit{Vanilla TCN with dilation factor $D$}. Whenever $ K_l = K $ for all $l \in \{1,\dots, L\}$, we say that $f$ has \textit{kernel size $K$}.
\end{definition}

\begin{figure}[h]
	\centering
	\begin{tikzpicture}[shorten >=1pt,->,draw=black!80, node distance=\layersep]
	\tikzstyle{every pin edge}=[<-,shorten <=1pt]
	\tikzstyle{neuron}=[circle,fill=black!25,minimum size=6pt,inner sep=0pt, draw=black!80]
	\tikzstyle{input neuron}=[neuron, fill=blue!50];
	\tikzstyle{output neuron}=[neuron, fill=red!50];
	\tikzstyle{hidden1 neuron}=[neuron, fill=white!50];
	\tikzstyle{hidden2 neuron}=[neuron, fill=white!50];
	\tikzstyle{annot} = [text width=4em, text centered]
	
	\foreach \name / \y in {1,...,16}
	\pgfmathsetmacro{\ysc}{\y*\scaler}
	\node[input neuron] (I-\name) at (-\ysc, 0) {};
	
	\foreach \name / \y in {1,...,16}
	\pgfmathsetmacro{\ysc}{\y*\scaler}
	\path node[hidden1 neuron] (H1-\name) at (-\ysc cm, \layersep) {};
	
	\foreach \name / \y in {1,...,16}
	\pgfmathsetmacro{\ysc}{\y*\scaler}
	\path node[hidden2 neuron] (H2-\name) at (-\ysc cm, 2*\layersep) {};
	
	\foreach \name / \y in {1,...,16}
	\pgfmathsetmacro{\ysc}{\y*\scaler}
	\path node[hidden2 neuron] (H3-\name) at (-\ysc cm, 3*\layersep) {};
	
	\foreach \name / \y in {1,...,16}
	\pgfmathsetmacro{\ysc}{\y*\scaler}
	\path node[hidden2 neuron] (H4-\name) at (-\ysc cm, 4*\layersep) {};
	
	\foreach \name / \y in {1,...,16}
	\pgfmathsetmacro{\ysc}{\y*\scaler}
	\path node[output neuron](O-\name) at (-\ysc cm, 5*\layersep) {};
	
%
%
%
%
%
	
	\path (H4-1.north) edge [thick] (O-1);
	
	\path (H3-1.north) edge [thick] (H4-1);
	\path (H3-9.north) edge [thick] (H4-1);
	
	\path (H2-1.north) edge [thick] (H3-1);
	\path (H2-5.north) edge [thick] (H3-1);
	\path (H2-9.north) edge [thick] (H3-9);
	\path (H2-13.north) edge [thick] (H3-9);
	
	\path (H1-1.north) edge [thick] (H2-1);
	\path (H1-3.north) edge [thick] (H2-1);
	\path (H1-5.north) edge [thick] (H2-5);
	\path (H1-7.north) edge [thick] (H2-5);
	\path (H1-9.north) edge [thick] (H2-9);
	\path (H1-11.north) edge [thick] (H2-9);
	\path (H1-13.north) edge [thick] (H2-13);
	\path (H1-15.north) edge [thick] (H2-13);
	
	\path (I-1.north) edge [thick] (H1-1);
	\path (I-2.north) edge [thick] (H1-1);
	\path (I-3.north) edge [thick] (H1-3);
	\path (I-4.north) edge [thick] (H1-3);
	\path (I-5.north) edge [thick] (H1-5);
	\path (I-6.north) edge [thick] (H1-5);
	\path (I-7.north) edge [thick] (H1-7);
	\path (I-8.north) edge [thick] (H1-7);
	\path (I-9.north) edge [thick] (H1-9);
	\path (I-10.north) edge [thick] (H1-9);
	\path (I-11.north) edge [thick] (H1-11);
	\path (I-12.north) edge [thick] (H1-11);
	\path (I-13.north) edge [thick] (H1-13);
	\path (I-14.north) edge [thick] (H1-13);
	\path (I-15.north) edge [thick] (H1-15);
	\path (I-16.north) edge [thick] (H1-15);
	
	\node[annot,right of=I-1, node distance=1.4cm, text width=3cm] (o1) {\footnotesize Input Layer};
	\node[annot,right of=H1-1, node distance=1.4cm, text width=3cm] (hl1) {\footnotesize Hidden Layer};
	\node[annot,right of=H2-1, node distance=1.4cm, text width=3cm] (hl2) {\footnotesize Hidden Layer};
	\node[annot,right of=H3-1, node distance=1.4cm, text width=3cm] (hl3) {\footnotesize Hidden Layer};
	\node[annot,right of=H4-1, node distance=1.4cm, text width=3cm] (hl4) {\footnotesize Hidden Layer};
	\node[annot,right of=O-1, node distance=1.4cm, text width=3cm] (o1) {\footnotesize Output Layer};
	\end{tikzpicture}
	\caption{Vanilla TCN with $4$ hidden layers, kernel size $K=2$ and dilation factor $D=2$ (cf. \cite{wavenet}).}
	\label{fig:tcn_dilation}
\end{figure}
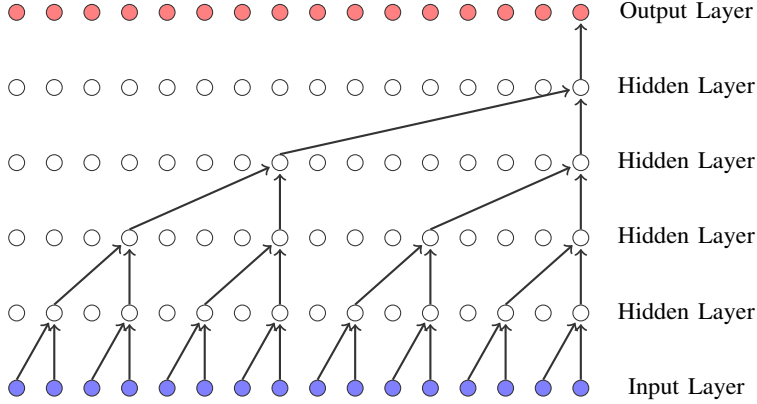

TCN's ability to model long-range dependencies becomes ultimately apparent when comparing the two Vanilla TCNs displayed in \autoref{fig:tcn_no_dilation} and \autoref{fig:tcn_dilation}. In \autoref{fig:tcn_no_dilation}, the network is a function of $5$ sequence elements, whereas the network in \autoref{fig:tcn_dilation} has $16$ sequence elements as input. We call the number of sequence elements that the TCN can capture the \textit{receptive field size} and give a formal definition below:
\begin{definition}[Receptive field size]
	Let $f \in \tcn_{d_0, d_1, L}$ and let $S_1, \dots, S_L$ be as in \autoref{def:tcn}. The constant
	\begin{equation*}
	\rfs \coloneqq 1 + \sum_{l=1}^{L} S_l
	\end{equation*}
	is called \textit{receptive field size} (\textit{RFS}). 
\end{definition}
\begin{remark}
	For Vanilla TCNs with kernel size $K$ and dilation factor $D>1$, the RFS $\rfs$ can easily be computed using the formula for the sum of a geometric sequence with finite length:
	\begin{equation*}
		\rfs =1+ (K-1) \cdot \normalbrack{\dfrac{D^{L} - 1}{D - 1}} ~.
	\end{equation*}
	Therefore, the RFS $\rfs$ is the minimum initial time dimension $T_0$ of an input $X \in \R^{N_0 \times T_0}$ such that the sequence $X$ can be inferred (compare \autoref{eq:t_L_rfs}).
\end{remark}
\begin{remark}
	Note that an MLP can be seen as a Vanilla TCN in which each causal convolution is a $1\times 1$ convolution. Thus, MLPs are a subclass of TCNs with an RFS equal to one. 
\end{remark}

The idea of residual connections can also be used.

\begin{definition}[TCN with skip connections]
\label{def:tcn_skip}
Assume the notation from \autoref{def:tcn} and for $N_{skip} \in \N$ let 
\[
\gamma_l : \R^{N_{l-l} \times T_{l-1}} \to \R^{N_{l} \times T_{l}} \times \R^{N_{skip} \times T_L } \quad \textrm{for} \ l \in \{1,\dots, L\}
\]
denote block modules. Moreover, let $\gamma$ be a block module with arguments $(N_{skip}, N_{L+1}, 0)$. 
If the output $Y \in \R^{N_{L+1} \times T_{L}}$ of a TCN $f:\R^{N_0 \times T_0} \times \Theta \to \R^{N_{L+1} \times T_{L}} $ is defined recursively by
\begin{align*}
\left( X^{(l)}, H^{(l)} \right) &= \gamma_l\left(X^{(l-1)}\right) \quad \textrm{for} \ l \in \{1,\dots, L\}\\
Y &= \gamma\left(\sum_{l=1}^L H^{(l)}\right)  ~,
\end{align*}
where $X^{(0)} \in \R^{N_0 \times T_0}$, then $f$ is called a \textit{temporal convolutional network with skip connections}.
\end{definition}

One of the downsides of TCNs is that the length of time series to be processed is restricted to the TCN's RFS. Hence, in order to model long-range dependencies we require a RFS $\rfs \gg 1$ leading to computational bottlenecks. Furthermore, it becomes questionable if such large networks can be trained to model something meaningful and if sufficient data is available. Although interesting extensions to TCNs to model long-range dependencies (cf. \cite{Dielemann2018}), we leave it as future work to develop these methods. 

\begin{figure}[H]
	\centering
	\begin{tikzpicture}[shorten >=1pt,->,draw=black!80, node distance=\layersep]
	\tikzstyle{every pin edge}=[<-,shorten <=1pt]
	\tikzstyle{neuron}=[circle,fill=black!25,minimum size=6pt,inner sep=0pt, draw=black!80]
	\tikzstyle{operation}=[circle,fill=black!25,minimum size=12pt,inner sep=0pt, draw=black!80]
	\tikzstyle{input neuron}=[neuron, fill=blue!50];
	\tikzstyle{output neuron}=[neuron, fill=red!50];
	\tikzstyle{hidden1 neuron}=[neuron, fill=white!50];
	\tikzstyle{hidden2 neuron}=[neuron, fill=white!50];
	\tikzstyle{annot} = [text width=4em, text centered]
	
	
	\foreach \name / \y in {1,...,16}
	\pgfmathsetmacro{\ysc}{\y*\scaler}
	\path node[input neuron] (H0-\name) at (-\ysc cm, 0*\layersep) {};
	
	\foreach \name / \y in {1,...,16} 
	\pgfmathsetmacro{\ysc}{\y*\scaler}
	\path node[hidden1 neuron] (H1-\name) at (-\ysc cm, \layersep) {};
	
	\foreach \name / \y in {1,...,16} 
	\pgfmathsetmacro{\ysc}{\y*\scaler}
	\path node[hidden2 neuron] (H2-\name) at (-\ysc cm, 2*\layersep) {};

	\foreach \name / \y in {1,...,16} 
	\pgfmathsetmacro{\ysc}{\y*\scaler}
	\path node[hidden2 neuron] (H3-\name) at (-\ysc cm, 3*\layersep) {};
	
	\foreach \name / \y in {1,...,16} 
	\pgfmathsetmacro{\ysc}{\y*\scaler}
	\path node[hidden2 neuron, anchor=south](O-\name) at (-\ysc cm, 4*\layersep) {};

	\path (H3-1.north) edge [thick] (O-1);
	\path (H3-9.north) edge [thick] (O-1);
	
	\path (H2-1.north) edge [thick] (H3-1);
	\path (H2-5.north) edge [thick] (H3-1);
	\path (H2-9.north) edge [thick] (H3-9);
	\path (H2-13.north) edge [thick] (H3-9);
	
	\path (H1-1.north) edge [thick] (H2-1);
	\path (H1-3.north) edge [thick] (H2-1);
	\path (H1-5.north) edge [thick] (H2-5);
	\path (H1-7.north) edge [thick] (H2-5);
	\path (H1-9.north) edge [thick] (H2-9);
	\path (H1-11.north) edge [thick] (H2-9);
	\path (H1-13.north) edge [thick] (H2-13);
	\path (H1-15.north) edge [thick] (H2-13);
	
	\path (H0-1.north) edge [thick] (H1-1);
	\path (H0-2.north) edge [thick] (H1-1);
	\path (H0-3.north) edge [thick] (H1-3);
	\path (H0-4.north) edge [thick] (H1-3);
	\path (H0-5.north) edge [thick] (H1-5);
	\path (H0-6.north) edge [thick] (H1-5);
	\path (H0-7.north) edge [thick] (H1-7);
	\path (H0-8.north) edge [thick] (H1-7);
	\path (H0-9.north) edge [thick] (H1-9);
	\path (H0-10.north) edge [thick] (H1-9);
	\path (H0-11.north) edge [thick] (H1-11);
	\path (H0-12.north) edge [thick] (H1-11);
	\path (H0-13.north) edge [thick] (H1-13);
	\path (H0-14.north) edge [thick] (H1-13);
	\path (H0-15.north) edge [thick] (H1-15);
	\path (H0-16.north) edge [thick] (H1-15);
	

	\path node[operation] (Add) at (1 cm, 2*\layersep) {};
	\node[annot,right of=Add, node distance=0cm, text width=0.3cm] (Add) {\footnotesize +};
	
	\path node[output neuron] (O) at (2 cm, 2*\layersep) {};
	
	\path (H0-1.east) edge [thick] (Add);
	\path (H1-1.east) edge [thick] (Add);
	\path (H2-1.east) edge [thick] (Add);
	\path (H3-1.east) edge [thick] (Add);
	\path (O-1.east) edge [thick] (Add);
	
	\path (Add) edge [thick] (O);
	
	\end{tikzpicture}
	\caption{Vanilla TCN with skip connections.}
	\label{fig:tcn_skip}
\end{figure}
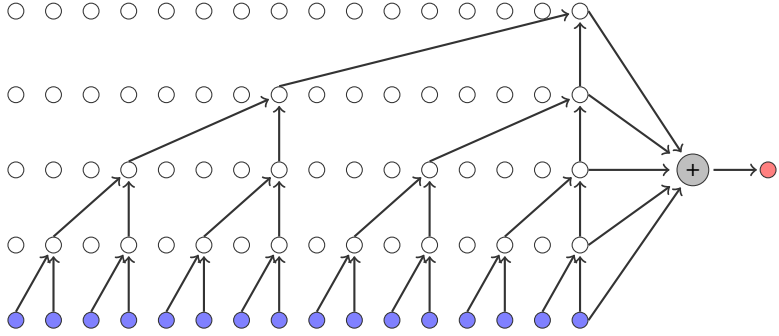

The NN topologies introduced in this section are the core components used to achieve the results presented in \hyperref[sec:numerical]{Section~\ref*{sec:numerical}}. In order to train these networks to generate time series, we now formulate GANs in the setting of random variables and stochastic processes.

\section{Generative adversarial networks}
\label{sec:gans}
Generative adversarial networks (GANs) \parencite{gans} are a relatively new class of algorithms to learn a generative model of a sample of a random variable, i.e. a dataset, or the distribution of a random variable itself. Originally, GANs were applied to generate images. In this section, we introduce GANs for random variables and then extend them to discrete-time stochastic processes with TCNs.

Throughout this section let $N_Z, N_X \in \N$ and let $(\Omega, \mathcal{F}, \Prob)$ be a probability space. Furthermore, assume that $X$ and $Z$ are $\data$ and $\latent$-valued random variables, respectively. The distribution of a random variable $X$ will be denoted by $\Prob_X$. 

\subsection{Formulation for random variables}
In the context of GANs, $(\latent, \B(\latent))$ and $(\data, \B(\data))$ are called the \textit{latent} and the \textit{data measure space}, respectively. The random variable $Z$ represents the \textit{noise prior} and $X$ the \textit{targeted (or data) random variable}. The goal of GANs is to train a network $g:\latent\times\Theta^{(g)} \to \data$ such that the induced random variable $g_\theta(Z) \coloneqq g_\theta \circ Z$ for some parameter $\theta \in \Theta^{(g)}$ and the targeted random variable $X$ have the same distribution,  i.e. $g_\theta(Z)\eqdist X$. To achieve this, \citeauthor{gans} proposed the adversarial modeling framework for NNs and introduced the \textit{generator} and the \textit{discriminator} as follows:

\begin{definition}[Generator]
	\label{def:generator}
	Let $g: \latent \times \Theta^{(g)} \to \data $ be a network with parameter space $\Theta^{(g)}$. The random variable $\tilde{X}$, defined by 
	\begin{align*}
	\tilde{X}: \Omega \times \Theta^{(g)} &\to \data\\
	(\omega, \theta)&\mapsto g_\theta(Z(\omega)) ~,
	\end{align*}
	is called the {\it generated random variable}. Furthermore, the network $g$ is called {\it generator} and $\tilde{X}_\theta$ the {\it generated random variable with parameter $\theta$}.\footnote{The subscript $\theta$ of $\tilde{X}_\theta$ represents the dependency with respect to the neural operator's parameters $\theta$.}
\end{definition}

\begin{definition}[Discriminator]
	\label{def:discriminator}
	Let $\tilde{d}: \data \times \Theta^{(d)} \to \R $ be a network with parameters $\eta \in \Theta^{(d)}$ and $\sigma:\R\to[0,1] : x \mapsto \frac{1}{1 + e^{-x}}$ be the sigmoid function. A function $d: \data \times \Theta^{(d)} \to [0,1]$ defined by $d:(x, \eta) \mapsto \sigma \circ \tilde{d}_{\eta}(x)$
	is called a {\it discriminator}.
\end{definition}

Throughout this section we assume the notation used in \autoref{def:generator} and \autoref{def:discriminator}. 

\begin{definition}[Sample]
	A collection $\{Y_i\}_{i=1}^M$ of $M$ independent copies of some random variable $Y$ is called \textit{$Y$-sample} of size $M$. The notation $\{y_i\}_{i=1}^M$ refers to a realisation $\{Y_i(\omega)\}_{i=1}^M$ for some $\omega\in\Omega$.
\end{definition}

In the adversarial modeling framework two agents, the {generator} and the {discriminator} (also referred to as the adversary), are contesting with each other in a game-theoretic zero-sum game. Roughly speaking, the generator aims at generating samples $\left\lbrace \tilde{x}_{\theta, i} \right\rbrace_{i=1}^M$ such that the discriminator can not distinguish whether the realizations were sampled from the target or the generator distribution. In other words, the discriminator $\discriminator_\eta: \data \rightarrow [0, 1] $ acts as a classifier that assigns to each sample $x\in \mathbb{R}^{N_X}$ a probability of being a realization of the target distribution.

The optimization of GANs is formulated in two steps. First, the discriminator's parameters $\eta\in\Theta^{(d)}$ are chosen to maximize the function $\mathcal L(\theta,\cdot)$, $\theta\in\Theta^{(g)}$, given by
\begin{align*}
\mathcal{L}(\theta, \eta)  &:= 
\Ex{\log(\discriminator_{\eta}(X))}
+ \Ex{\log(1-\discriminator_{\eta}(\generator_\theta(Z)))} \\
&= \Ex{\log(\discriminator_{\eta}(X))}
+ \Ex{\log(1-\discriminator_{\eta}(\tilde{X}_\theta))} ~.
\end{align*}
In this sense, the discriminator learns to distinguish real and generated data.
In the second step, the generator's parameters $\theta\in\Theta^{(g)}$ are trained to minimize the probability of generated samples being identified as such and not from the data distribution. In summary, we receive the min-max game
\begin{equation*}
\min_{\theta\in\Theta^{(g)}} \max_{\eta\in\Theta^{(d)}} \mathcal{L}(\theta, \eta) ~,
\end{equation*}
which we refer to as the \textit{GAN objective}.

\paragraph{Training}
The generator's and discriminator's parameters $(\theta, \eta)$ are trained by alternating the computation of their gradients $\nabla_\eta \mathcal{L}(\theta, \eta)$ and $\nabla_\theta \mathcal{L}(\theta, \eta)$ and updating their respective parameters. To get a close approximation of the optimal discriminator $\discriminator_{\eta^{*}}$ \parencite[Proposition~1]{gans} it is common to compute the discriminators gradient multiple times and ascent the parameters $\eta$. Algorithm \autoref{algo:gan_algo} describes the procedure in detail. 

For further information on Algorithm~\ref{algo:gan_algo}, we refer to Section~4 in \cite{gans}, where this algorithm was developed. In particular, regarding the convergence of Algorithm~\ref{algo:gan_algo} we refer to Section 4.2 in \cite{gans}.

\begin{algorithm}[ht]
	\textbf{INPUT:} generator $\generator$, discriminator $\discriminator$, sample size $M\in\N$, generator learning rate $\alpha_g$, discriminator learning rate $\alpha_d$, number of discriminator optimization steps $k$
	\\
	\textbf{OUTPUT:} parameters $(\theta, \eta)$
	\begin{algorithmic}
		\WHILE{not converged}
		\FOR{k steps}
		\STATE Let $\left\{\tilde{x}_{\theta, i}\right\}_{i=1}^M$ be a realisation of an $\tilde{X}_\theta$-sample of size $M$.
		\STATE Let $\left\{x_i\right\}_{i=1}^M$ be a realisation of an $X$-sample of size $M$.
		\STATE Compute and store the gradient 
		\begin{equation*}
		\Delta_\eta \gets \nabla_\eta \dfrac{1}{M} \sum_{i=1}^M \log(\discriminator(x_i)) + \log(1 - \discriminator(\tilde{x}_{\theta, i})) ~.
		\end{equation*}
		\STATE Ascent the discriminator's parameters: $\eta \gets \eta + \alpha_d\cdot\Delta_\eta ~.$
		
		\ENDFOR
		\STATE Let $\left\{\tilde{x}_{\theta, i}\right\}_{i=1}^M$ be a realisation of an $\tilde{X}_\theta$-sample of size $M$.
		\STATE Compute and store the gradient
		\begin{equation*}
		\Delta_\theta \gets \nabla_\theta \dfrac{1}{m} \sum_{i=1}^m \log(\discriminator(\tilde{x}_{\theta, i})) ~.
		\end{equation*}
		\STATE Descent the generator's parameters: $\theta \gets \theta - \alpha_g\cdot\Delta_\theta ~.$
		\ENDWHILE{}
	\end{algorithmic}
	\caption{GAN optimization.}
	\label{algo:gan_algo}
\end{algorithm} 

\subsection{Formulation for stochastic processes}
\label{sec:gans_sp}
We now consider the formulation of GANs in the context of stochastic process generation by using TCNs, as its properties are intriguing for this setting. The following notation is used for brevity. 

\begin{notation}
	 Consider a stochastic process $(X_t)_{t \in \mathbb{Z}}$ parametrized by some $\theta \in \Theta$. For $s, t \in \mathbb{Z}, \ s\leq t$, we write 
	\[
	X_{s:t, \theta} := \normalbrack{X_{s, \theta}, \dots, X_{t, \theta}}
	\]
	and for an $\omega$-realization
	\[
	X_{s:t, \theta}(\omega) := \normalbrack{X_{s, \theta}(\omega), \dots, X_{t, \theta}(\omega)} \in \R^{N_X \times (t-s+1)}.
	\]
\end{notation}

We can now introduce the concept of neural (stochastic) processes.
\begin{definition}[Neural process]
	\label{def:np}
	Let $(Z_t)_{t \in \mathbb{Z}}$ be an i.i.d. noise process with values in $\latent$ and $g: \R^{N_Z\times\rfsg} \times \Theta^{(g)} \to \data$ a TCN with RFS $\rfsg$ and parameters $\theta \in \Theta^{(g)}$. A stochastic process $\tilde{X}$, defined by
	\begin{align*}
	\tilde X: \Omega \times \mathbb{Z} \times \Theta^{(g)} &\to \data\\
	(\omega, t,  \theta) &\mapsto g_\theta(Z_{t-(\rfsg-1):t}(\omega))
	\end{align*}
	such that $\tilde X_{t,\theta}: \Omega \to \data$ is  a $\mathcal{F}-\B(\data)$-measurable mapping for all $t\in\mathbb{Z}$ and $\theta\in\Theta^{(g)}$, is called \textit{neural process} and will be denoted by $\tilde X_\theta := (\tilde X_{t, \theta})_{t \in \mathbb{Z}}$.
\end{definition}

In the context of GANs, the i.i.d. noise process $Z = (Z_t)_{t \in \mathbb{Z}}$ from \autoref{def:np} represents the noise prior. Throughout this paper we assume for simplicity that for all $t \in \mathbb{Z}$ the random variable $Z_t$ follows a multivariate standard normal distribution, i.e. $Z_t \sim \mathcal{N}(0, I)$. In particular, the neural process $\tilde X_\theta = (\tilde X_{t, \theta})_{t \in \mathbb{Z}}$ is obtained by inferring $Z$ through the TCN generator $g$.

In our GAN framework for stochastic processes, the discriminator is similarly represented by a TCN $d: \R^{N_X \times \rfsd} \times \Theta^{(d)} \to [0,1]$ with RFS $\rfsd$. With these modifications to the original GAN setting for random variables, the GAN objective for stochastic processes can be formulated as
\[ \min_{\theta\in\Theta^{(g)}} \max_{\eta\in\Theta^{(d)}} \mathcal{L}(\theta, \eta) ~, \]
where
\[ \mathcal{L}(\theta, \eta) := \Ex{\log(\discriminator_{\eta}\left(X_{1:\rfsd} \right))}
+ \Ex{\log(1-\discriminator_{\eta}(\tilde{X}_{1:\rfsd,\theta}))} \]
and $X_{1:\rfsd}$ and $\tilde{X}_{1:\rfsd, \theta}$ denote the real and the generated process, respectively. Hence, analogue to the GAN setting for random variables the discriminator is trained to distinguish real from generated sequences, whereas the generator aims at simulating sequences which the discriminator can not distinguish from the real ones.

In order to train the generator and discriminator we proceed in a similar fashion as in the case of random variables. We consider realisations $\lbrace\tilde{x}^{(i)}_{1:\rfsd, \theta}\rbrace_{i=1}^M $ and $\lbrace {x}^{(i)}_{1:\rfsd}\rbrace_{i=1}^M $ of samples of size $M$ of the generated neural process and of the target distribution, respectively. Each element of these samples is then inferred into the discriminator to generate a $[0,1]$-valued output (the classifications), which are then averaged sample-wise to give a Monte Carlo estimate of the discriminator loss function.

\section{The model}
\label{sec:svnns}
After we have introduced the main NN topologies in \hyperref[sec:neural_network_topologies]{Section~\ref*{sec:neural_network_topologies}} and defined GANs in the context of stochastic process generation via TCNs in \hyperref[sec:gans]{Section~\ref*{sec:gans}}, we now turn to the problem of generating financial time series. We start by defining the generator function of Quant GANs: the \textit{stochastic volatility neural network} (SVNN). The induced process of an SVNN will be called \emph{log return neural process} (log return NP). The remainder of this section is devoted to answering the following theoretical aspects of our model:
\begin{itemize}
	\item  \hyperref[sec:lp_prop]{Section~\ref*{sec:lp_prop}}: How heavy are the tails generated by a log return NP?
	\item \hyperref[sec:generating_fat_tails]{Section~\ref*{sec:generating_fat_tails}}: Can the log return NP be transformed to model heavy-tails and which assumptions are implied?
	\item \hyperref[sec:risk_neutral]{Section~\ref*{sec:risk_neutral}}: Can the risk-neutral distribution of log return NPs be derived?
	\item \hyperref[sec:constrained_log_NP]{Section~\ref*{sec:constrained_log_NP}}: Can log return NPs be seen as a natural extension of already existing time-series models? 
\end{itemize}

\subsection{Log return neural processes}
Log return NPs are inspired by the volatility-innovation decomposition of various stochastic volatility models used in practice \cite{garch, stylized_facts, heston, tankov2003financial}. The corresponding generator architecture, the stochastic volatility neural network, consists of a volatility and drift TCN and another network which models the innovations. \autoref{fig:svnn_tcn} illustrates the architecture in use.

\begin{definition}[Log return neural process]
	\label{def:svnn}
	Let $Z = (Z_t)_{t \in \mathbb{Z}}$ be $\Z$-valued i.i.d. Gaussian noise, $g^{(\textrm{TCN})}: \R^{N_Z\times \rfsg}\times\Theta^{(\textrm{TCN})} \to \mathbb{R}^{2N_X}$ a TCN with RFS $\rfsg$ and $g^{(\epsilon)}:\Z\times\Theta^{(\epsilon)} \to \X$ be a network. Furthermore, let $\alpha \in \Theta^{(\textrm{TCN})}$ and $\beta \in \Theta^{(\epsilon)}$ denote some parameters. 
	A stochastic process $R$, defined by
	\begin{align*}
		R: \Omega \times \mathbb{Z} \times \Theta^{(\textrm{TCN})} \times \Theta^{(\epsilon)} &\to \X \\
		(\omega, t, \alpha, \beta) &\mapsto  \sqbrack{\sigma_{t, \alpha} \odot \epsilon_{t, \beta} + \mu_{t, \alpha}}(\omega) ~,
	\end{align*}
	where  $\odot$ denotes the Hadamard product and
	\begin{align*}
	h_t & \coloneqq g_{\alpha}^{(\textrm{TCN})} \left({Z_{t-\rfsg:(t-1)}} \right) \\
	\sigma_{t, \alpha} &\coloneqq \left| h_{t, 1:N_X} \right| \\
	\mu_{t, \alpha} &\coloneqq h_{t, (N_X+1):2N_X} \\
	\epsilon_{t, \beta}&\coloneqq g_{\beta}^{(\epsilon)}(Z_{t}) ~,
	\label{eq:vola_np}
	\end{align*}
	is called \textit{log return neural process}.
	The generator architecture defining the log return NP is called \textit{stochastic volatility neural network (SVNN)}. 
	The NPs $\sigma_\alpha \coloneqq (\sigma_{t, \alpha})_{t \in \mathbb{Z}}$,  $\mu_\alpha \coloneqq (\mu_{t, \alpha})_{t \in \mathbb{Z}}$ and $\epsilon_\beta \coloneqq (\epsilon_{t, \beta})_{t \in \mathbb{Z}}$ are called \textit{volatility, drift} and \textit{innovation NP}, respectively.
\end{definition}

\begin{remark}
	For simplicity, we do not distinguish the different NPs' parameters below and just write $\theta$ instead of $(\alpha, \beta)$, and $\Theta$ instead of $\Theta^{(\textrm{TCN})} \times \Theta^{(\epsilon)}$.	
\end{remark}

\begin{figure}[h]
	\centering
	\def\layersep{0.8cm}
\def\scaler{1}
\def\shift{-0}
\begin{tikzpicture}[shorten >=1pt,->,draw=black!80, node distance=\layersep]
\tikzstyle{every pin edge}=[<-,shorten <=1pt]
\tikzstyle{neuron}=[circle,fill=black!25,minimum size=6pt,inner sep=0pt, draw=black!80]
\tikzstyle{input neuron}=[neuron, fill=blue!50];
\tikzstyle{output neuron}=[neuron, fill=red!50];
\tikzstyle{hidden1 neuron}=[neuron, fill=white!50];
\tikzstyle{hidden2 neuron}=[neuron, fill=white!50];
\tikzstyle{annot} = [text width=4em, text centered]

\foreach \name / \y in {1,...,8}
\pgfmathsetmacro{\ysc}{\y*\scaler}
\node[input neuron] (I-\name) at (-\ysc, 0) {};

\foreach \name / \y in {1,...,8}
\pgfmathsetmacro{\ysc}{\y*\scaler}
\path node[hidden1 neuron] (H1-\name) at (-\ysc cm, \layersep) {};

\foreach \name / \y in {1,...,8}
\pgfmathsetmacro{\ysc}{\y*\scaler}
\path node[hidden2 neuron] (H2-\name) at (-\ysc cm, 2*\layersep) {};

\foreach \name / \y in {1,...,8}
\pgfmathsetmacro{\ysc}{\y*\scaler}
\path node[hidden2 neuron] (H3-\name) at (-\ysc cm, 3*\layersep) {};

\foreach \name / \y in {1,...,8}
\pgfmathsetmacro{\ysc}{\y*\scaler}
\path node[output neuron](O-\name) at (-\ysc cm, 4*\layersep) {};

\node[annot,above of=O-1, node distance=0.5cm, text width=3cm] (o1) {\footnotesize $(\sigma_{t, \theta}, \mu_{t, \theta}) $};
\path (H3-1.north) edge [thick] (O-1);

\path (H2-1.north) edge [thick] (H3-1);
\path (H2-5.north) edge [thick] (H3-1);

\path (H1-1.north) edge [thick] (H2-1);
\path (H1-3.north) edge [thick] (H2-1);
\path (H1-5.north) edge [thick] (H2-5);
\path (H1-7.north) edge [thick] (H2-5);

\path (I-1.north) edge [thick] (H1-1);
\path (I-2.north) edge [thick] (H1-1);
\path (I-3.north) edge [thick] (H1-3);
\path (I-4.north) edge [thick] (H1-3);
\path (I-5.north) edge [thick] (H1-5);
\path (I-6.north) edge [thick] (H1-5);
\path (I-7.north) edge [thick] (H1-7);
\path (I-8.north) edge [thick] (H1-7);

\foreach \name / \y in {1,...,8}
\pgfmathsetmacro{\shiftt}{int(\y-\shift)}
\node[annot,below of=I-\name, node distance=0.5cm, text width=3cm] (o1) {\footnotesize $Z_{t-\shiftt}$};

\path node[input neuron] (I-0) at (1*\scaler cm, 0*\layersep) {};
\path node[hidden1 neuron] (H1-0) at (1*\scaler cm, \layersep) {};
\path node[hidden1 neuron] (H2-0) at (1*\scaler cm,2*\layersep) {};
\path node[hidden1 neuron] (H3-0) at (1*\scaler cm,3*\layersep) {};
\path node[output neuron](O-0) at (1*\scaler cm, 4*\layersep) {};
\node[annot,above of=O-0, node distance=0.5cm, text width=3cm] (o1) {\footnotesize $\epsilon_{t, \theta}$};
\node[annot,below of=I-0, node distance=0.5cm, text width=3cm] (o1) {\footnotesize $Z_{t}$};

\path node[anchor=south] (label) at (-0.5*\scaler cm, 3*\layersep) {};

\path node[anchor=south] (label2) at (.8*\scaler cm, 3*\layersep) {};

\path (I-0.north) edge [thick] (H1-0);
\path (H1-0.north) edge [thick] (H2-0);
\path (H2-0.north) edge [thick] (H3-0);
\path (H3-0.north) edge [thick] (O-0);

\end{tikzpicture}
	\caption{Structure of the SVNN architecture. The volatility and drift component are generated by inferring the latent process $Z_{t-8:t-1}$ through the TCN, whereas the innovation is generated by inferring $Z_t$.}
	\label{fig:svnn_tcn}
\end{figure}
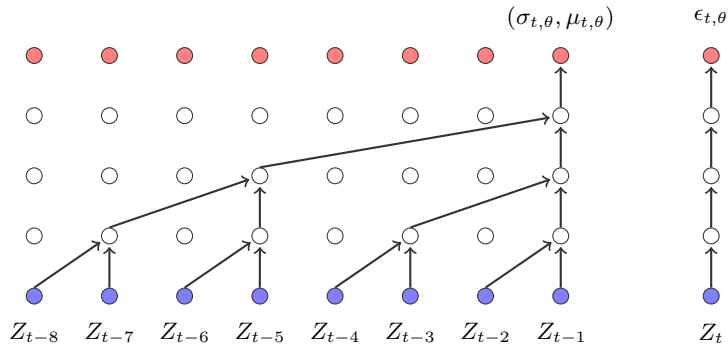

\begin{remark}[Independence of the SVNN-NPs]
	\label{rem:independence}
	Denote by $(\mathcal{F}_t^Z)_{t \in \mathbb{Z}}$ the natural filtration of the latent process $Z$ and let $t \in \mathbb{Z}$. By construction, $\sigma_{t, \theta}$ and $\mu_{t, \theta}$ are $\mathcal{F}_{t-1}^Z$-measurable. $\epsilon_{t, \theta}$ and $R_{t, \theta}$ are $\mathcal{F}_t^Z$-measurable. Observe further that the random variables $(\sigma_{t, \theta}, \mu_{t,\theta})$ and $\epsilon_{t, \theta}$ are independent, since $Z$ is an i.i.d. Gaussian noise process. As it turns out, the proposed construction is convenient when deriving the transition to the risk-neutral distribution in \hyperref[sec:risk_neutral]{Section~\ref*{sec:risk_neutral}}.
\end{remark}

\subsection{\boldmath$L^p$-space characterization of \boldmath$R_\theta$}
\label{sec:lp_prop}
We inspect next whether log return NPs exhibit heavy-tails. We first prove a result concerning generative networks in general and then conclude a corollary for the log return NP. 
\begin{theorem}[$L^p$-characterization of neural networks]
	\label{prop:lp_space_property_of_nn}
	Let $p \in \mathbb{N}$, $Z \in L^p(\Z)$ and ${g:\Z \times \Theta \to \X}$ a network with parameters $\theta \in \Theta$. Then, $g_\theta(Z) \in L^p(\X)$. 
\end{theorem}
\begin{proof}
	Observe that for any Lipschitz continuous function $f:\R^n \to \R^m$ there exists a suitable constant $L>0$ such that
	\begin{equation}
	\label{eq:lip_prop}
	\norm{f(x) - f(0)} \leq L \norm{x} \Rightarrow \norm{f(x)} \leq L \norm{x} + \norm{f(0)}
	\end{equation}
	as $\norm{x}-\norm{y}\leq \norm{x-y}$ for $x,y\in\R^n$. Now, using the Lipschitz property of neural networks (cf. \autoref{rem:nn_lip_prop}), we can apply \autoref{eq:lip_prop} and as $Z $ is an element of the space $L^p(\Z)$ we obtain
	\begin{align*}
	\E \sqbrack{\norm{g_\theta(Z)}^p} &\leq \E \sqbrack{\normalbrack{L \norm{Z} + \norm{g_\theta(\mathbf 0)}}^p}\\
	&= \sum_{k=0}^{p} {p \choose k} \  L^k \E\sqbrack{ \norm{Z}^{k}}  \ \norm{g_\theta(\mathbf 0)}^{p-k}\\
	&<\infty ~,
	\end{align*}
	where $L$ is the networks Lipschitz constant and $\mathbf{0} \in \Z$ the zero vector. This proves the statement. 
\end{proof}
Since we assume that our latent process $Z$ is Gaussian i.i.d. noise and thus square-integrable, we conclude from \autoref{prop:lp_space_property_of_nn} that the mean and the variance of the volatility, drift and innovation NP are finite. Additionally, these properties carry over to the log return NP as the following corollary proves.

\begin{corollary}
	\label{cor:svnp_bounded}
	Let $R_\theta$ be a log return NP parametrized by some $\theta \in \Theta$. Then, for all $t \in \mathbb{Z}$ and $p \in \mathbb{N}$ the random variable $R_{t, \theta}$ is an element of the space $L^p(\X)$. 
\end{corollary}
\begin{proof}
	The latent process $Z$ is Gaussian i.i.d. noise. Hence, \autoref{prop:lp_space_property_of_nn} yields  ${\sigma_{t, \theta}, \epsilon_{t, \theta}, \mu_{t, \theta} \in L^p(\X)}$. Since
	\[ \norm{R_{t,\theta}}^p =
	\norm{\sigma_{t,\theta} \odot \epsilon_{t, \theta} + \mu_{t,\theta}}^p \leq (\norm{\sigma_{t,\theta} \odot \epsilon_{t, \theta}} + \norm{\mu_{t,\theta}})^p ~,
	\]
	we obtain using the binomial identity
	\begin{equation*}
	\begin{split}
	\norm{R_{t,\theta}}_p^p &=  \mathbb{E} [\norm{R_{t,\theta}}^p]  \\
	&\leq \sum_{k=0}^{p} {p \choose k} \mathbb{E} [\norm{\sigma_{t, \theta} \odot \epsilon_{t, \theta}}^k \norm{\mu_{t, \theta}}^{p-k}] \\
	&\leq \sum_{k=0}^{p} {p \choose k} \left( \mathbb{E} \left[\norm{\sigma_{t, \theta} \odot \epsilon_{t, \theta}}^{2k} \right] \mathbb{E} \left[\norm{\mu_{t, \theta}}^{2(p-k)} \right]  \right)^{\frac{1}{2}},
	\end{split}
	\end{equation*}
	where the last inequality derives from the Cauchy-Schwarz inequality. Using the independence and the $L^p$-property of the volatility and innovation NP (cf. \autoref{rem:independence}), we obtain for arbitrary $q \in \mathbb{N}$ that
	\begin{equation*}
	\begin{split}
	\mathbb{E} \left[\norm{\sigma_{t, \theta} \odot \epsilon_{t, \theta}}^{q} \right]
	= \mathbb{E} \left[\sum_{i=1}^{N_X} | \sigma_{t, \theta,i} \ \epsilon_{t, \theta,i}|^q \right] 
	\end{split} 
	= \sum_{i=1}^{N_X} \mathbb{E} \left[| \sigma_{t, \theta,i}|^q \right] \mathbb{E} \left[| \epsilon_{t, \theta,i}|^q \right] < \infty.
	\end{equation*} 

\end{proof}

Since financial time series are generally considered to exhibit heavy-tails (see \hyperref[sec:stylized_facts]{Section~\ref*{sec:stylized_facts}}), the existence of all moments is an undesirable property. In the next section, we present a heuristic that works well to preprocess the historical log return process and enables the log return NP to empirically generate heavy-tails. 

Last, we motivate the choice of the latent process. We give a proof for the MLP which can be naturally extended to the setting of TCNs and especially SVNNs. The statement demonstrates that choosing the latent process as i.i.d. Gaussian noise benefits stability during optimization, i.e. back-propagation and parameter updates. 
\begin{corollary}
	\label{cor:backprop}
	Under the assumptions of \autoref{prop:lp_space_property_of_nn} the random variable of back-propagated gradients $\nabla_\theta g_\theta(Z)$ is an element of the space $L^p(\Theta)$.
\end{corollary}
\begin{proof}
	Without loss of generality assume that $N_X=1$. Using the notation from \autoref{def:mlp} and that $g_\theta$ has $L$ hidden layers, the gradient of $g_\theta(z)$ with respect to the hidden weight matrix $W^{(k)}, \ k \leq L+1$, is defined for $z \in \Z$ by
	\begin{equation*}
	\label{eq:backprop}
	\nabla_{W^{(k)}} g_\theta(z) =  \left(\prod_{l=k}^{L} D^{(l)}(z) W^{(l+1)^T}\right) \otimes  g_{1:k-1, \theta}(z)
	\end{equation*}
	where $\otimes$ denotes the outer product, $g_{1:k-1, \theta} \coloneqq g_{k-1, \theta} \circ \dots \circ g_{1, \theta}$ and $D^{(l)}(z) = \operatorname{diag}(\phi' (W^{(l)} g_{1:l-1,\theta}(z) ))$ (compare \cite[Chapter 6.5]{deeplearningbook}). Since the MLP is defined as a composition of Lipschitz functions \autoref{prop:lp_space_property_of_nn} yields $g_{1:k, \theta}(Z) \in L^p(\mathbb{R}^{N_k})$ for all $k \leq L$. Similarly, the boundedness of $\phi'$ implies that for an induced matrix norm the random variable $\norm{D^{(l)}(Z)}$ is $\mathbb{P}$-almost surely bounded by some constant $A>0$ for all $l \leq L$. By applying both properties we obtain 
	\[
	\mathbb{E} \left\lbrack \|\nabla_{W^{(k)}} g_\theta(Z)\|^p \right\rbrack \leq B_k \ \mathbb{E}\left\lbrack \|  g_{1:k-1, \theta}(Z) \|^p \right\rbrack < \infty
	\]
	with 
	\[ B_k \coloneqq A^{p(L-k+1)} \left( \prod_{l=k}^{^L} \| W^{(l+1)^T} \| \right)^p ~. \]
	With a similar argument one can show that the random gradients with respect to the biases $b^{(k)}, k = 1, \dots, L+1$ are also an element of $L^p$, thus concluding the proof. 
\end{proof}
Note that \autoref{cor:backprop} does not necessarily hold for a heavy-tailed latent process. Such that using a heavier-tailed latent process may enable the generation of heavy-tails, however may come at the cost of optimization instabilities. We leave it as future work to inspect how preprocessing techniques can be used to stabilize training.

\subsection{Generating heavier-tails and modeling assumptions}
\label{sec:generating_fat_tails}
\autoref{prop:lp_space_property_of_nn} implies that all moments of the log return NP exist. Furthermore,  \cite{copula_and_marginal_flows} show that the tails fall at least at a square-exponential rate. Real financial time series are generally considered to
exhibit heavy-tails and an unbounded p-th moment for some $p \in (2,5]$ \parencite{stylized_facts}. The Lambert W probability transform, as mentioned in \cite{lambert_gaussanize}, can therefore be used to generate heavier tails. The Lambert W probability transform of an $\R$-valued random variable is defined as follows.
\begin{definition}[Lambert W$\times F_X$]
	Let $\delta \in \R $ and $X$ be an $\R$-valued random variable with mean $\mu$, standard deviation $\sigma$ and cumulative distribution function $F_X$. The location-scale Lambert W$\times F_X$ transformed random variable $Y$ is defined by
	\begin{equation} 
	\label{eq:lambert_trans}
	Y = U \exp \normalbrack{\dfrac{\delta}{2} U^2 } \sigma + \mu ~,
	\end{equation}
	where $U := \dfrac{X - \mu}{\sigma}$ is the normalizing transform.
\end{definition}

For $\delta \in [0, \infty)$ the transformation used in \autoref{eq:lambert_trans} is of special interest as it is guaranteed to be bijective and differentiable. Hence, the transformations specific parameters $\gamma = (\mu, \sigma, \delta)$ can be estimated via maximum likelihood. Moreover, for $\delta > 0$ the Lambert $W \times F_X$ transformed random variable has heavier tails than $X$.

We therefore apply the \textit{inverse Lambert W} probability transform to the asset's log returns and use the principle of quasi maximum likelihood to estimate the model parameters \parencite[Section~4.1]{lambert_gaussanize}. The log return NP is then optimized to approximate the inverse Lambert W transformed (lighter-tailed) log return process, from here on denoted by $R^W\coloneqq (R^W_t)_{t\in \mathbb{N}}$. 

Using the Lambert W transformed log return process $R^W$ we can formulate our model assumptions, when using SVNNs as the underlying generator. 
\begin{assumption}
	The inverse Lambert W transformed spot log returns $R^W$ can be represented by a log return neural process $ R_\theta$ for some $\theta \in \Theta$. 
	\label{modeling_assumption}
\end{assumption}
\autoref{modeling_assumption} has two important implications. First, by construction the log return NP is stationary such that the historical log return process is assumed to be stationary. Second, log return NPs can capture dynamics up to the RFS of the TCN in use. Therefore, \autoref{modeling_assumption} implies for an RFS $T^{(g)}$ that for any $t \in \mathbb{Z}$ the random variables $R_{t}, R_{t+T^{(g)}+1}$ are independent.

\subsection{Risk-neutral representation of \boldmath$R_\theta$}
\label{sec:risk_neutral}
At this point we cannot value options under a log return NP, as we do not know a transition to its risk-neutral distribution. We address this aspect in this section. To this end, consider a one-dimensional log return NP 
\[ R_{t, \theta} = \sigma_{t,\theta} \, \epsilon_{t,\theta} + \mu_{t, \theta} ~.\]
The spot prices are then defined recursively by
\[ S_{t, \theta} = S_{t-1, \theta} \, \exp(R_{t, \theta}) \quad \textrm{for all} \ t \in \mathbb{N} ~, \]
where $S_{0, \theta} = S_0$ denotes the current price of the underlying. Moreover, assume a constant interest rate $r$ and define the discounted stock price process $(\tilde{S}_{t, \theta})_{t \in \mathbb{N}}$ by
\[ \tilde{S}_{t, \theta} := \frac{S_{t, \theta}}{\exp(rt)} ~.\]
In particular, the discounted price process fulfils the recursion
\[ \tilde{S}_{t, \theta} = \tilde{S}_{t-1, \theta} \, \exp(R_{t, \theta} - r) ~.\]
In its risk-neutral representation, the discounted stock price process has to be a martingale. Therefore, we can use that $\tilde{S}_{t-1, \theta}$ is $\mathcal{F}^Z_{t-1}$-measurable and get
\begin{equation*}
\begin{split}
\mathbb{E}[\tilde{S}_{t, \theta} | \mathcal{F}^Z_{t-1}] &= \mathbb{E}[ \tilde{S}_{t-1, \theta} \, \exp(R_{t, \theta} - r) |  \mathcal{F}^Z_{t-1}] \\
&= \tilde{S}_{t-1, \theta} \,  \exp(-r) \, \mathbb{E}[\exp(\sigma_{t, \theta} \, \epsilon_{t, \theta} + \mu_{t, \theta}) | \mathcal{F}^Z_{t-1}] ~.
\end{split}
\end{equation*}
Hence to obtain a martingale, we have to correct for the corresponding term. Therefore, let us consider the conditional expectation in more detail. As the volatility and drift NPs are $\mathcal{F}_{t-1}^Z$-measurable and $\epsilon_{t, \theta}$ is independent of $\mathcal{F}_{t-1}^Z$, we can write
\[ \mathbb{E}[\exp(\sigma_{t, \theta} \, \epsilon_{t, \theta} + \mu_{t, \theta}) | \mathcal{F}^Z_{t-1}] = \mathbb{E}[\exp(\sigma \, \epsilon_{t, \theta} + \mu)]_{\substack{\sigma = \sigma_{t, \theta} \\ \mu = \mu_{t, \theta}}} =: h(\sigma_{t, \theta}, \mu_{t, \theta}) ~.\]
Depending on the innovation NP $\epsilon_{t, \theta}$, the function $h$ might be given explicitly or has to be estimated using a Monte Carlo estimator.

As a result, we can define the \textit{risk-neutral log return Neural Process} $R_{t,\theta}^M$ as 
\[ R_{t, \theta}^M := R_{t, \theta} - \log(h(\sigma_{t,\theta}, \mu_{t,\theta})) + r ~,\]
which is a corrected log return NP. The corresponding \textit{discounted risk-neutral spot price process} is then given by the recursion
\[ \tilde{S}_{t, \theta}^M = \tilde{S}_{t-1, \theta}^M \, \exp(R_{t, \theta}^M - r) = \tilde{S}_{t-1, \theta}^M \, \exp(R_{t, \theta} - \log(h(\sigma_{t, \theta}, \mu_{t, \theta})))\]
and defines a martingale.

In particular, this recursion can be solved to obtain an explicit formula for the (discounted) risk-neutral spot price process
\[ \tilde{S}_{t, \theta}^M = S_0 \, \exp \left(\sum_{s=1}^t [R_{s, \theta} - \log(h(\sigma_{s,\theta}, \mu_{s,\theta}))] \right) \]
\[ S_{t, \theta}^M = S_0 \, \exp \left(\sum_{s=1}^t [R_{s, \theta} - \log(h(\sigma_{s,\theta}, \mu_{s,\theta}))] + rt \right) ~.\]

It remains the problem of inferring the parameters of the underlying model. In the case of financial time series, the discriminator is used to distinguish between generated and real (observable) financial time series. What is different here is that risk-neutral asset paths are not observable. Therefore, we can not train the generator-discriminator pair in the same way as for financial time series. An approach would be a classical least square calibration by option prices, where we would use Monte Carlo of the generated risk-neutral paths as an estimate of the models option price. We leave this as future work.

\subsection{Constrained log return neural processes}
\label{sec:constrained_log_NP}
An interesting application is to constrain either the volatility or the innovations NP to satisfy certain conditions. We exemplify this for the one-dimensional case ($N_X = 1$), where the innovations NP is constrained to represent a standard normal distributed random variable
\begin{equation*}
\epsilon_{t, \theta} \sim \mathcal{N}(0, 1) \quad \textrm{for all} \ t \in \mathbb{Z}~.
\end{equation*}

In this case, the risk-neutral dynamics can be simplified, since the conditional expectation $h(\sigma_{t,\theta}, \mu_{t,\theta})$ can be calculated explicitly:
\[ h(\sigma_{t,\theta}, \mu_{t,\theta}) = \mathbb{E}[\exp(\underbrace{\sigma \, \epsilon_{t, \theta} + \mu}_\text{$\sim \mathcal{N}(\mu, \sigma^2)$})]_{\substack{\sigma = \sigma_{t, \theta} \\ \mu = \mu_{t, \theta}}} = \exp \left(\mu_{t, \theta} + \frac{\sigma_{t,\theta}^2}{2} \right) ~.\]
Hence, the risk-neutral log return NP is given by
\[ R_{t, \theta}^M = \sigma_{t, \theta} \, \epsilon_{t, \theta} - \frac{\sigma_{t, \theta}^2}{2} + r \]
and the discounted risk-neutral price process fulfils the recursion
\[ \tilde{S}_{t, \theta}^M = \tilde{S}_{t-1, \theta}^M \, \exp \left(\sigma_{t, \theta} \, \epsilon_{t, \theta} - \frac{\sigma_{t, \theta}^2}{2} \right) ~.\]
In particular, solving this recursion gives the explicit representations
\[ \tilde{S}_{t, \theta}^M = S_0 \, \exp \left(\sum_{s=1}^t \left( \sigma_{s, \theta} \, \epsilon_{s, \theta} - \frac{\sigma_{s, \theta}^2}{2} \right) \right) \]
\[ S_{t, \theta}^M = S_0 \, \exp \left(\sum_{s=1}^t \left[ \sigma_{s, \theta} \, \epsilon_{s, \theta} - \frac{\sigma_{s, \theta}^2}{2} \right] + rt \right) ~. \]
\begin{remark}[Comparison to the Black-Scholes model]
	In the one-dimensional Black-Scholes model, the risk-neutral distribution of the price process is given by
	\[ S_t^{Q, BS} = S_0 \, \exp \left( \left( r - \frac{1}{2} \sigma^2 \right) t + \sigma W_t^Q \right) = S_0 \, \exp \left(\sigma W_t^Q - \frac{1}{2} \sigma^2 t + rt \right) ~. \]
	The similarities to the price process given by the risk-neutral log return NP are clearly visible. Most importantly, in contrast to Black-Scholes, the model presented here does not assume a constant volatility and instead models it using the volatility generator. 
\end{remark}

In the same way, the volatility NP can be constrained to represent a known stochastic process such as the CIR process or the variance process of the $\operatorname{GARCH}(p,q)$ model. Both settings allow us to generate insights of the latent dynamics of the stochastic process at hand and thereby enable to validate modeling assumptions.

\section{Preprocessing}
\label{sec:preprocessing}
Prior to passing a realization of a financial time series $s_{0:T} \in \R^{N_X\times (T+1)}$ to the discriminator, the time series has to be preprocessed. The applied pipeline is displayed in \autoref{tikz:pipeline}. We briefly explain each of the steps taken. Note that all of the used transformations, excluding the rolling window, are invertible and thu, allow a series sampled from a log return NP to be post-processed by inverting the steps 1-4 to obtain the desired form. Also, observe that the pipeline includes the inverse Lambert W transformation as earlier discussed in \autoref{sec:generating_fat_tails}.
\tikzset{
	>=stealth',
	punktchain/.style={
		rectangle, 
		rounded corners, 
		draw=black, very thick,
		text width=6em, 
		minimum height=4em, 
		text centered, 
		on chain},
	line/.style={draw, thick, <-},
	element/.style={
		tape,
		top color=white,
		bottom color=blue!50!black!60!,
		minimum width=8em,
		draw=blue!40!black!90, very thick,
		text width=10em, 
		minimum height=3.5em, 
		text centered, 
		on chain},
	every join/.style={->, thick,shorten >=1pt},
	decoration={brace},
	tuborg/.style={decorate},
	tubnode/.style={midway, right=2pt},
}
\begin{figure}[H]
	\centering
	\begin{tikzpicture}
	[node distance=.4cm, start chain=going right,]
	\node[punktchain, join] (time series) {Time Series $s_{0:T}$};
	\node[punktchain, join] (log) {\textbf{Step 1}:\\Log Returns $r_{1:T}$};
	\node[punktchain, join] (scaler1) {\textbf{Step 2}: Normalize};
	\node[punktchain, join] (lambert) {\textbf{Step 3}: Inverse Lambert W Transform};
	\node[punktchain, join, below=of lambert] (scaler2) {\textbf{Step 4}: Normalize};
	\node[punktchain, join, left=of scaler2] (rollingwindow) {\textbf{Step 5}: Rolling Window };
	\node[punktchain, join, left=of rollingwindow] (torch) {Cast to PyTorch};
	\node[punktchain, join, left=of torch] (final) {Preprocessed Series};
	\end{tikzpicture}
	\caption{Condensed representation of the preprocessing pipeline.}
	\label{tikz:pipeline}
\end{figure}
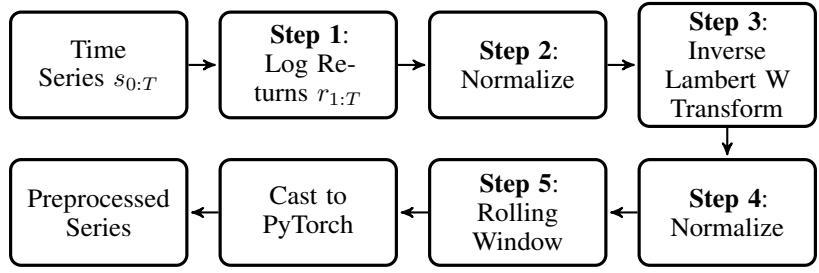

\subsubsection*{Step 1: Log returns \boldmath$r_{1:T}$}
Calculate the log return series
\[ r_t = \log \left(\frac{s_t}{s_{t-1}} \right)  \quad \textrm{for all} \ t \in \{1,...,T\} ~.\]
\subsubsection*{Step 2 \& 4: Normalize}
For numerical reasons, we normalize the data in order to obtain a series with zero mean and unit variance, which is thoroughly derived in \cite{efficient_backprop}.

\begin{figure}[htp]
	\centering
	\begin{subfigure}[b]{0.48\textwidth}
		\includegraphics[width=\textwidth]{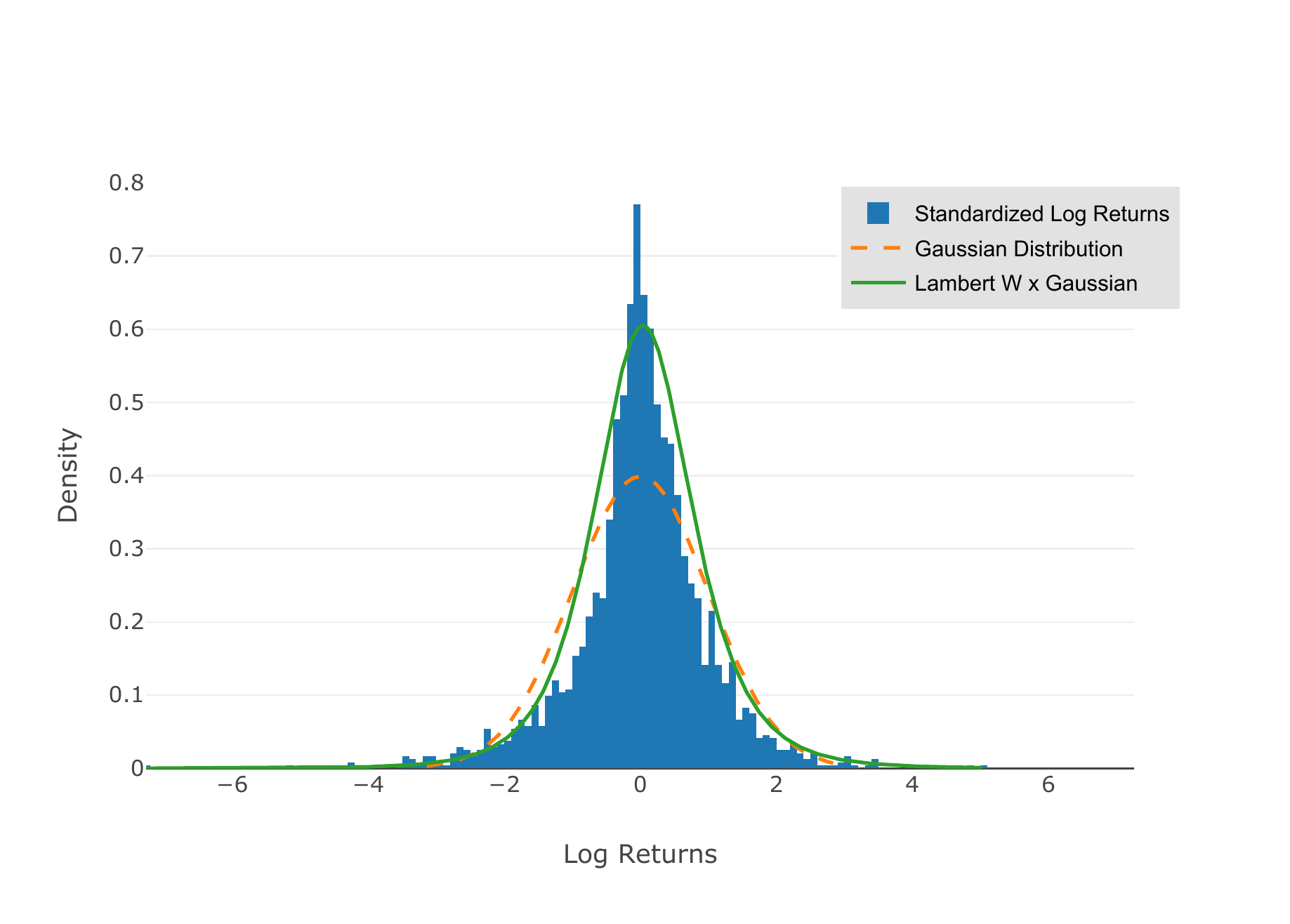}
		\caption{Standardized Log Returns}
		\label{fig:}
	\end{subfigure}
	\begin{subfigure}[b]{0.48\textwidth}
		\includegraphics[width=\textwidth]{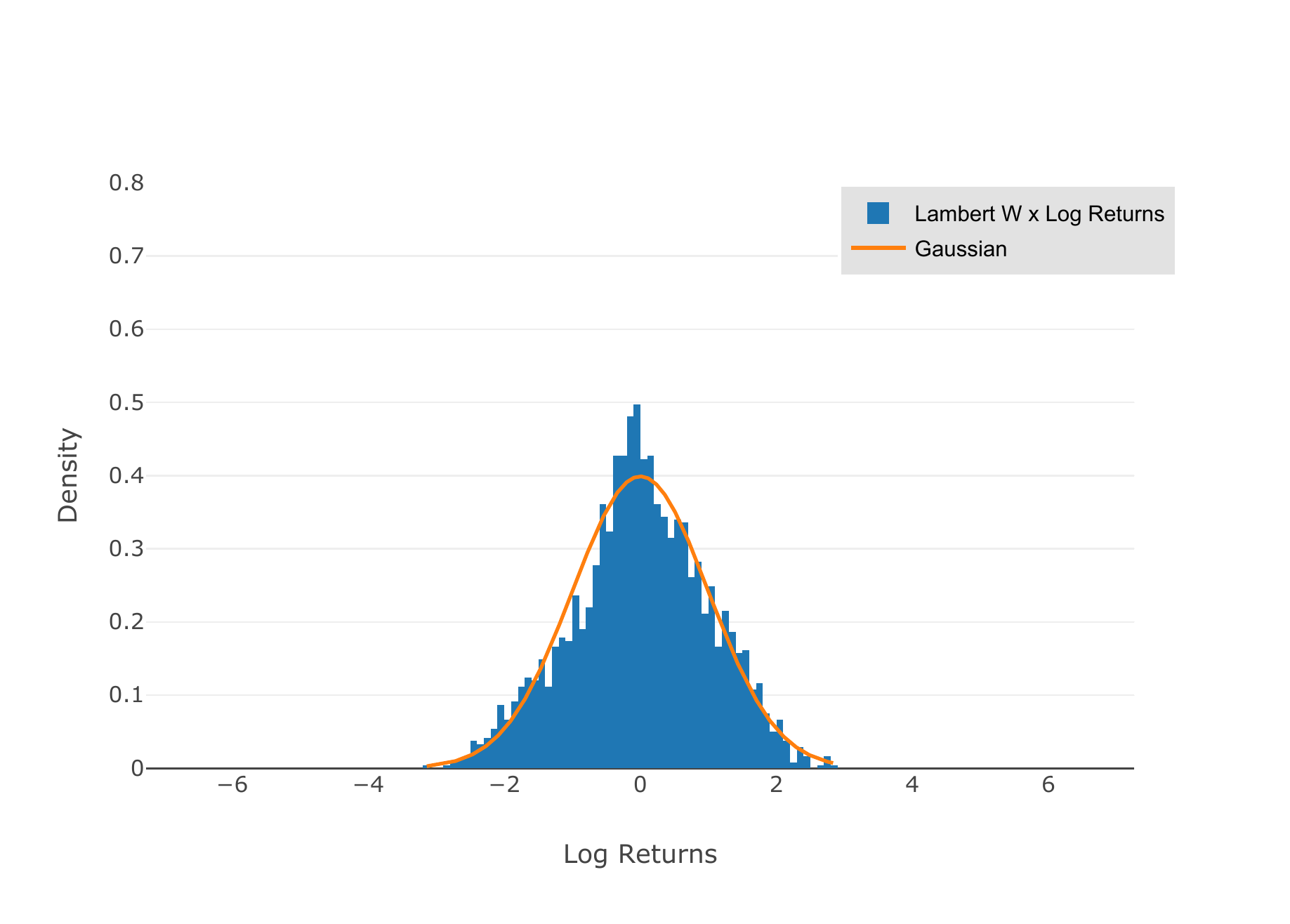}
		\caption{Lambert W x Log Returns}
		\label{fig:lambertwapplied}
	\end{subfigure}
	\caption{(a) The original S\&P 500 log returns and the fitted probability density function of a Lambert W x Gaussian random variable. (b) The inverse Lambert W transformed log returns and the probability density function of a Gaussian random variable.}
	\label{fig:lambertw}
\end{figure}

\subsubsection*{Step 3: Inverse Lambert W transform}
The suggested transformation applied to the log returns of the S\&P 500 is displayed in \autoref{fig:lambertw}. It shows the standardized original distribution of the S\&P 500 log returns and the inverse Lambert W transformed log return distribution. Observe that the transformed standardized log return distribution in \autoref{fig:lambertwapplied} approximately follows the standard normal distribution and thereby circumvents the issue of not being able to generate the heavy-tail of the original distribution.

\subsubsection*{Step 5: Rolling window}
When considering a discriminator with receptive field size $T^{(d)}$, we apply a rolling window of corresponding length and stride one to the preprocessed log return sequence $r^{(\rho)}_t$. Hence, for $t \in \{1, \dots, T-\rfsd\}$ we define the sub-sequences
\begin{equation*}
r_{1:\rfsd}^{(t)} \coloneqq r^{(\rho)}_{t:(\rfsd+t -1)} \in \R^{N_Z \times \rfsd} ~.
\label{eq:rolling_window}
\end{equation*}

\begin{remark}
Note that sliding a rolling window introduces a bias, since log returns at the beginning and end of the time series are under-sampled when training the Quant GAN. This bias can be corrected by using a (non-uniform) weighted sampling scheme when sampling batches from the training set.
\end{remark}

\section{Numerical results}
\label{sec:numerical}
In this section we test the generative capabilities of Quant QANs by modeling the log returns of the S\&P 500 index. For comparison we apply the well-known GARCH(1,1) model to the same data. Our numerical results highlight that Quant GANs can learn a neural process that matches the empirical distribution and dependence properties far better than the presented GARCH model.

\subsection{Setting and implementation}
The implementation was written with the programming language python. NN architectures were implemented by employing the python package pytorch \parencite{paszke2017automatic, pytorch}, an automatic differentiation library primarily used for NN computations and optimization. The training time of the NNs was decreased by using the CUDA-backend of pytorch in combination with a CUDA-enabled graphics processing unit (GPU). We used the GPU RTX 2070 from NVIDIA in order to train larger models. During preprocessing and evaluation we used the packages numpy and scipy \cite{scipy}\footnote{Note that preprocessing and approximation of the Lambert W transform relevant parameters can be equivalently done by using the package pytorch.}. 

The architecture of the TCN used for the generator and the discriminator of the Quant GAN models is an extension of the architecture proposed in \cite[Section 3]{bai_empirical} by using skip connections (cf. \cite{wavenet}, \autoref{def:tcn_skip}). A detailed description of the architecture is given in \autoref{appendix_architecture}. During training of the generator and discriminator, we applied the \textit{GAN stability algorithm} from \cite{mescheder_which_gan}. 

\subsection{Models}
For modeling the log returns of the S\&P 500, we have a look at three different model architectures: a pure TCN model, a constrained log return NP and for comparison a simple GARCH model. 

In contrast to other approaches in the literature, the proposed models are based on TCNs instead of LSTMs. Note that this is mainly motivated by the superior performance of convolutional architectures to typical recurrent architectures in many sequence related tasks \parencite{bai_empirical}.  Therefore, we do not consider LSTM based models here and leave this as future work.

\subsubsection*{Pure TCN}
To evaluate the capabilities of a pure TCN model, we model the log returns directly using a TCN with receptive field size $T^{(g)}=127$ as the generator function, i.e. the return process is given by
\[ R_{t,\theta} = g_{\theta}(Z_{t-(T^{(g)}-1):t}) \]
for the three-dimensional noise prior $Z_t \stackrel{iid}{\sim} \mathcal{N}(\mathbf{0},I)$, ($N_Z=3$).
\subsubsection*{Constrained log return NP}
Assume a constrained log return NP (see \autoref{def:svnn} and \hyperref[sec:constrained_log_NP]{Section~\ref*{sec:constrained_log_NP}})
\[ R_{t, \theta} = \sigma_{t,\theta} \, \epsilon_{t,\theta} + \mu_{t, \theta} \]
with volatility NP $\sigma_{t,\theta}$, drift NP $\mu_{t,\theta}$ and an innovations NP $\epsilon_{t, \theta}$ constrained to being i.i.d. $\mathcal{N}(0,1)$-distributed. Here again the latent process is i.i.d. Gaussian noise $Z_t \stackrel{iid}{\sim} \mathcal{N}(\mathbf{0},I)$ for $N_Z=3$. The innovation process takes the form $\epsilon_{t, \theta} = Z_{t, 1}$ for any $t \in \mathbb{Z}$.
\subsubsection*{GARCH(1,1) with constant drift}
Assume a GARCH(1,1) model with constant drift, where
\begin{equation*}
\begin{split}
R_{t, \theta} &= \xi_t + \mu \\
\xi_t &= \sigma_t \,\epsilon_t \\
\sigma_t^2 &= \omega + \alpha \,\xi_{t-1}^2 + \beta \,\sigma_{t-1}^2 \\
\epsilon_t &\stackrel{iid}{\sim} \mathcal{N}(0,1)
\end{split}
\end{equation*}
for $\mu \in \mathbb{R}$, $\omega >0$, $\alpha, \beta \in [0,1]$ such that $\alpha + \beta < 1$ and the parameter vector $\theta = (\omega, \alpha, \beta, \mu)$. For more details on GARCH-processes see \cite{garch}.

\subsection{Evaluating a path simulator: metrics and scores}
\label{sec:metrics_scores}
To compare the dynamics of the three different models with the S\&P 500, we propose the use of the following metrics and scores. 
\subsubsection{Distributional metrics}
\paragraph{Earth mover distance}
Let $\mathbb{P}^h$ denote the historical and $\mathbb{P}^g$ the generated distribution of the (possibly lagged) log returns. The \textit{Earth Mover Distance} (or \textit{Wasserstein-1 distance}) is defined by
\[ \operatorname{EMD} \left(\mathbb{P}^h, \mathbb{P}^g \right) = \inf_{\pi \in \Pi(\mathbb{P}^h, \mathbb{P}^g)}  \mathbb{E}_{(X,Y) \sim \pi} [ || X - Y ||] ~,\]
where $\Pi(\mathbb{P}^h, \mathbb{P}^g)$ denotes the set of all joint probability distributions with marginals $\mathbb{P}^h$ and $\mathbb{P}^g$. Loosely speaking the earth mover distance describes how much probability \textit{mass} has to be moved to transform $\mathbb{P}^h$ into $\mathbb{P}^g$. For more details, see \cite{villani}.
\paragraph{DY metric}
Additionally we compute the \textit{DY metric} proposed in \cite{dragulescu_heston}. The DY metric is for $t \in \mathbb{N}$ defined by
\[
\operatorname{DY}(t) = \sum_{x} \abs{\log P_t^{h}(A_{t,x}) - \log P_t^g(A_{t,x})} ,
\]
where $P_t^{h}$ and $P_t^g$ denote the empirical probability density function of the historical and generated $t$-differenced log path. Further, $(A_{t, x})_x$ denotes a partitioning of the real number line such that for fixed $t$ and all $x$ we (approximately) have 
\[\log P_t^{h}(A_{t,x}) = 5 / T\] 
for $T$ the number of historical log returns. During evaluation we consider the time lags ${t \in \{1,5,20,100\}}$, which represent a comparison of the daily, weekly, monthly and 100-day log returns. 
\subsubsection{Dependence scores}
\paragraph{\boldmath$\operatorname{ACF}$ score}
The ACF score is proposed to compare the dependence properties of the historical and the generated time series. Let $r_{1:T}$ denote the historical log return series and $\{r^{(1)}_{1:\tilde T, \theta}, \dots, r^{(M)}_{1:\tilde T, \theta}\}$ a set of generated log return paths of length $\tilde T \in \mathbb{N}$. The autocorrelation is defined as a function of the time lag $\tau$ and a series $r_{1:T}$ and measures the correlation of the lagged time series with the series itself
\[
\mathcal{C}(\tau; r) = \operatorname{Corr}(r_{t+\tau}, r_t).
\]
Denoting by $C:\R^T \to [-1, 1]^S: r_{1:T} \mapsto (\mathcal{C}(1; r), \dots, \mathcal{C}(S; r))$ the autocorrelation function up to lag $S \leq T-1$, the $\operatorname{ACF}(f)$ score is computed for a function $f: \R \to \R $ as
\[
\operatorname{ACF}(f) \coloneqq \norm{
	C(f(r_{1:T})) - \dfrac{1}{M} \sum_{i=1}^M {C} \left(f \left(r^{(i)}_{1:T, \theta} \right) \right)
}_2 ~,
\]
where the function $f$ is applied element-wise to the series. 
The $\operatorname{ACF}$ score is computed for the functions $f(x) = x$, $f(x)=x^2$ and $f(x) = \abs{x}$ and constants $S=250$, $M=500$, $\tilde T = 4000$.

\paragraph{Leverage effect score}
Similar to the $\operatorname{ACF}$ score, the leverage effect score provides a comparison of the historical and the generated time dependence. The leverage effect is measured using the correlation of the lagged, squared log returns and the log returns themselves, i.e. we consider
\[ \mathcal{L}(\tau;r) = \operatorname{Corr} \left(r^2_{t + \tau}, r_t \right) \]
for lag $\tau$. Denoting by $L:\R^T \to [-1, 1]^S : r_{1:T} \mapsto \left(\mathcal{L}(1;r),..., \mathcal{L}(S;r) \right)$ the leverage effect function up to lag $S \leq T-1$, the leverage effect score is defined by
\[
\norm{
	L(r_{1:T}) - \dfrac{1}{M} \sum_{i=1}^M L \left(r^{(i)}_{1:\tilde T, \theta} \right)
}_2 ~.
\]
In the benchmark, the leverage effect score is computed for $S=250$, $M=500$ and $\tilde T = 4000$; the same as for the $\operatorname{ACF}$ score.

\subsection{Generating the S\&P 500 index}
We consider daily spot-prices of the S\&P 500 from May 2009 until December 2018. As expected, the GAN training was very irregular and did not converge to a local optimum of the objective function. Instead for each GAN model, several learned parameter sets were saved and the best setup was chosen based on the evaluation metrics described in  \autoref{sec:metrics_scores}. For each model, we only present the results of the best performing setup. In the appendix, we show:
\begin{itemize}
	\item histograms of the real and the generated log returns on a daily, weekly, monthly and 100-day basis,
	\item mean-autocorrelation functions of the serial, squared and absolute log returns together with the empirical confidence bands,
	\item the correlations between the squared, lagged and the non-squared log returns together with the empirical confidence band as a proxy for the leverage effect,
	\item 5 plus additionally 50 exemplary generated log paths.
\end{itemize}
Further, \autoref{table:sp500_metrics} presents the values of the evaluated metrics for each of the models. 

\begin{center}
	\begin{tabular}{ c c c c c c c}
		\clineB{1-5}{2.5}
		time series & time span & \# of observations & ADF-statistic & p-value
		\\\hline
		S\&P 500 & May 2009 - December 2018 & 2413 & -10.87 & $1.36 \times 10^{-19} $
		\\ 
		\clineB{1-5}{2.5}
	\end{tabular}
	\captionof{table}{Considered financial time series.}
\end{center}
We validate \autoref{modeling_assumption} by applying the augmented Dickey-Fuller test \parencite{adf} to the time series and obtain a test statistic of -10.87 and a p-value of $1.36 \times 10^{-19} $, which is a strong indication that the series is stationary. Further as expected, the Lambert W transformation significantly helped the reported models to capture the stylized facts of the asset returns; in particular in modeling the tails of the distribution.

\subsubsection{Pure TCN}
The displayed graphics show that the TCN model is capable of precisely modeling distributional and dependence properties present in the real S\&P. 

As can be seen in \autoref{fig:s&p_tcn_hist}, the generated log returns closely match the histogram of the real returns on each of the presented time scales. Even for 100-day lagged returns, the fit is quite good. 

The same holds true for the ACFs and the leverage effect plot, which deal with the dependence structure inherent in the data (see \autoref{fig:s&p_tcn_acf}). The TCN  accurately models the sharp drop in the ACF of the serial returns as well as the slowly decaying ACF of the squared and absolute log returns. Moreover, the leverage effect is captured by a negative correlation between the squared and non-squared log returns for small time lags.

Recall that the displayed ACFs corresponding to the generated returns are mean ACFs and thereby much smoother than the ACF of the real returns.

Furthermore, the exemplary log paths shown in \autoref{fig:s&p_tcn_5_log_paths} and \autoref{fig:s&p_tcn_log_paths} exhibit reasonable patterns and demonstrate the structural diversity possible in the TCN model.

For all except two metrics evaluated in \autoref{table:sp500_metrics}, the TCN performs best. In particular, the TCN clearly outperforms the GARCH(1,1) model in each metric, often by a factor 2-10.

\subsubsection{Constrained SVNN}
The C-SVNN performs comparable to the TCN, but slightly worse in most of the evaluated metrics. This is expectable as the C-SVNN has less degrees of freedom due to the restrictions compared to the pure TCN.

As is displayed in the graphics, the C-SVNN is able to capture the same properties as the TCN (see \autoref{fig:s&p_c_svnn_drift_hist} and \autoref{fig:s&p_c_svnn_drift_acf}). Merely the ACF of the squared and absolute log returns is better modeled by the TCN. 

Despite the restrictions, the evaluated metrics of the C-SVNN in \autoref{table:sp500_metrics} are nearly as good as the results of the TCN. Furthermore, the C-SVNN also outperforms GARCH(1,1) model significantly.

Recall further that the SVNN has the structural advantages that the volatility can be directly modeled and a transition to its martingale distribution is known.

\subsubsection{GARCH(1,1) with constant drift}
The GARCH model is clearly outperformed by the previously considered GAN-approaches two presented models in modeling distributional as well as dependence properties.

As indicated by the stylized facts of asset returns (see \hyperref[sec:stylized_facts]{Section~\ref*{sec:stylized_facts}}), the assumed normal distribution of the GARCH model places too less probability mass at the peak and the tails of the log return distribution as displayed by the histograms in \autoref{fig:s&p_garch_drift_hist}.

In contrast, the autocorrelation function is captured quite well. This should not come as a surprise, as the GARCH structure was designed to capture this dependence. The main characteristics of the ACF plots in \autoref{fig:s&p_garch_drift_acf} are modeled, but the GARCH approach fails in exactness compared to the TCN and C-SVNN model. Note further that the leverage effect is not captured at all.

\autoref{table:sp500_metrics} supports this graphical assessment as the GARCH model performs worst in each of the evaluated metrics. For the ACF scores, the GARCH model is quite comparable to the other models as was already pointed out looking at the graphics. In contrast, in terms of the EMD and DY metric which focus on the return distribution, the GARCH model is clearly outperformed by the other two models.

\begin{table}[h]
	\begin{center}
		\begin{tabular}{c c c c}
			\clineB{1-4}{2.5}
			& TCN & C-SVNN with drift & GARCH(1,1) \\ \hline
			$\operatorname{EMD}(1)$ & \textbf{0.0039} & 0.0040 & 0.0199\\ 
			$\operatorname{EMD}(5)$ & \textbf{0.0039} & 0.0040 & 0.0145\\
			$\operatorname{EMD}(20)$ & \textbf{0.0040} & 0.0069 & 0.0276 \\
			$\operatorname{EMD}(100)$ & \textbf{0.0154} & 0.0464 & 0.0935 \\ 
			\hline
			DY(1) & \textbf{19.1199} & 19.8523 & 32.7090\\
			DY(5) & \textbf{21.1167} & 21.2445 & 27.4760\\
			DY(20) & 26.3294 & \textbf{25.0464} & 39.3796\\
			DY(100) & 28.1315 & \textbf{25.8081} & 46.4779 \\
			\hline
			$\operatorname{ACF}({\operatorname{id}})$ & \textbf{0.0212} & 0.0220 & 0.0223\\
			$\operatorname{ACF}({\abs{\cdot}})$ & \textbf{0.0248} & 0.0287 & 0.0291\\
			$\operatorname{ACF}({\normalbrack{\cdot}^2})$ & \textbf{0.0214} & 0.0245 & 0.0253\\ \hline
			Leverage Effect & \textbf{0.3291} & 0.3351 & 0.4636 \\
			\clineB{1-4}{2.5}
		\end{tabular}
		\caption{Evaluated metrics for the three models applied. For each row, the best value is printed bold.}
		\label{table:sp500_metrics}
	\end{center}
\end{table}

\section{Conclusion and future work}
\label{sec:conclusion}
In this paper we showed that recently developed NN architectures can be used in an adversarial modeling framework to approximate financial time series in discrete-time. Although these methods have been notoriously hard to train, advances in GANs showed that they can deliver competitive results and, as GAN training procedures develop, promise even better performance in future.

For Quant GANs to flourish in the future there are two fundamental challenges that need to be addressed. The first is an exact modeling and extrapolation of the generated tail by incorporating prior knowledge such as the estimated tail-index. Second, a single metric needs to be developed which unifies distributional metrics with dependence scores we used in this paper and allows to benchmark different generator architectures. Once these points are sufficiently studied and addressed, Quant GANs offer a data-driven method that surpasses the performance of other conventional models from mathematical finance.

\clearpage
\bibliography{references}
\clearpage
\begin{appendix}
\section{Numerical Results}
\renewcommand\thefigure{\thesection.\arabic{figure}}
\setcounter{figure}{0}      
\subsection{Pure TCN}

\begin{figure}[H]
	\centering
	\includegraphics[width=0.9\textwidth]{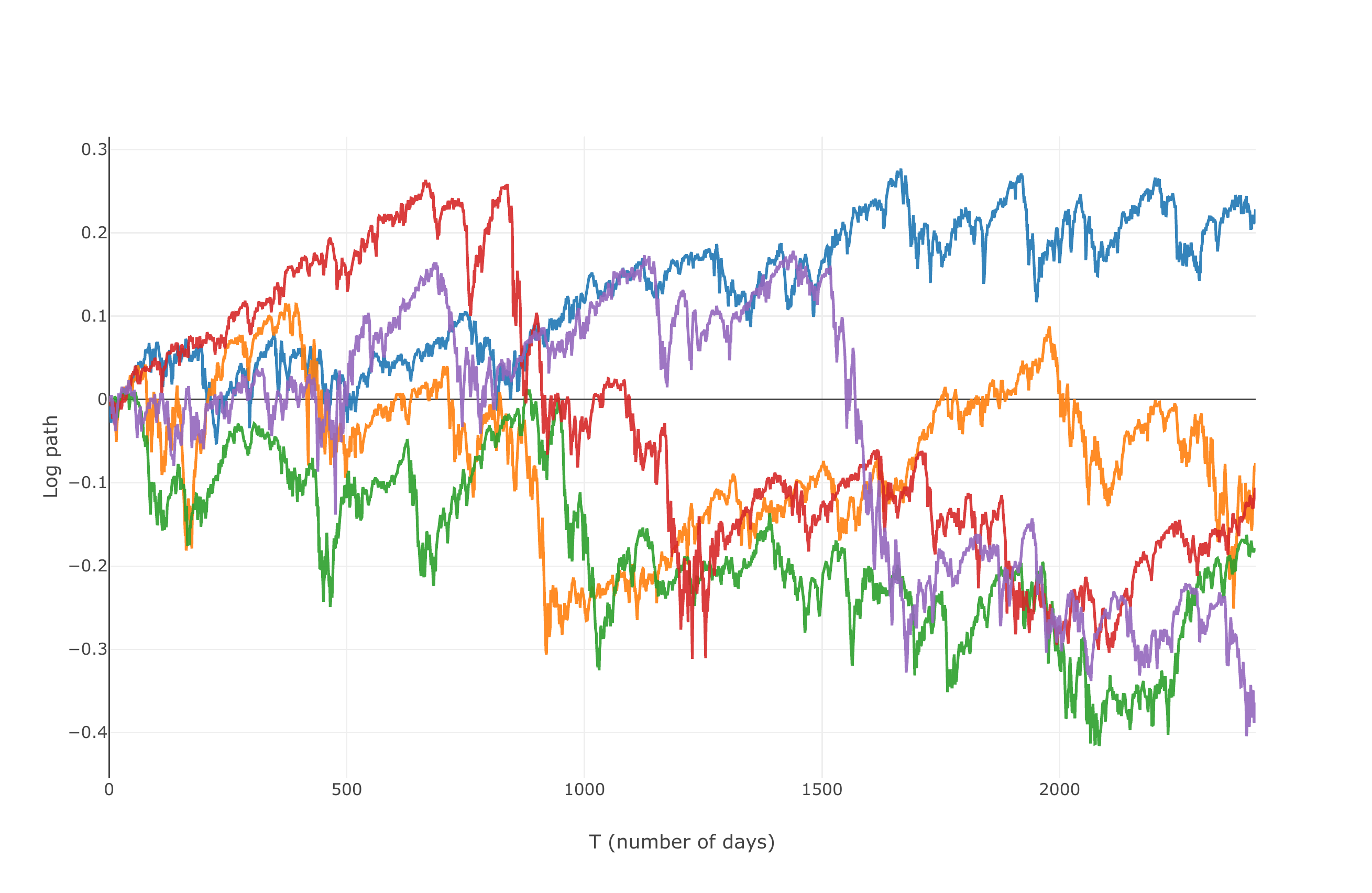}
	\caption{5 generated log paths}
	\label{fig:s&p_tcn_5_log_paths}
\end{figure}
\vspace{-0.75cm}
\begin{figure}[H]
	\centering
	\includegraphics[width=0.9\textwidth]{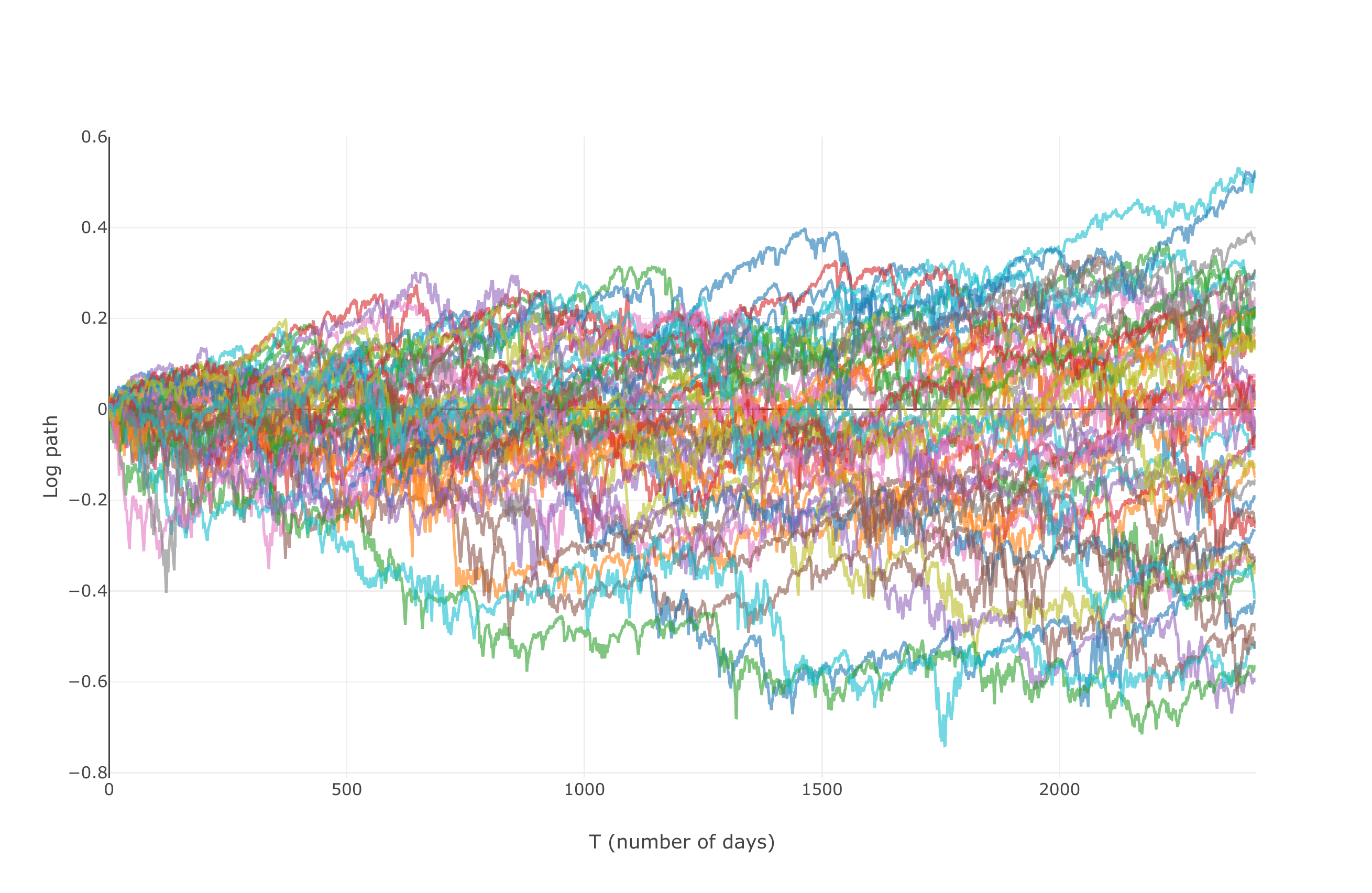}
	\caption{50 generated log paths}
	\label{fig:s&p_tcn_log_paths}
\end{figure}

\begin{figure}[H]
	\vspace*{-2em}
	\captionsetup[subfigure]{aboveskip=-2pt,belowskip=-2pt}
	\centering
	\begin{subfigure}[b]{0.49\textwidth}
		\includegraphics[width=0.9\textwidth]{figures/tcn/histogramm_lag_1.pdf}
		\caption{Daily}
		\label{fig:tcn_hist1}
	\end{subfigure} 
	\begin{subfigure}[b]{0.49\textwidth}
		\includegraphics[width=0.9\textwidth]{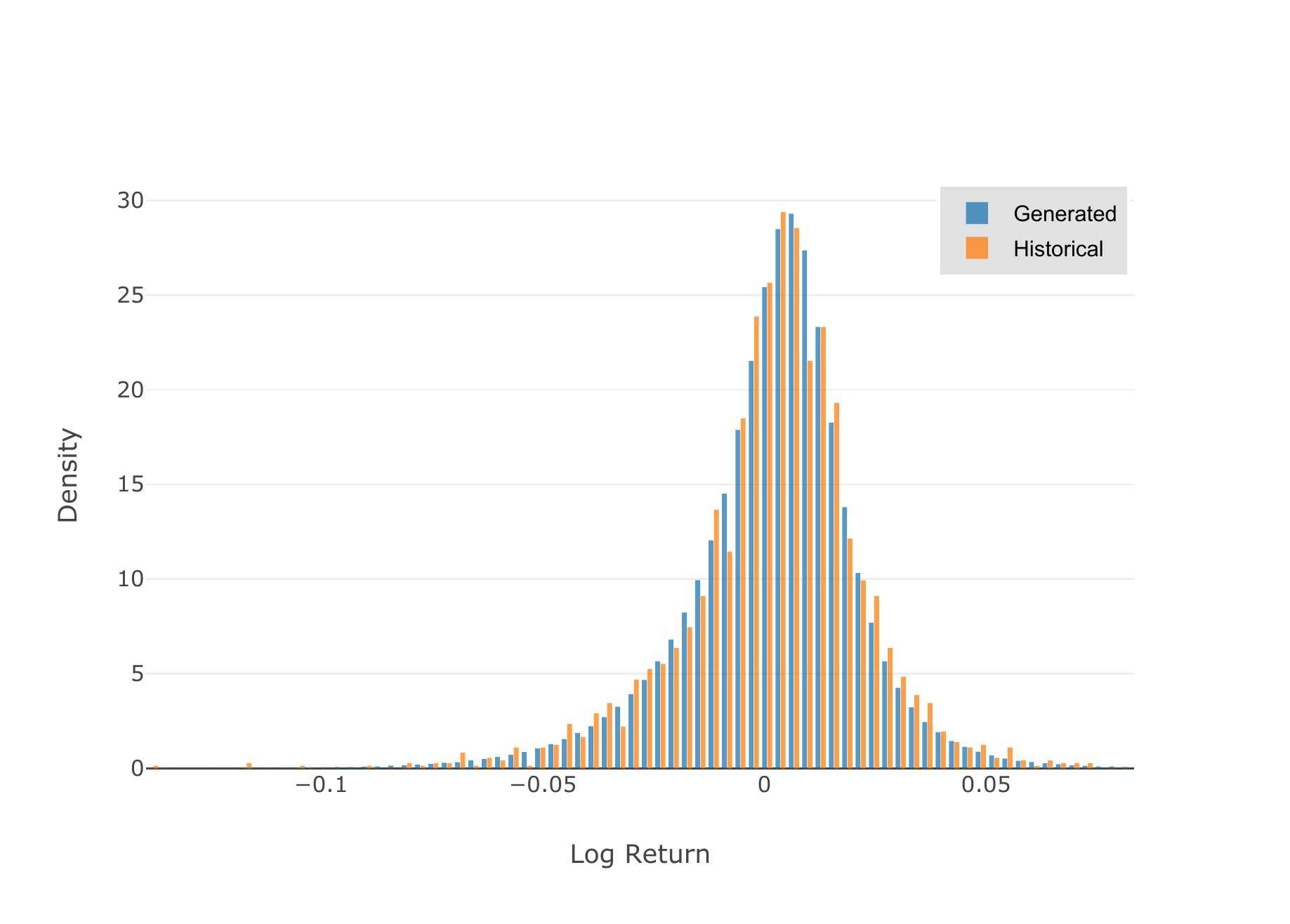}
		\caption{Weekly}
		\label{fig:tcn_hist5}
	\end{subfigure}
	\\
	\vspace{-0.05em}
	\begin{subfigure}[b]{0.49\textwidth}
		\includegraphics[width=0.9\textwidth]{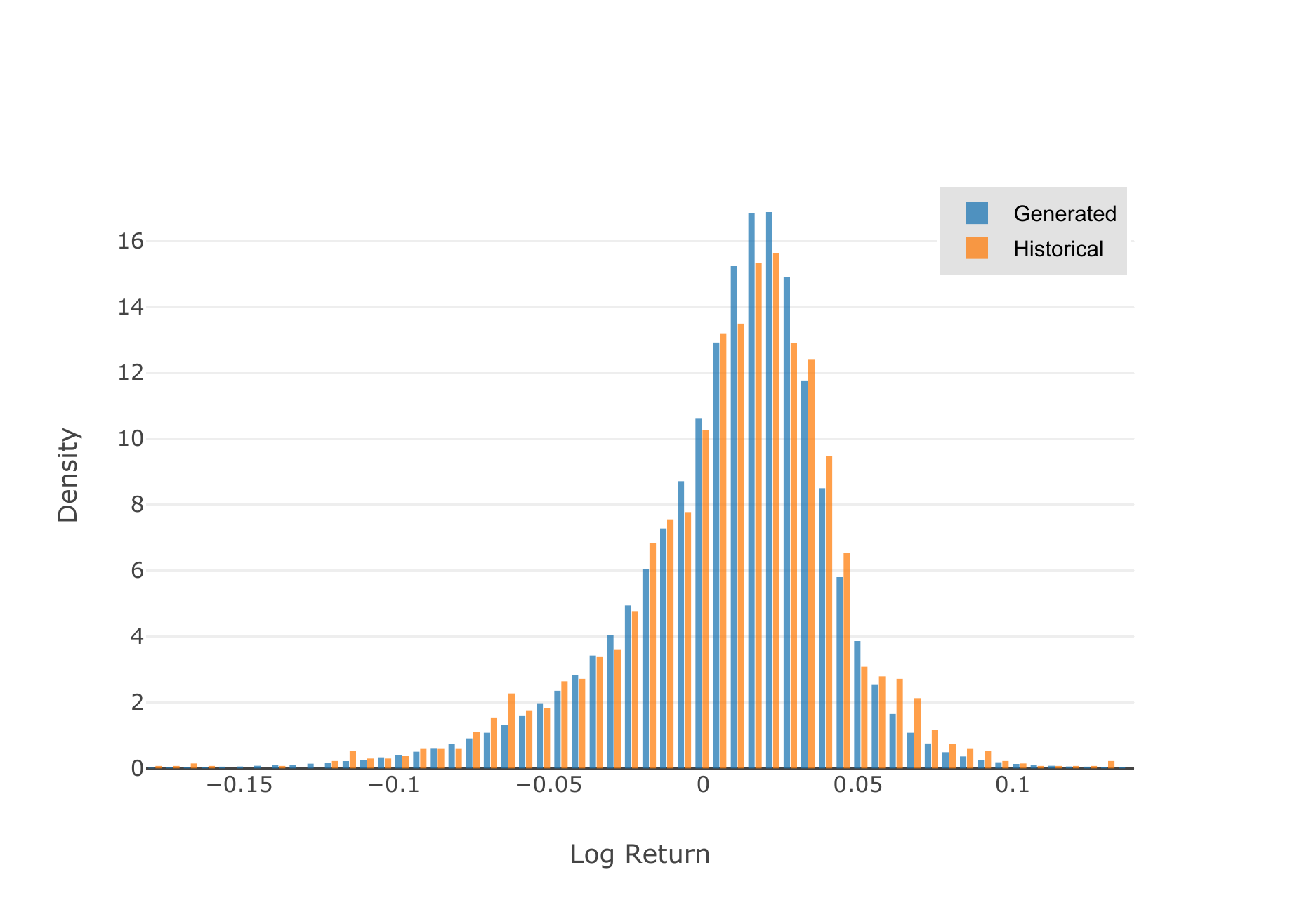}
		\caption{Monthly}
		\label{fig:tcn_hist20}
	\end{subfigure}
	\begin{subfigure}[b]{0.49\textwidth}
		\includegraphics[width=0.9\textwidth]{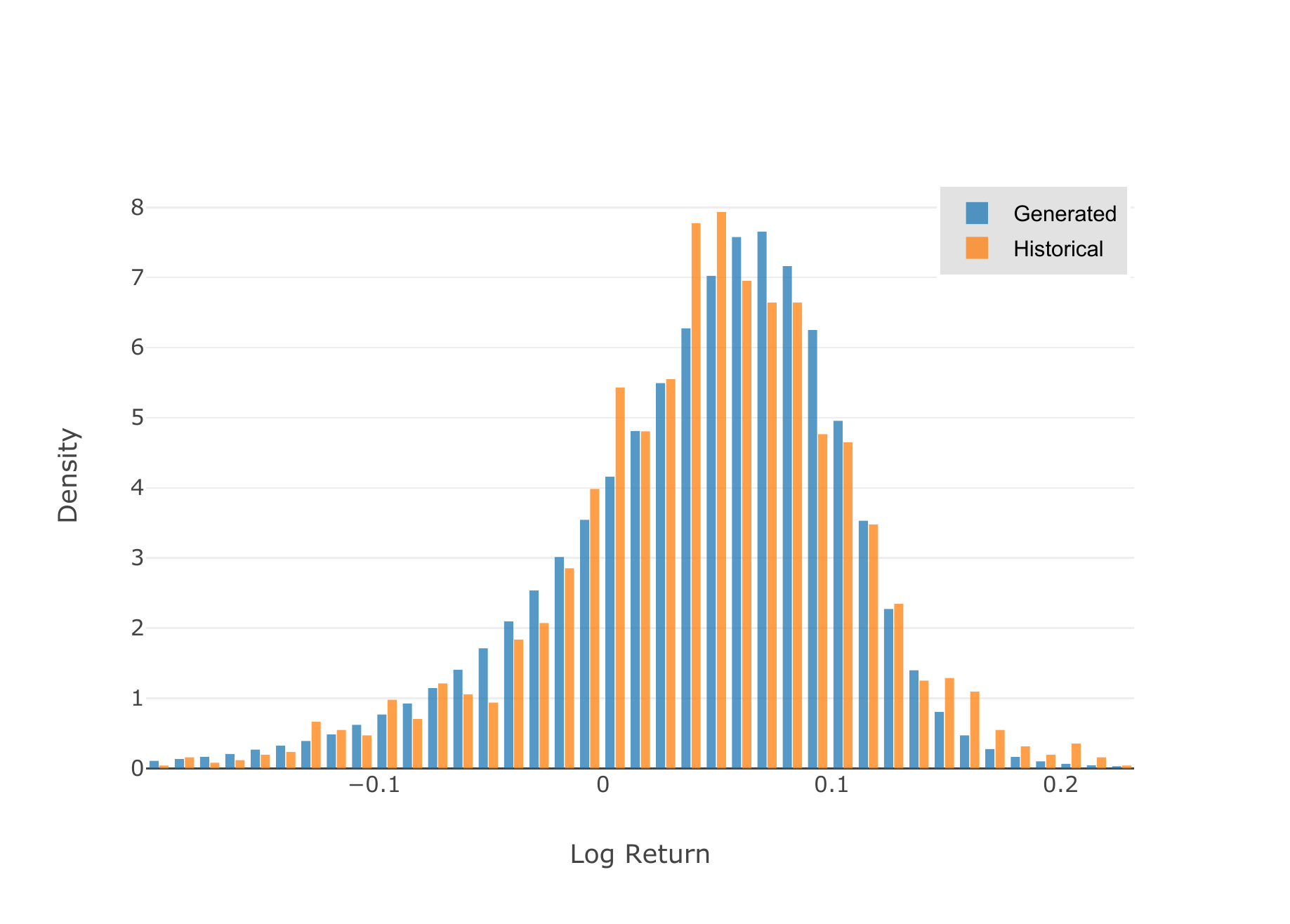}
		\caption{100-day}
		\label{fig:tcn_hist100}
	\end{subfigure}
	\caption{Comparison of generated and historical densities of the S\&P500.}
	\label{fig:s&p_tcn_hist}
\end{figure}
\vspace{-2.25em}
\begin{figure}[H]
	\captionsetup[subfigure]{aboveskip=-2pt,belowskip=-2pt}
	\centering
	\begin{subfigure}[b]{0.49\textwidth}
		\includegraphics[width=0.9\textwidth]{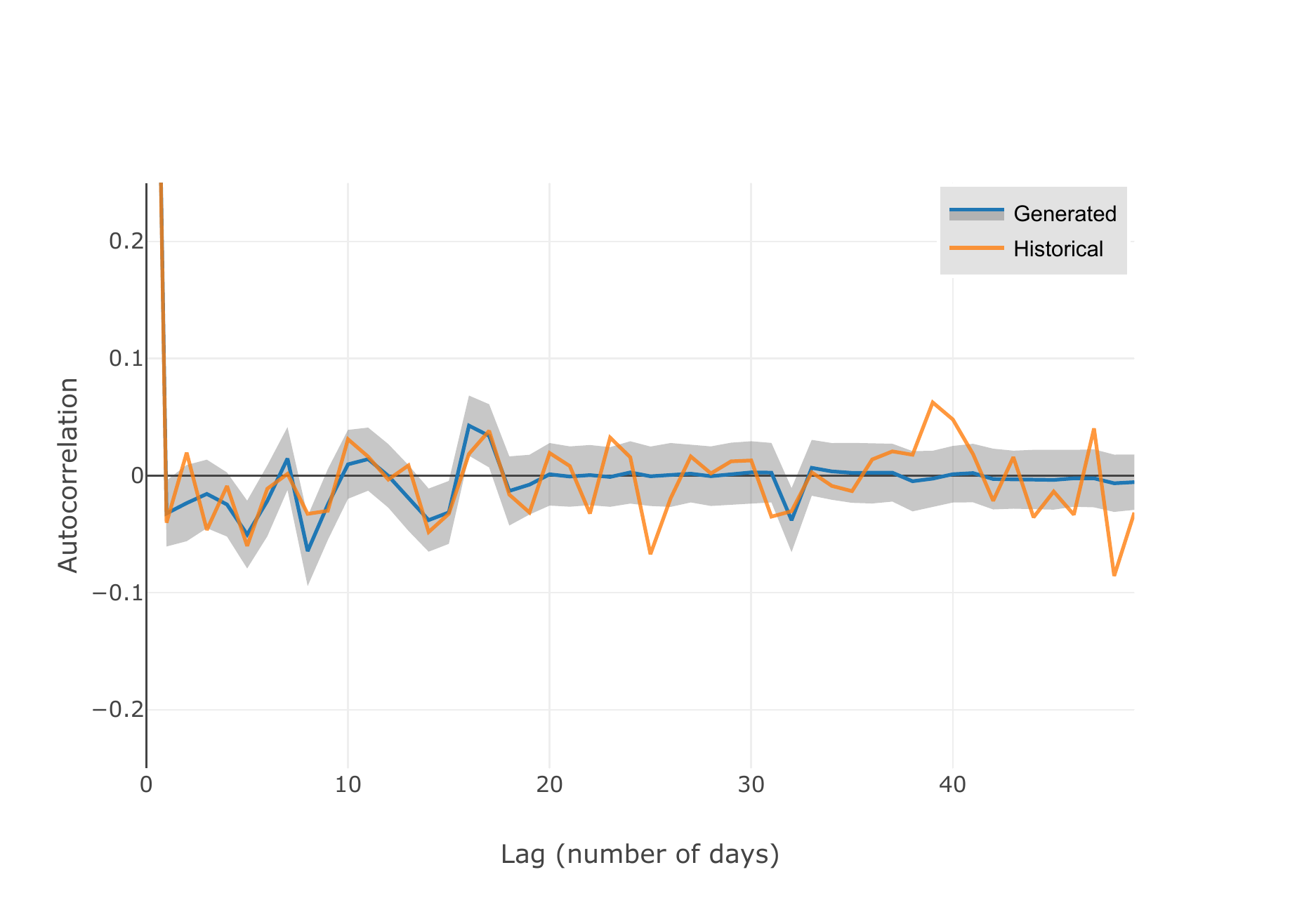}
		\caption{Serial}
		\label{fig:tcn_acf_id}
	\end{subfigure} 
	\begin{subfigure}[b]{0.49\textwidth}
		\includegraphics[width=0.9\textwidth]{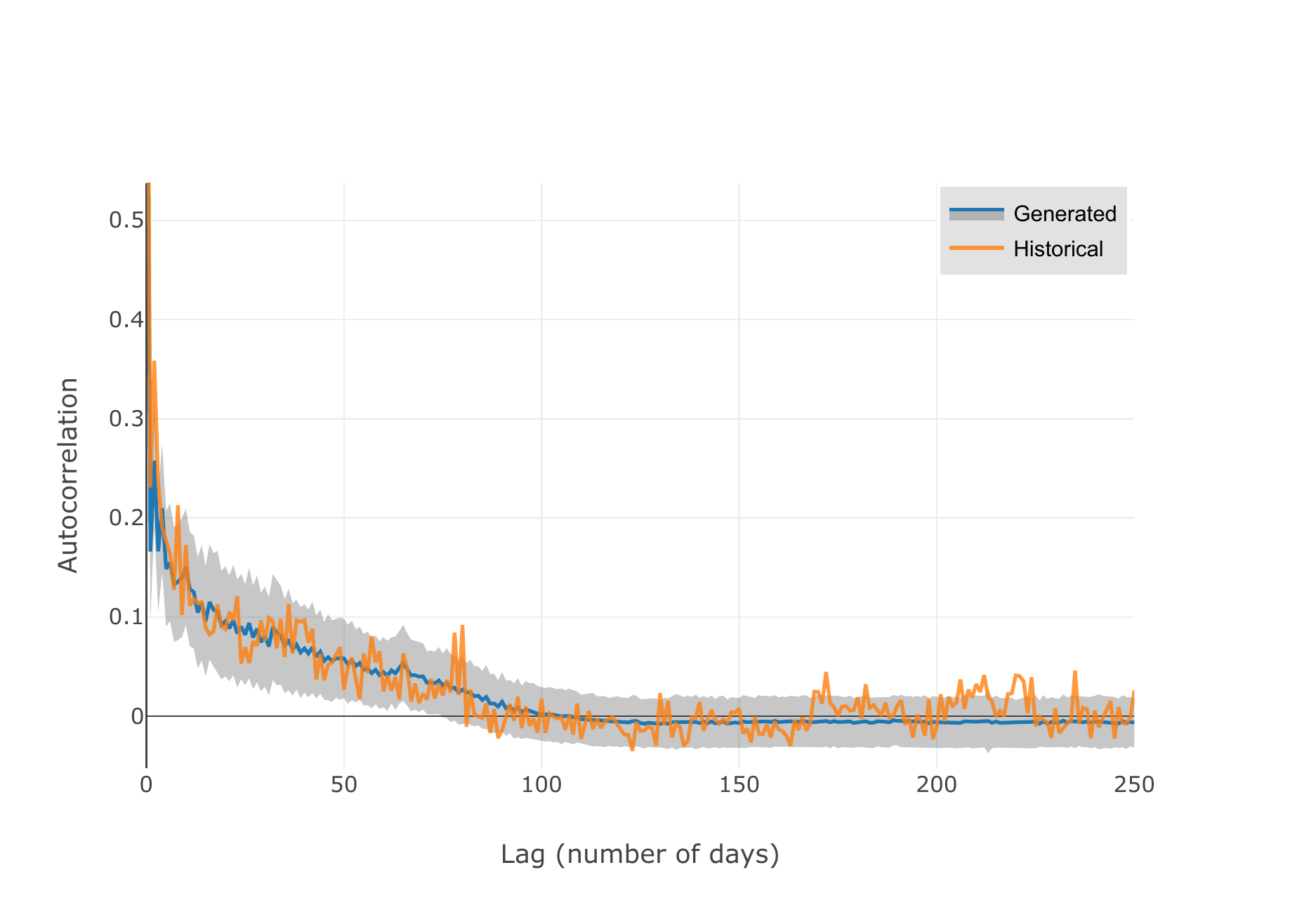}
		\caption{Squared}
		\label{fig:tcn_acf_sq}
	\end{subfigure}
	\\
	\vspace{-0.05em}
	\begin{subfigure}[b]{0.49\textwidth}
		\includegraphics[width=0.9\textwidth]{figures/tcn/acf_absolute_dev.pdf}
		\caption{Absolute}
		\label{fig:tcn_acf_abs}
	\end{subfigure}
	\begin{subfigure}[b]{0.49\textwidth}
		\includegraphics[width=0.9\textwidth]{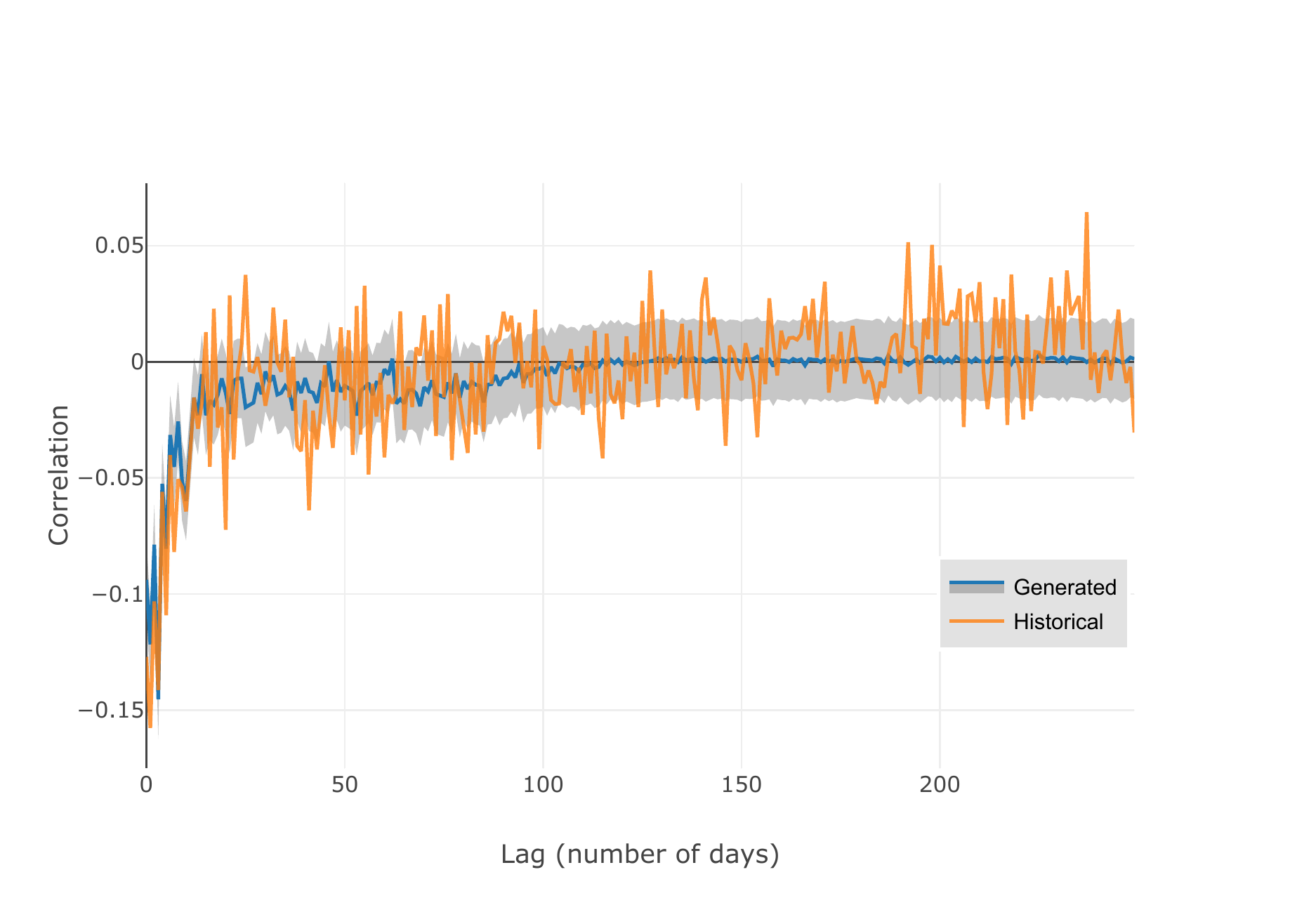}
		\caption{Leverage Effect}
		\label{fig:tcn_lev_eff}
	\end{subfigure}
	\caption{Mean autocorrelation function of the absolute, squared and identical log returns and leverage effect.}
	\label{fig:s&p_tcn_acf}
\end{figure}

\pagebreak
\subsection{Constrained SVNN}

\begin{figure}[H]
	\centering
	\includegraphics[width=0.9\textwidth]{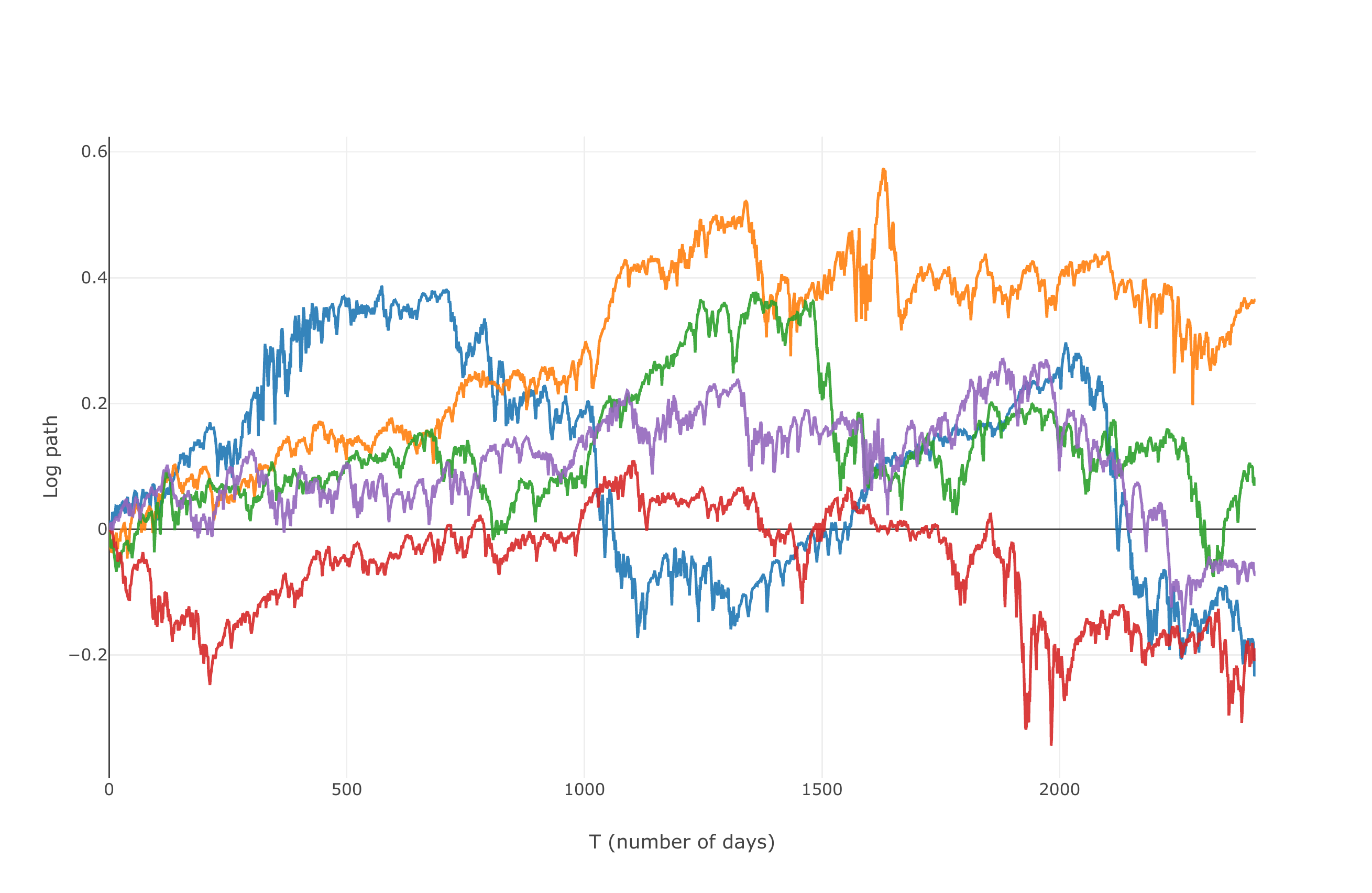}
	\caption{5 generated log paths}
	\label{fig:s&p_c_svnn_drift_5_log_paths}
\end{figure}

\begin{figure}[H]
	\centering
	\includegraphics[width=0.9\textwidth]{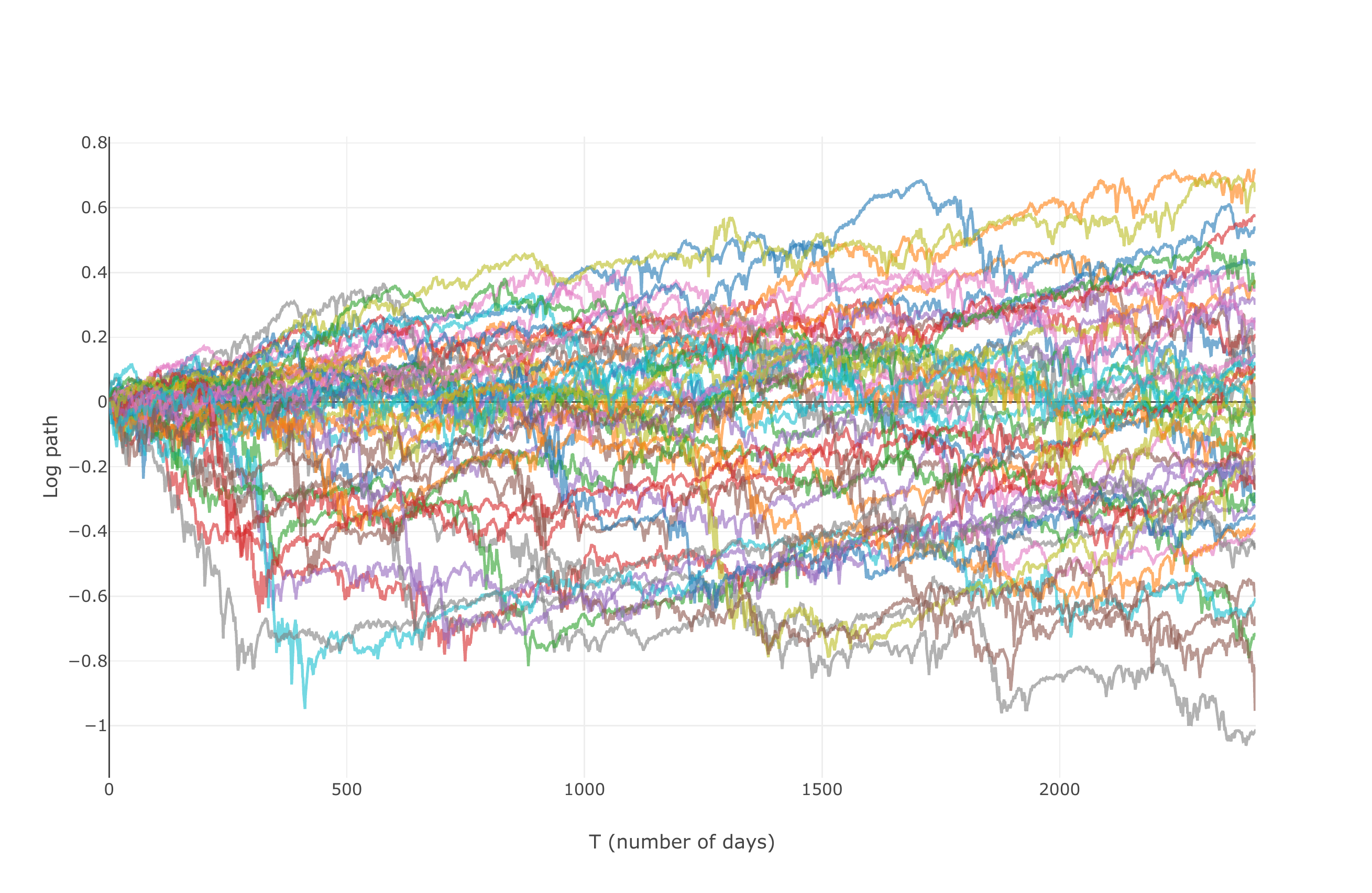}
	\caption{50 generated log paths}
	\label{fig:s&p_c_svnn_drift_log_paths}
\end{figure}

\begin{figure}[H]
	\vspace*{-2em}
	\captionsetup[subfigure]{aboveskip=-2pt,belowskip=-2pt}
	\centering
	\begin{subfigure}[b]{0.49\textwidth}
		\includegraphics[width=0.9\textwidth]{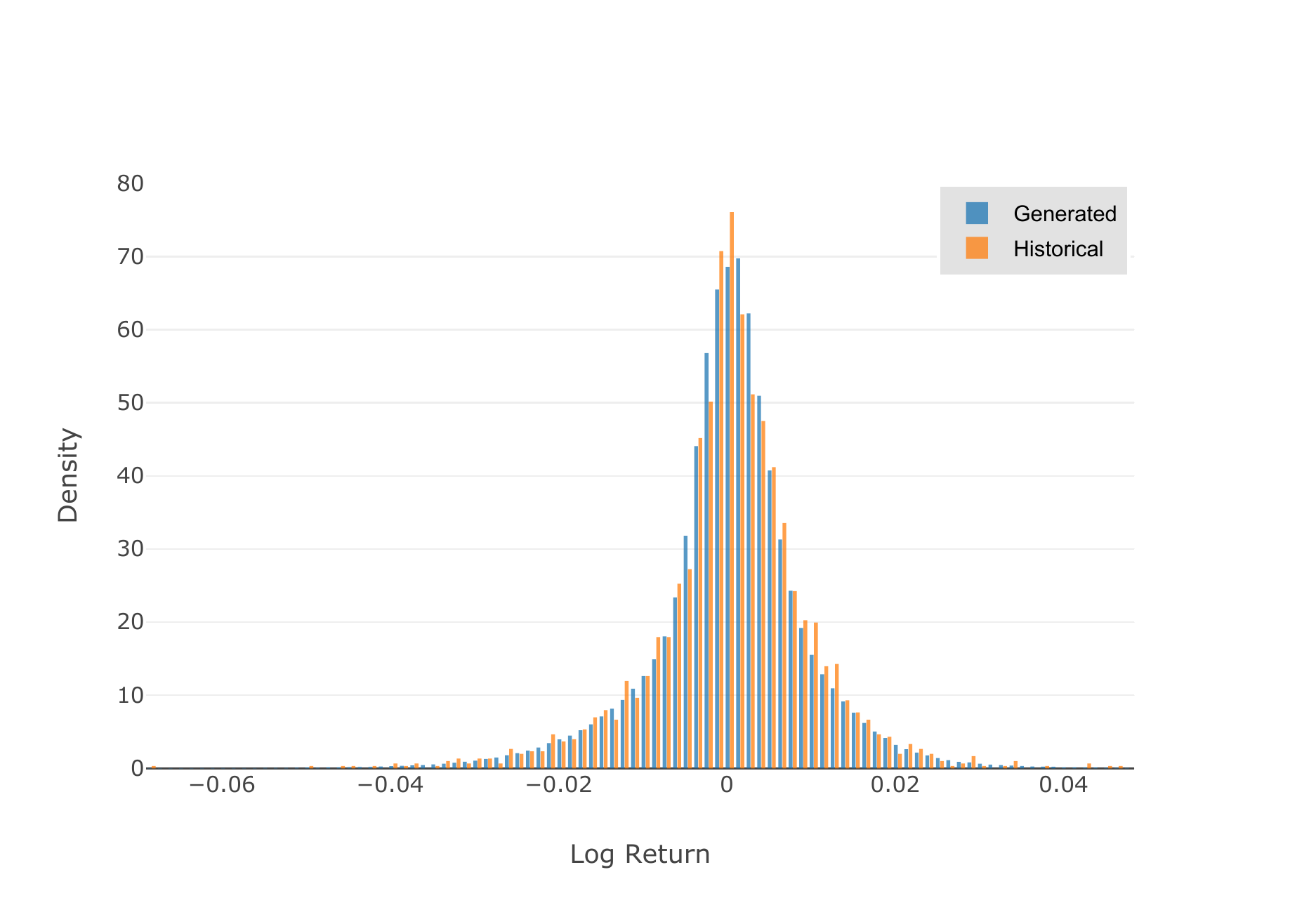}
		\caption{Daily}
		\label{fig:c_svvn_drift_hist1}
	\end{subfigure}
	\begin{subfigure}[b]{0.49\textwidth}
		\includegraphics[width=0.9\textwidth]{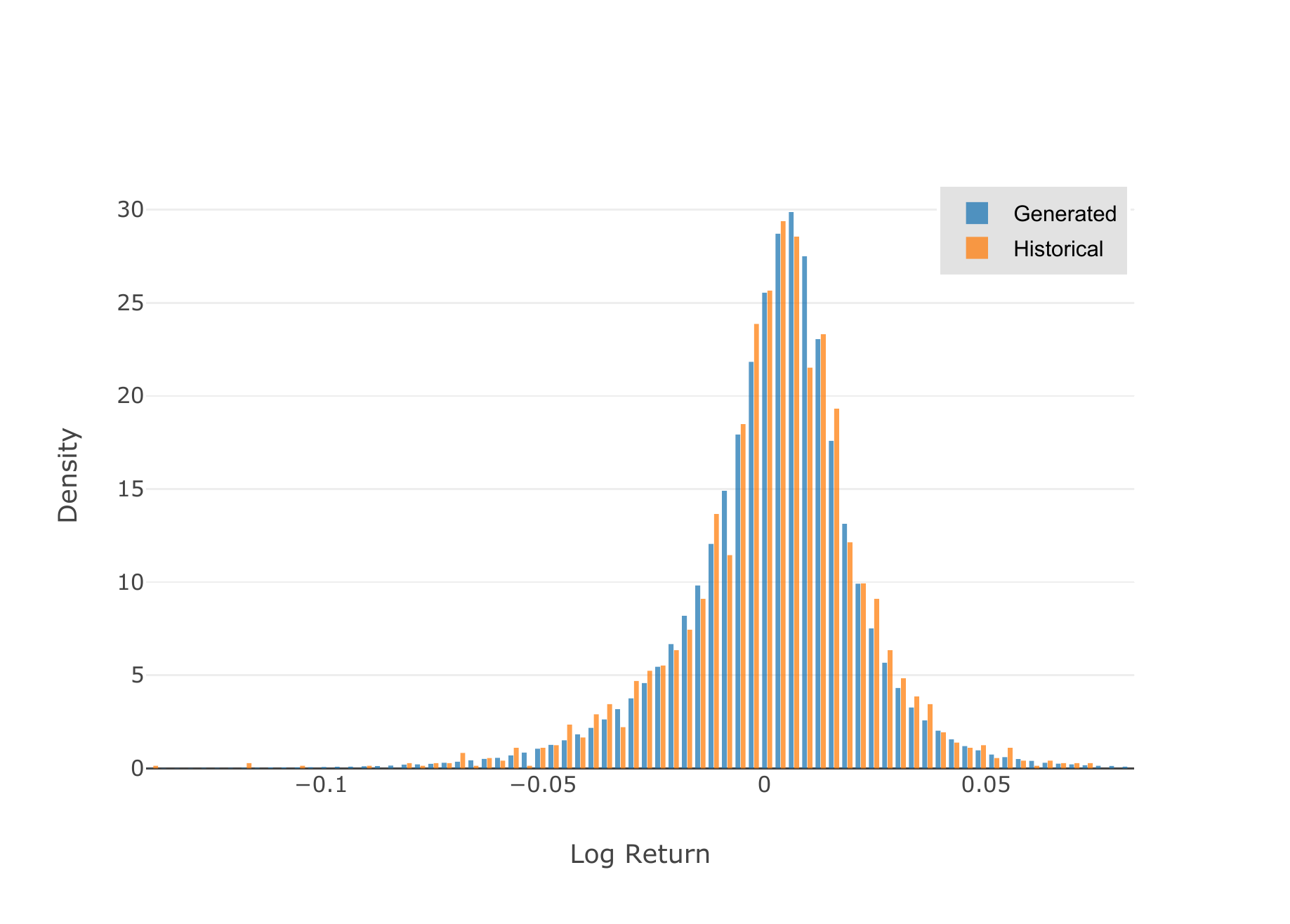}
		\caption{Weekly}
		\label{fig:c_svnn_drift_hist5}
	\end{subfigure}
	\\
	\vspace{-0.05em}
	\begin{subfigure}[b]{0.49\textwidth}
		\includegraphics[width=0.9\textwidth]{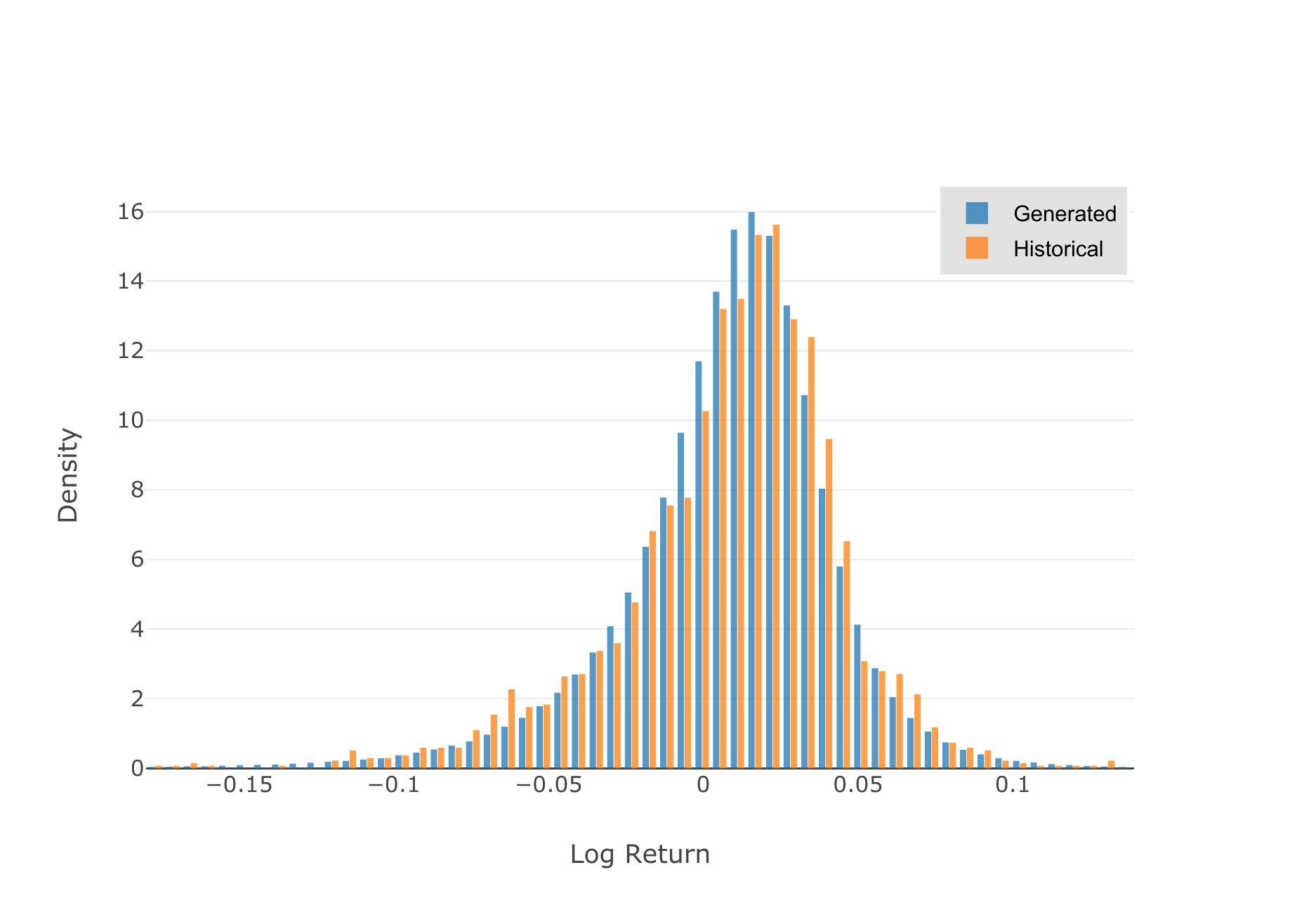}
		\caption{Monthly}
		\label{fig:c_svnn_drift_hist20}
	\end{subfigure}
	\begin{subfigure}[b]{0.49\textwidth}
		\includegraphics[width=0.9\textwidth]{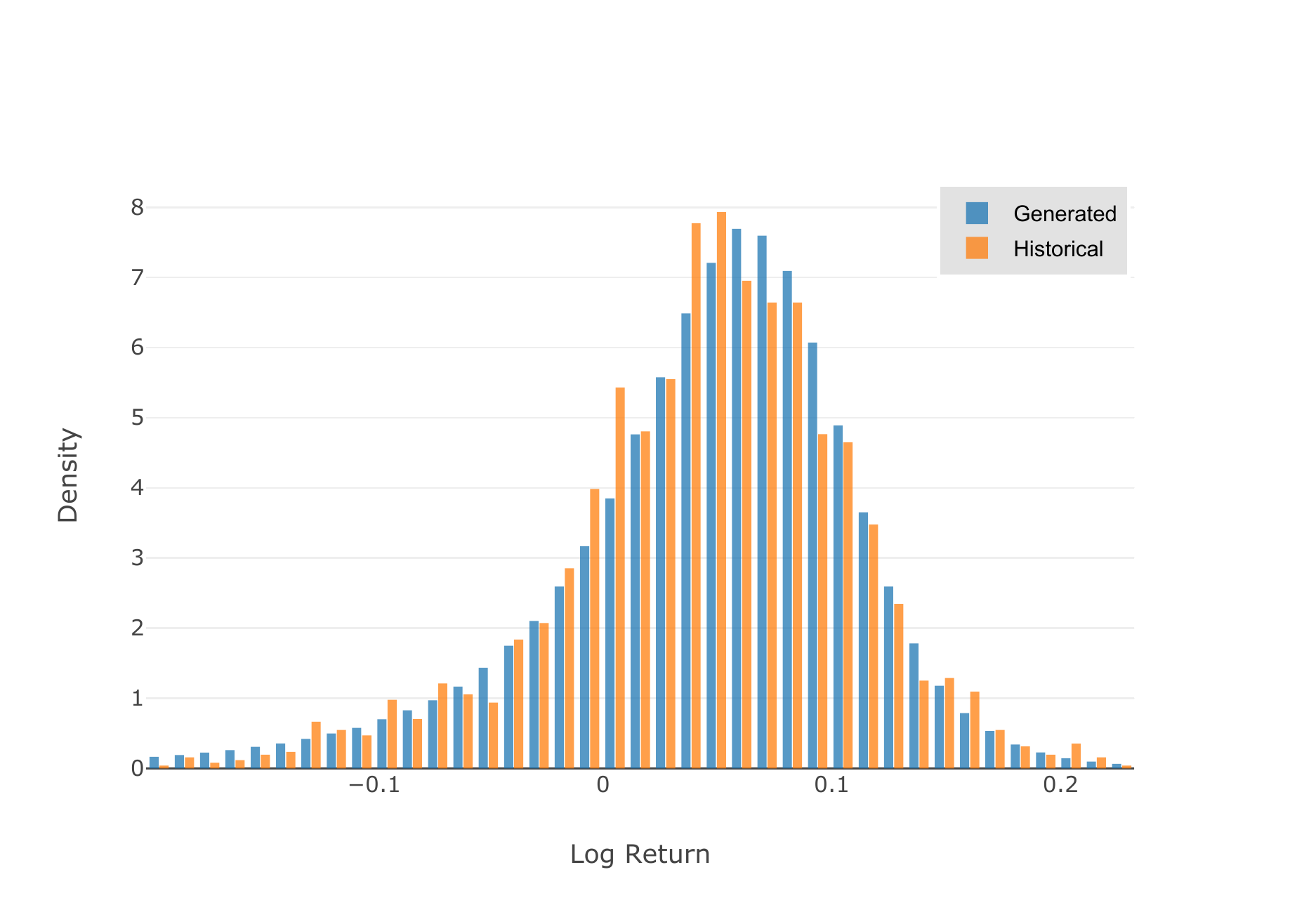}
		\caption{100-day}
		\label{fig:c_svnn_drift_hist100}
	\end{subfigure}
	\caption{Comparison of generated and historical densities of the S\&P500.}
	\label{fig:s&p_c_svnn_drift_hist}
\end{figure}
\vspace{-2.25em}
\begin{figure}[H]
	\captionsetup[subfigure]{aboveskip=-2pt,belowskip=-2pt}
	\centering
	\begin{subfigure}[b]{0.49\textwidth}
		\includegraphics[width=0.9\textwidth]{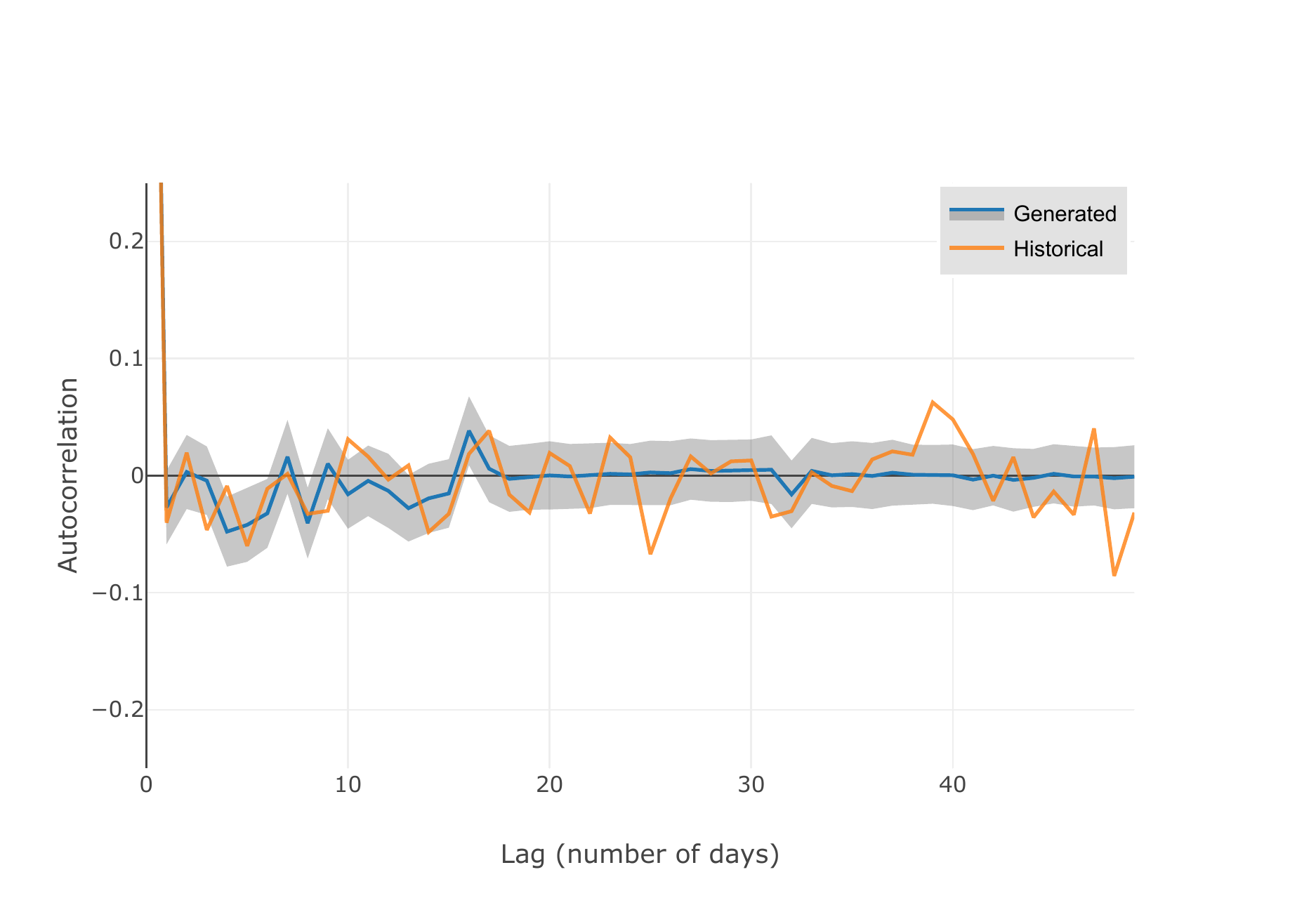}
		\caption{Serial}
		\label{fig:c_svnn_drift_acf_id}
	\end{subfigure}
	\begin{subfigure}[b]{0.49\textwidth}
		\includegraphics[width=0.9\textwidth]{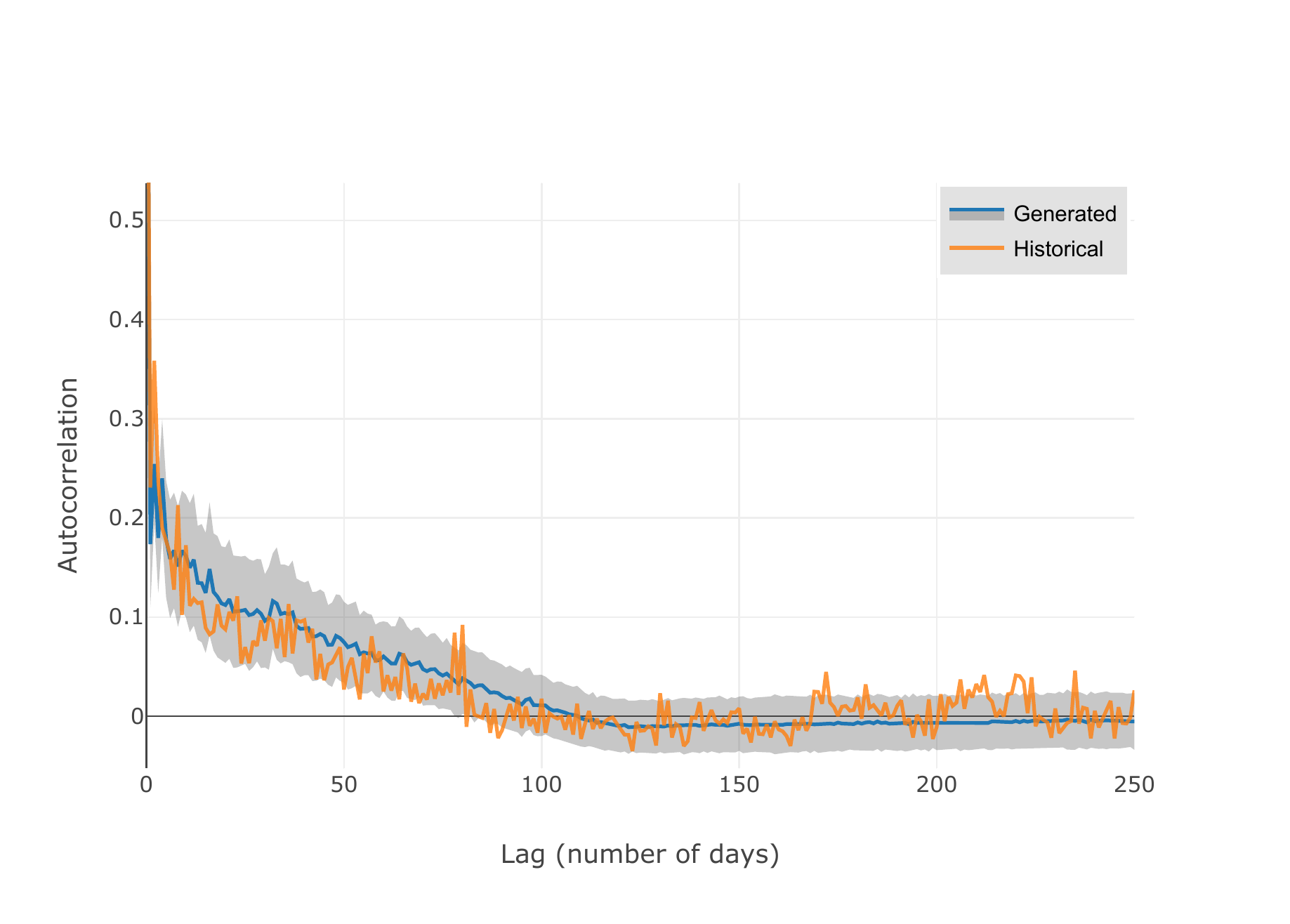}
		\caption{Squared}
		\label{fig:c_svnn_drift_acf_sq}
	\end{subfigure}
	\\
	\vspace{-0.05em}
	\begin{subfigure}[b]{0.49\textwidth}
		\includegraphics[width=0.9\textwidth]{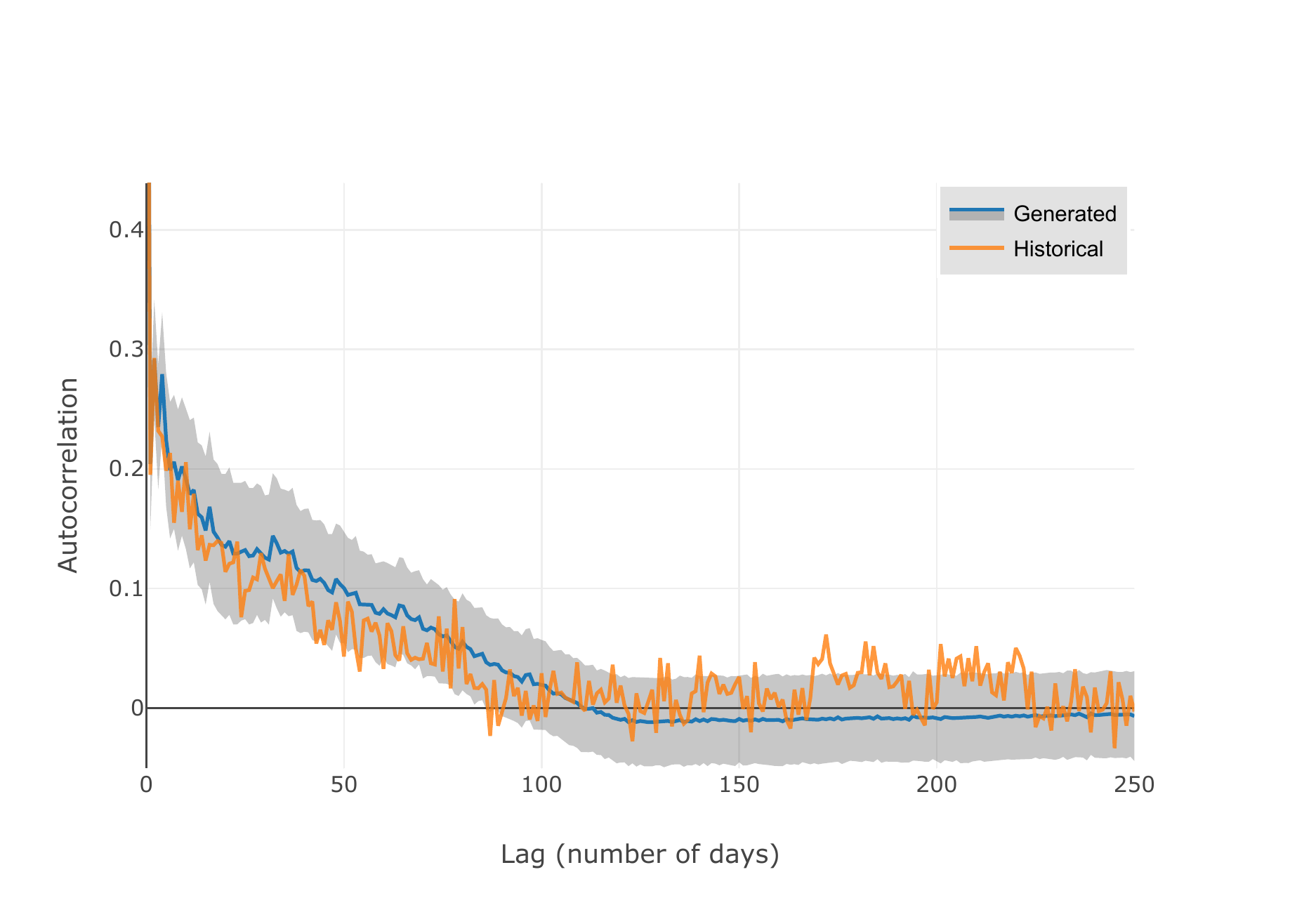}
		\caption{Absolute}
		\label{fig:c_svnn_drift_acf_abs}
	\end{subfigure}
	\begin{subfigure}[b]{0.49\textwidth}
		\includegraphics[width=0.9\textwidth]{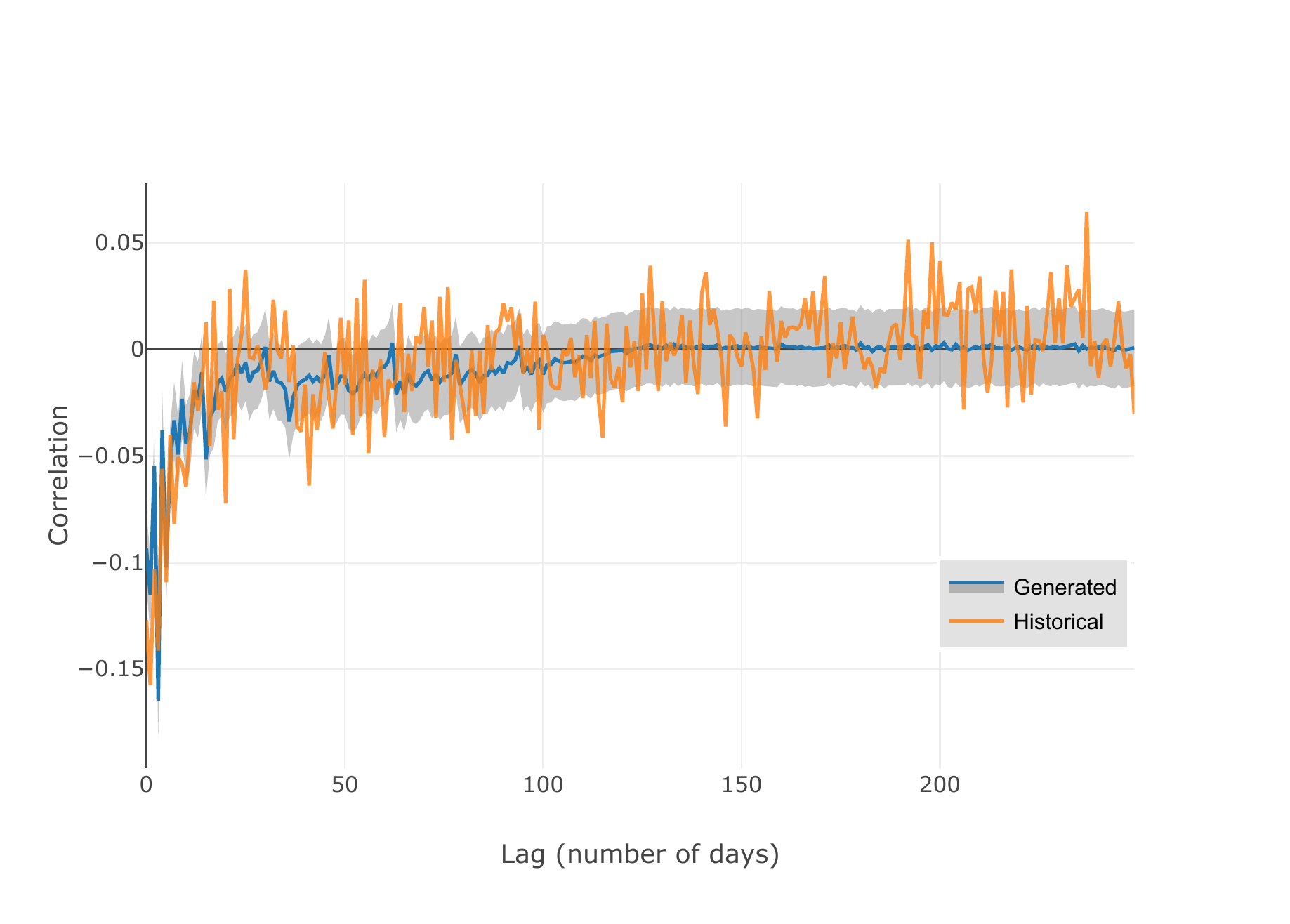}
		\caption{Leverage Effect}
		\label{fig:c_svnn_drift_lev_eff}
	\end{subfigure}
	\caption{Mean autocorrelation function of the absolute, squared and identical log returns and leverage effect.}
	\label{fig:s&p_c_svnn_drift_acf}
\end{figure}

\pagebreak
\subsection{GARCH(1,1) with constant drift}

\begin{figure}[H]
	\centering
	\includegraphics[width=0.9\textwidth]{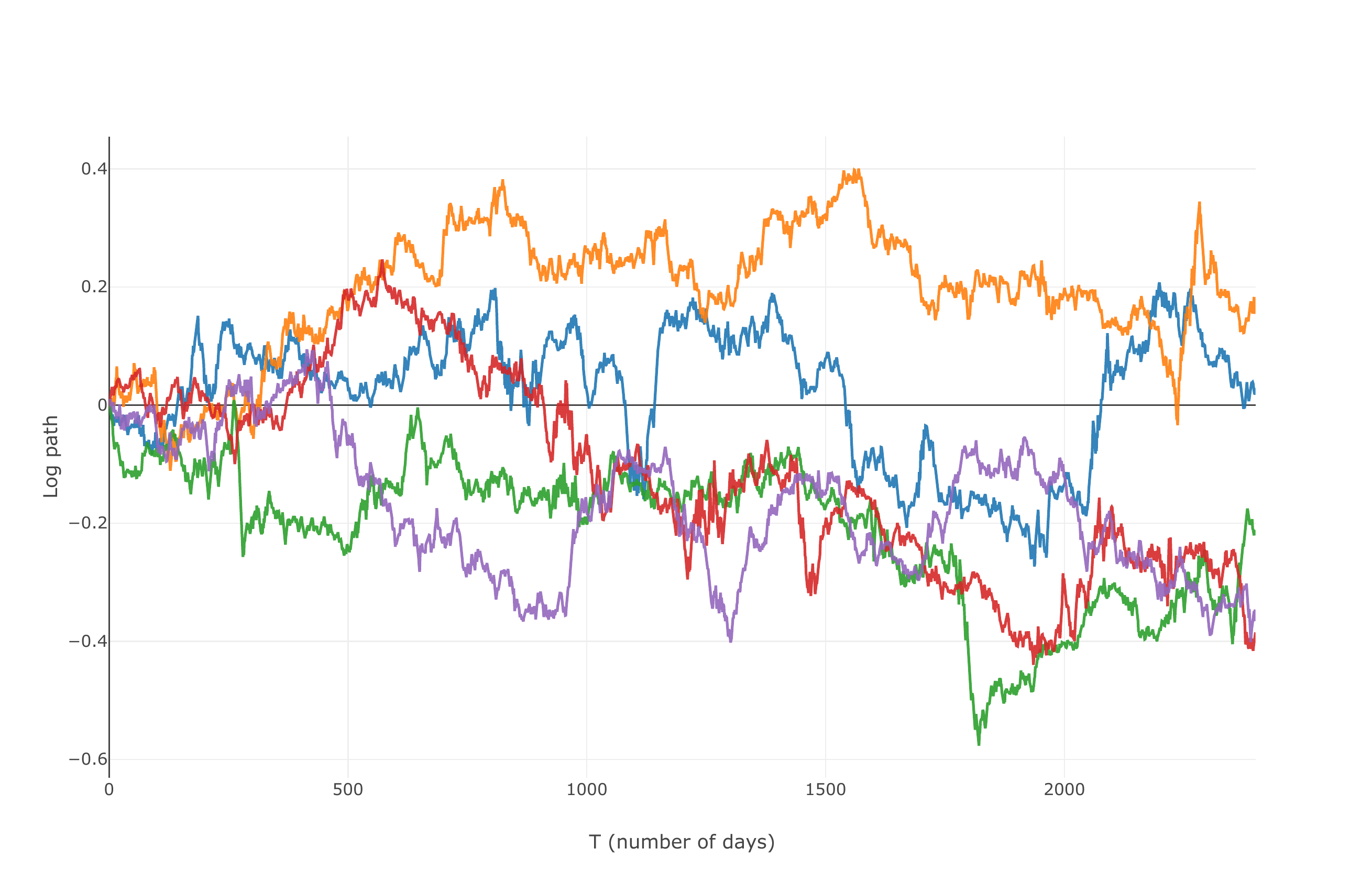}
	\caption{5 generated log paths}
	\label{fig:garch_drift_lev_eff}
\end{figure}

\begin{figure}[H]
	\centering
	\includegraphics[width=0.9\textwidth]{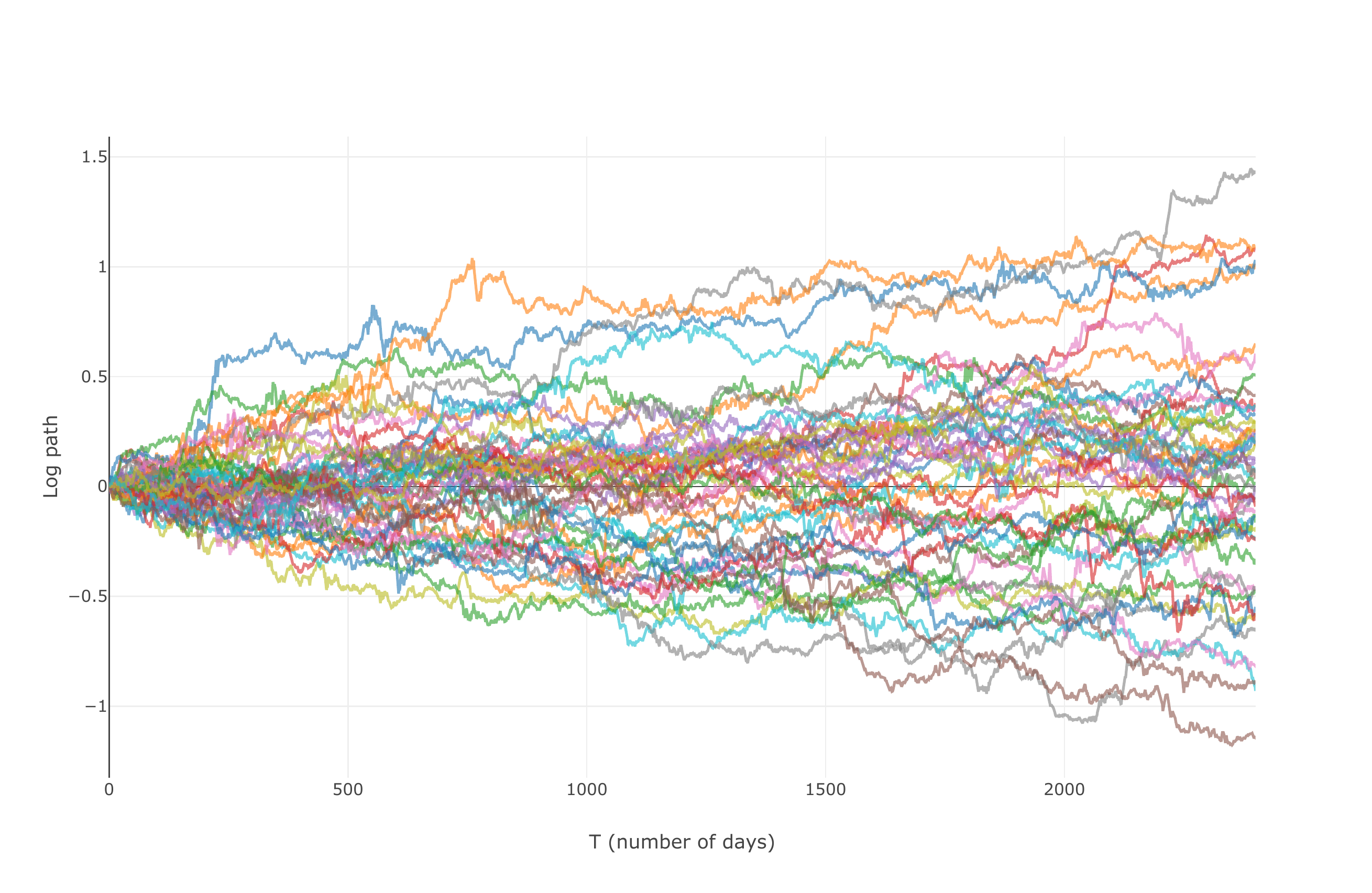}
	\caption{50 generated log paths}
	\label{fig:s&p_garch_drift_log_paths}
\end{figure}

\begin{figure}[H]
	\vspace*{-2em}
	\captionsetup[subfigure]{aboveskip=-2pt,belowskip=-2pt}
	\centering
	\begin{subfigure}[b]{0.49\textwidth}
		\includegraphics[width=0.9\textwidth]{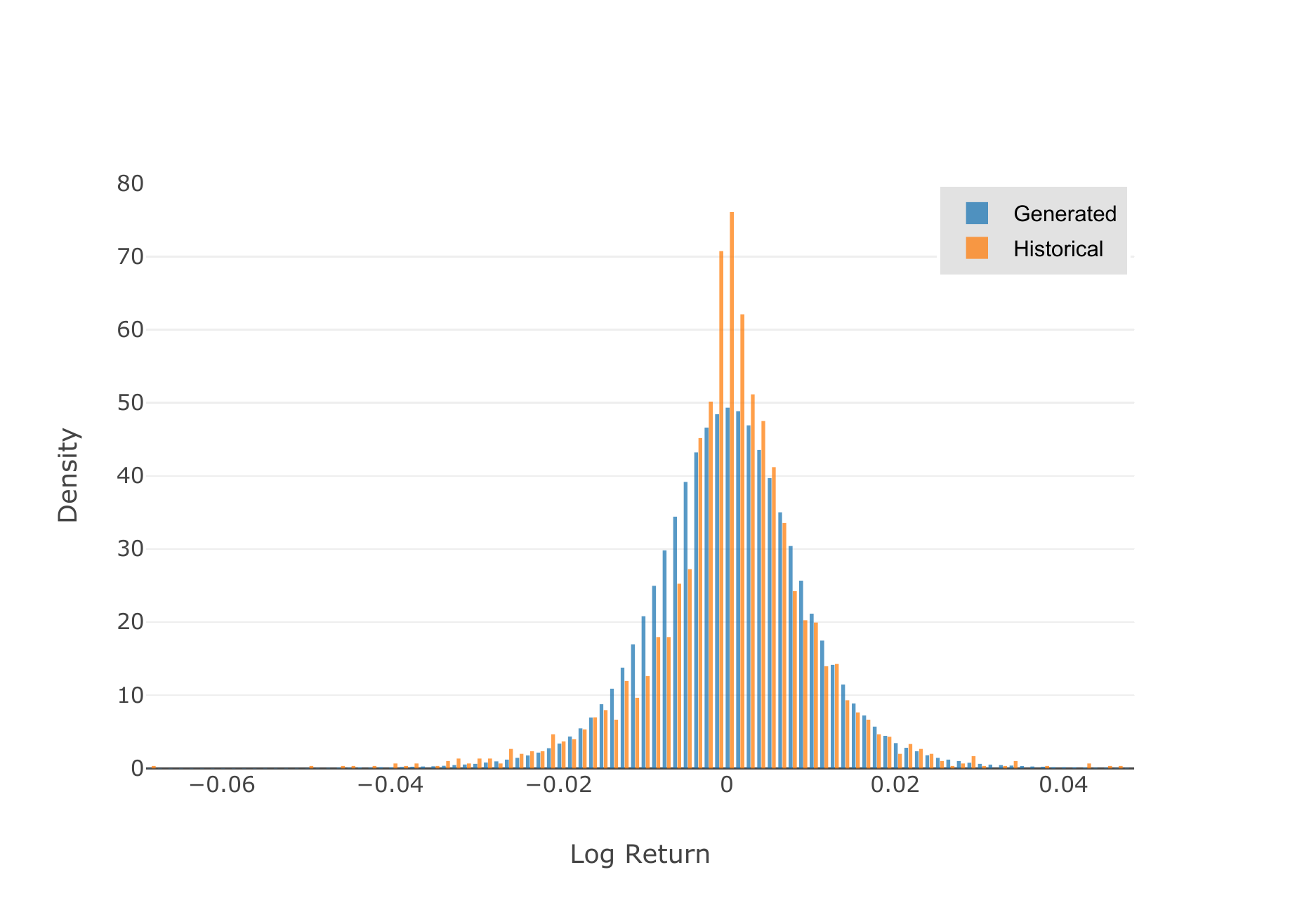}
		\caption{Daily}
		\label{fig:garch_svvn_drift_hist1}
	\end{subfigure}
	\begin{subfigure}[b]{0.49\textwidth}
		\includegraphics[width=0.9\textwidth]{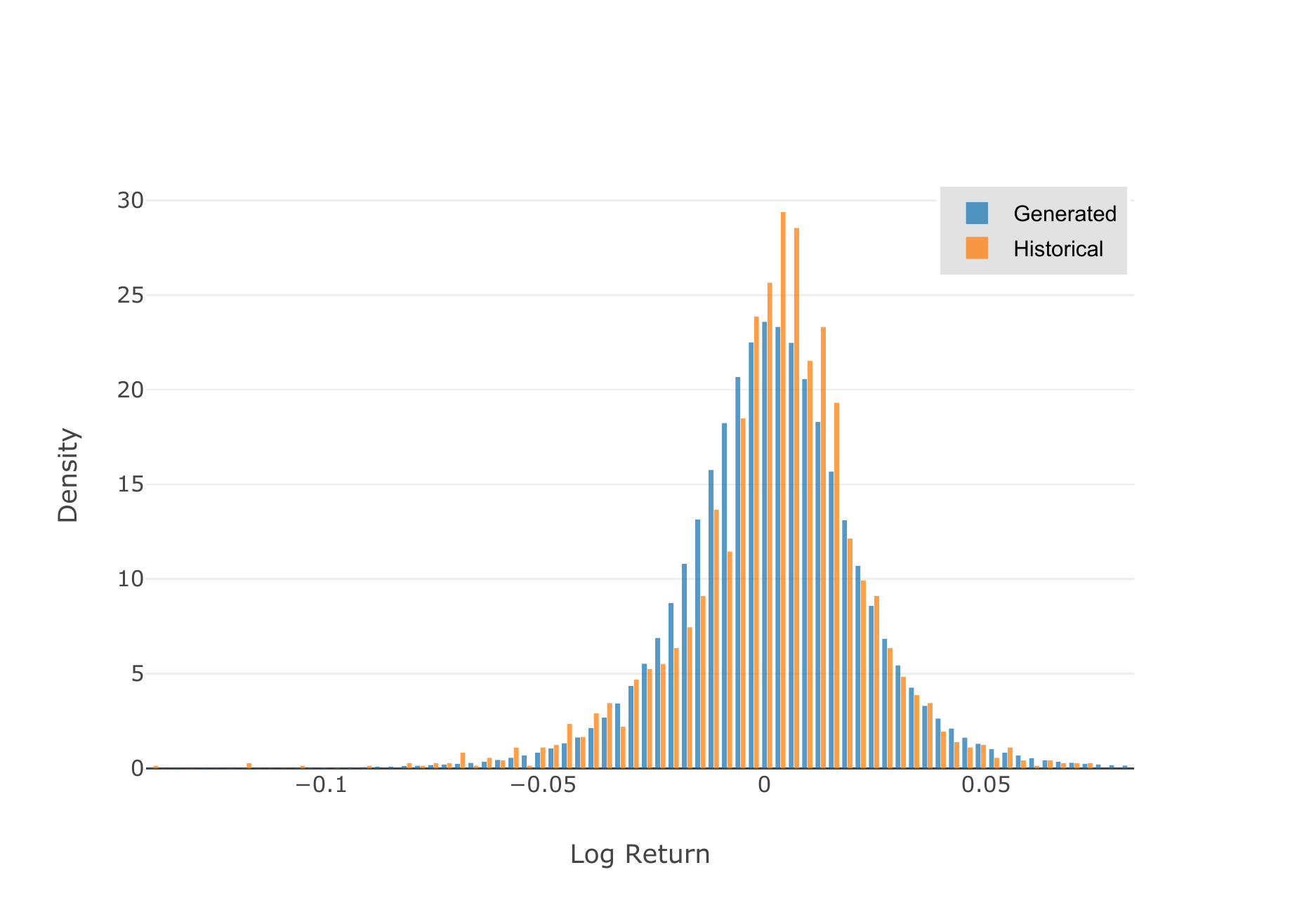}
		\caption{Weekly}
		\label{fig:garch_drift_hist5}
	\end{subfigure}
	\\
	\vspace{-0.05em}
	\begin{subfigure}[b]{0.49\textwidth}
		\includegraphics[width=0.9\textwidth]{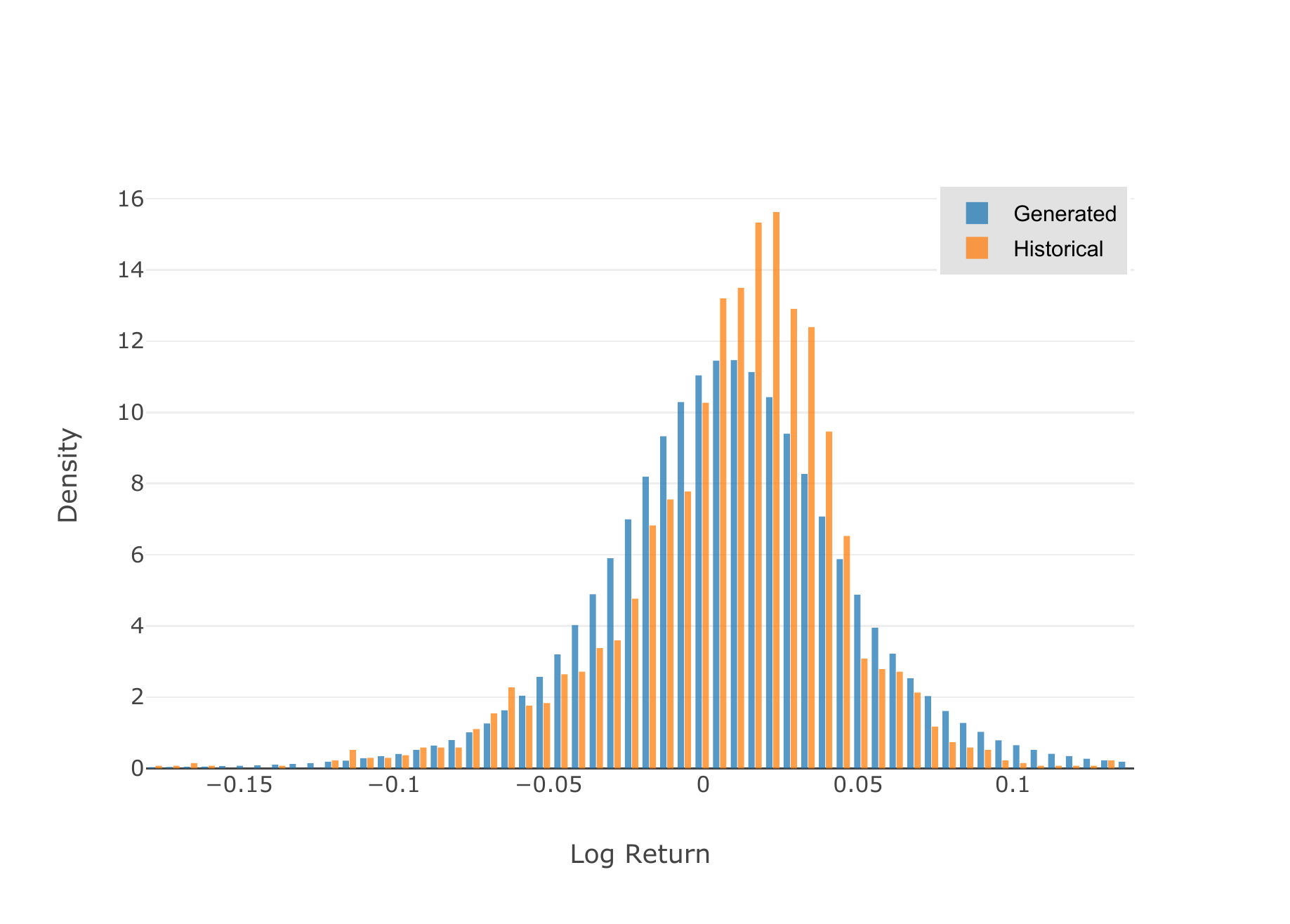}
		\caption{Monthly}
		\label{fig:garch_drift_hist20}
	\end{subfigure}
	\begin{subfigure}[b]{0.49\textwidth}
		\includegraphics[width=0.9\textwidth]{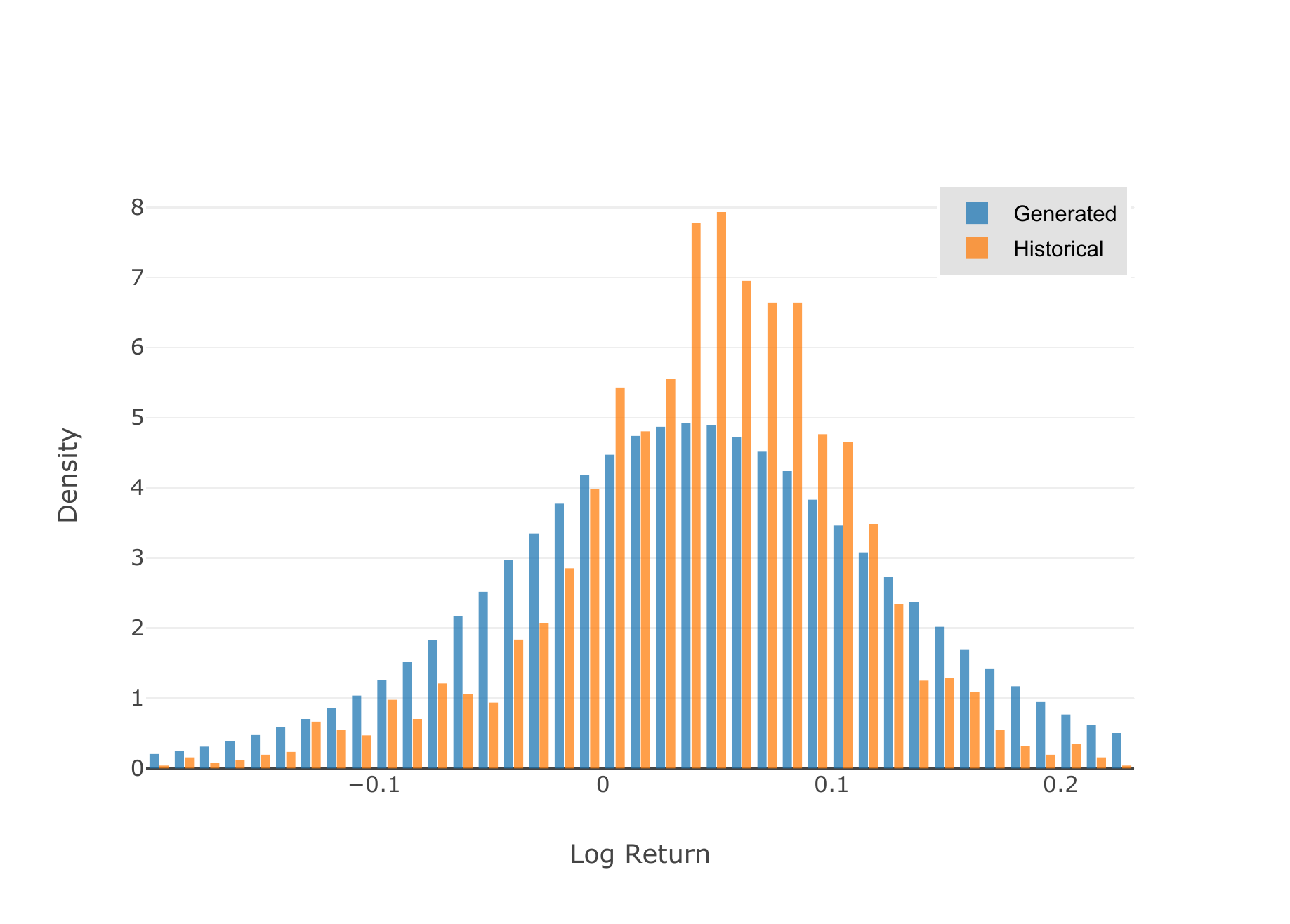}
		\caption{100-day}
		\label{fig:garch_drift_hist100}
	\end{subfigure}
	\caption{Comparison of generated and historical densities of the S\&P500.}
	\label{fig:s&p_garch_drift_hist}
\end{figure}
\vspace{-2.25em}
\begin{figure}[H]
	\captionsetup[subfigure]{aboveskip=-2pt,belowskip=-2pt}
	\centering
	\begin{subfigure}[b]{0.49\textwidth}
		\includegraphics[width=0.9\textwidth]{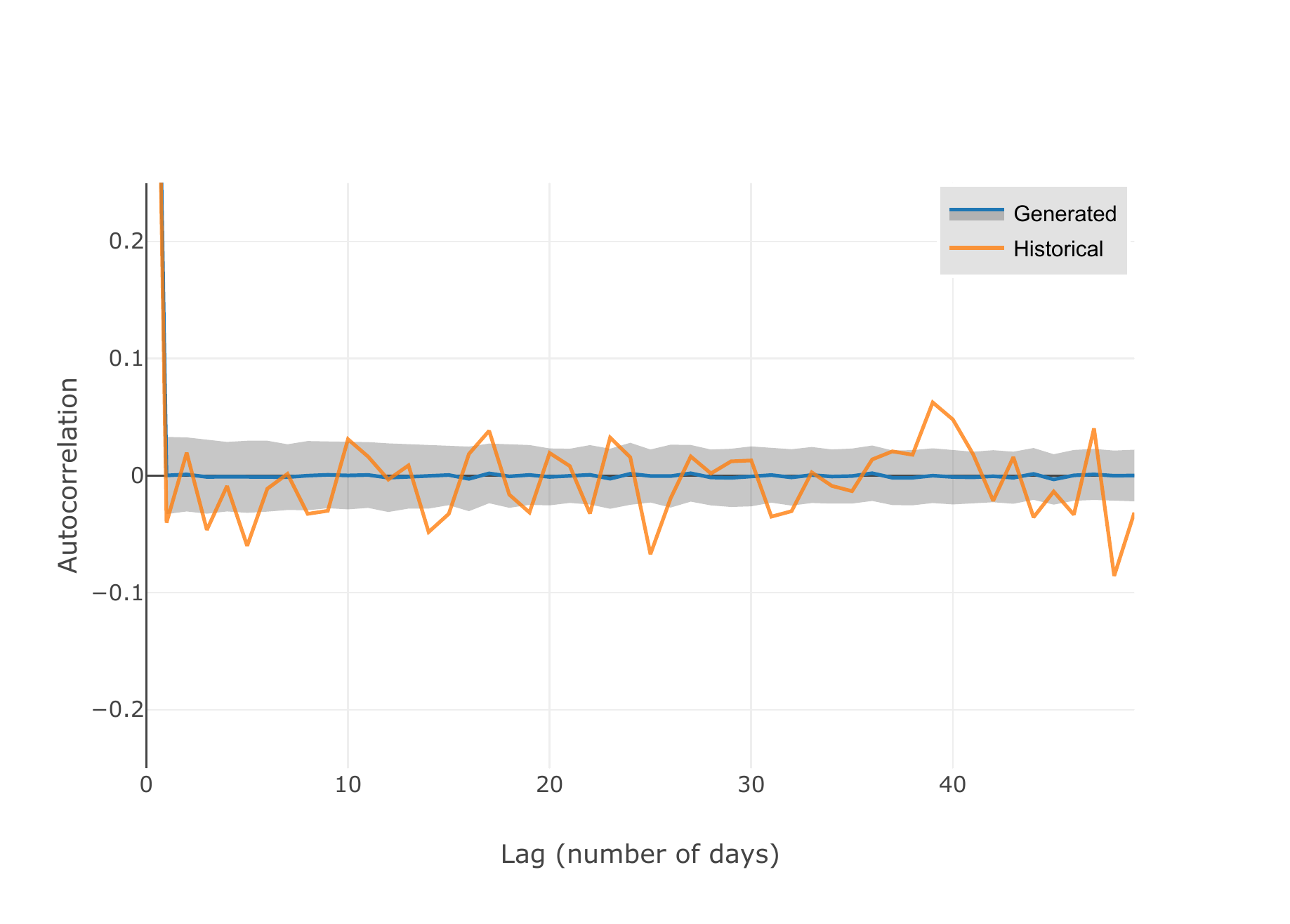}
		\caption{Serial}
		\label{fig:garch_drift_acf_id}
	\end{subfigure}
	\begin{subfigure}[b]{0.49\textwidth}
		\includegraphics[width=0.9\textwidth]{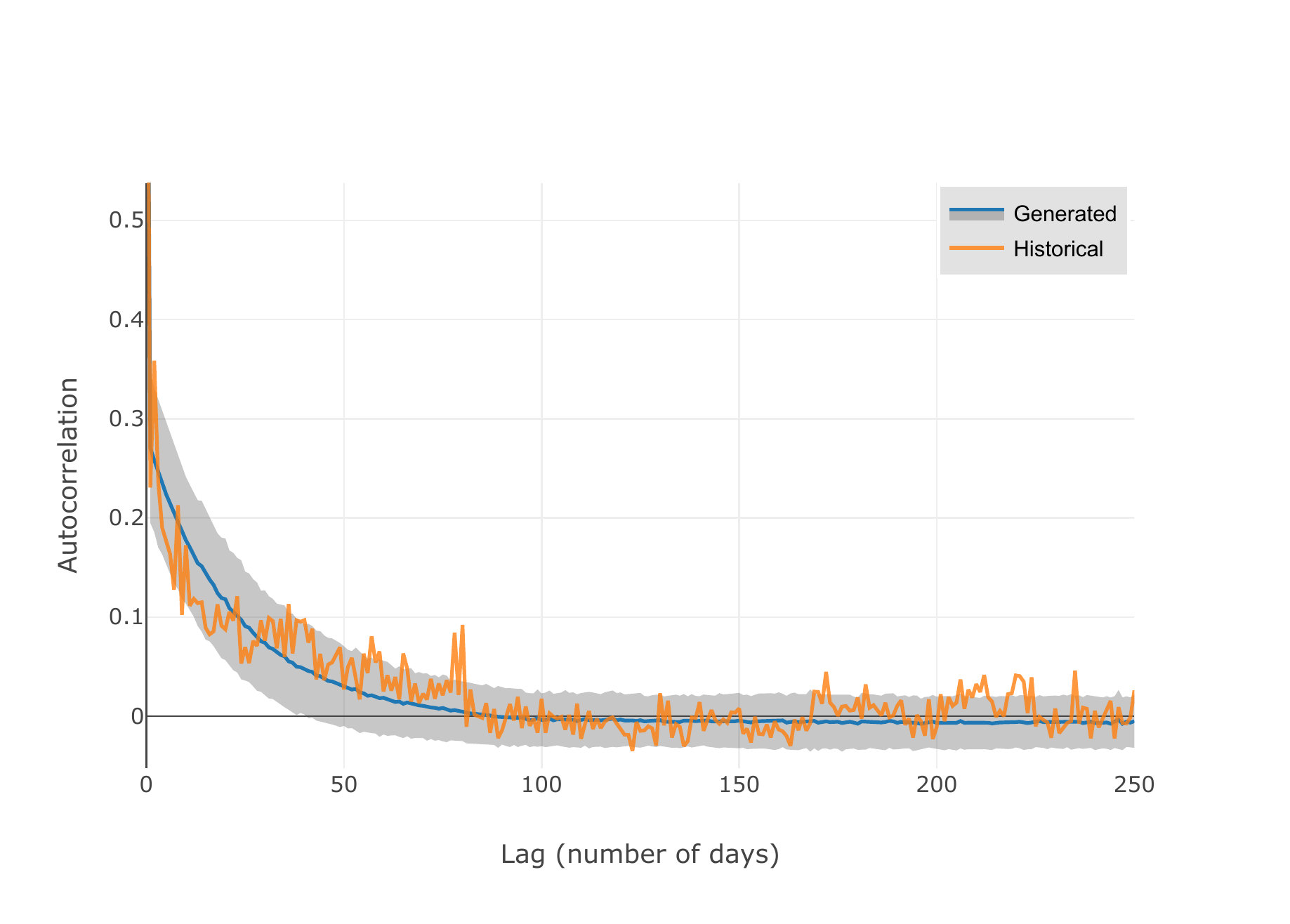}
		\caption{Squared}
		\label{fig:garch_drift_acf_sq}
	\end{subfigure}
	\\
	\vspace{-0.05em}
	\begin{subfigure}[b]{0.49\textwidth}
		\includegraphics[width=0.9\textwidth]{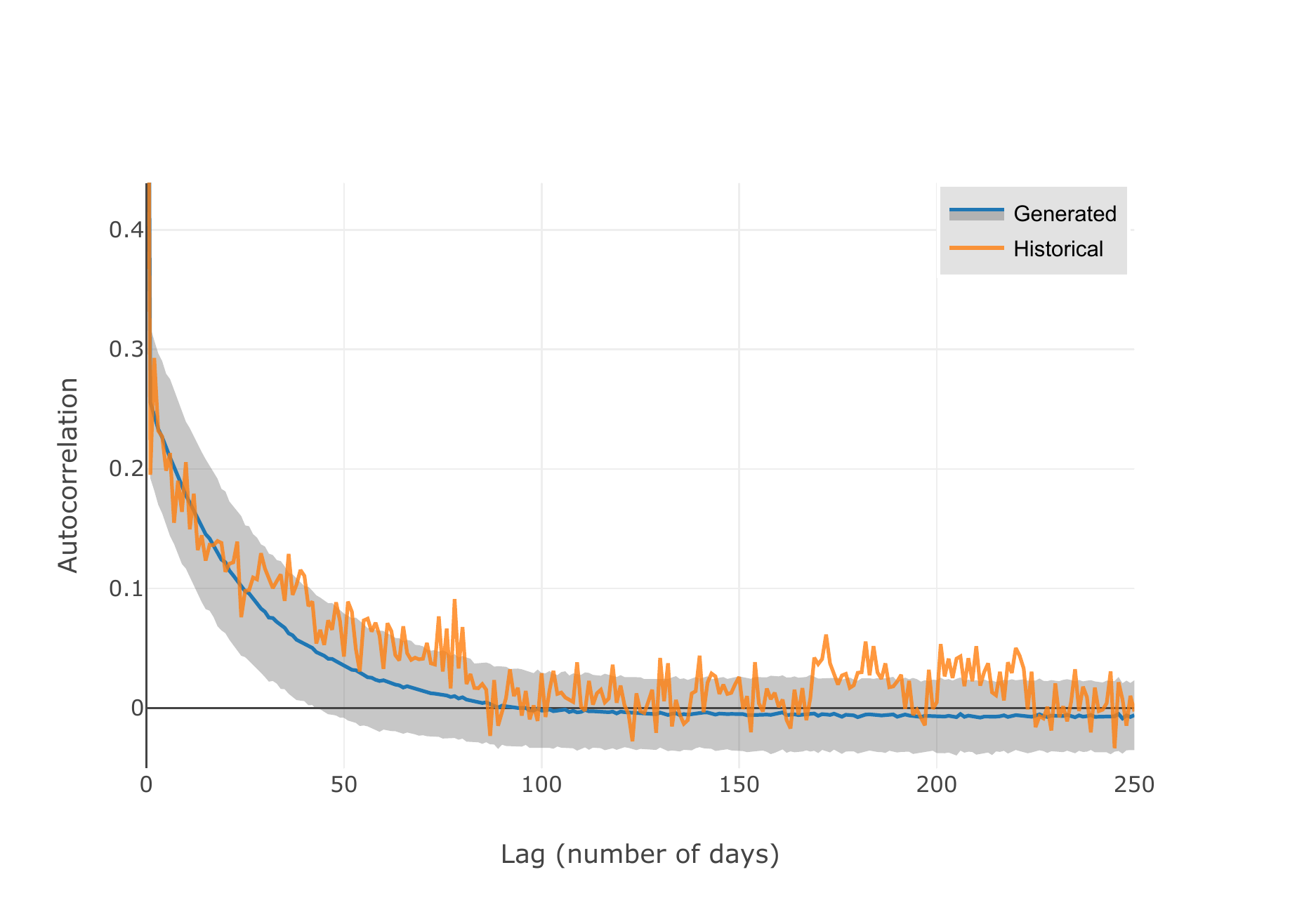}
		\caption{Absolute}
		\label{fig:garch_drift_acf_abs}
	\end{subfigure}
	\begin{subfigure}[b]{0.49\textwidth}
		\includegraphics[width=0.9\textwidth]{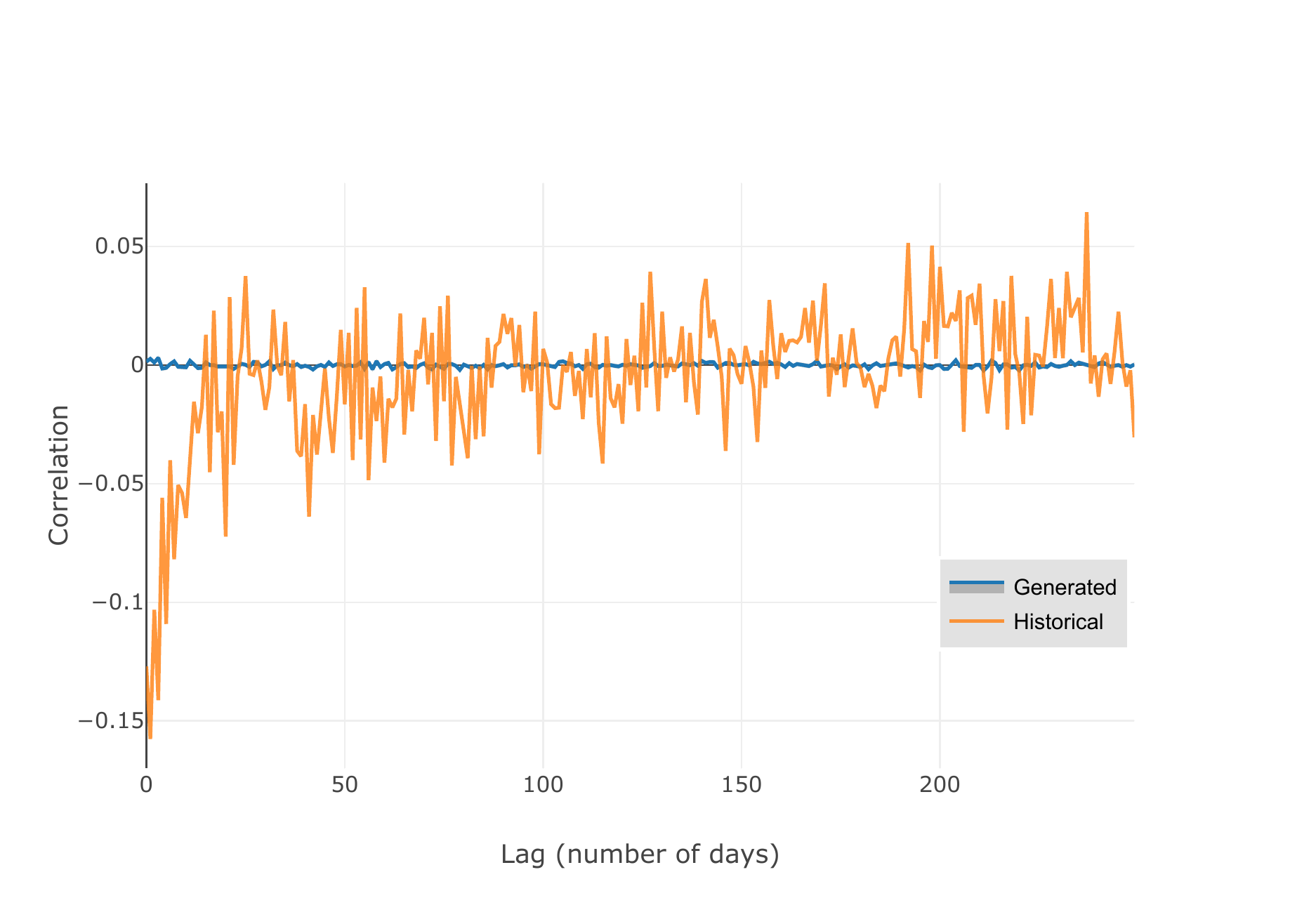}
		\caption{Leverage Effect}
		\label{fig:garch_drift_lev_eff}
	\end{subfigure}
	\caption{Mean autocorrelation function of the absolute, squared and identical log returns and leverage effect.}
	\label{fig:s&p_garch_drift_acf}
\end{figure}

\newpage
\section{Architecture}
\label{appendix_architecture}
For both the generator and the discriminator we used TCNs with skip connections. Inside the TCN architecture \emph{temporal blocks} were used as block modules. A temporal block consists of two dilated causal convolutions and two PReLUs \parencite{prelu_he} as activation functions. The primary benefit of using temporal blocks is to make the TCN more expressive by increasing the number of non-linear operations in each block module. A complete definition is given below. 
\begin{definition}[Temporal block]
	Let $N_I, N_H, N_O \in \mathbb{N}$ denote the input, hidden and output dimension and let $D, K \in \mathbb{N}$ denote the dilation and the kernel size. Furthermore, let $w_1, w_2$ be two dilated causal convolutional layers with arguments $(N_I, N_H, K, D)$ and $(N_H, N_O, K, D)$ respectively and let $\phi_1, \phi_2: \mathbb{R} \to \mathbb{R} $ be two PReLUs. 
	The function $f: \mathbb{R}^{N_I \times (2D(K-1)+1)} \to \mathbb{R}^{N_O}$ defined by 
	\begin{equation*}
	f(X) = \phi_2 \circ w_2 \circ \phi_1 \circ w_1(X)
	\end{equation*}
	is called \emph{temporal block} with arguments $(N_I, N_H, N_O, K, D)$.
\end{definition}
The TCN architecture used for the generator and the discriminator in the pure TCN and C-SVNN model is illustrated in \autoref{tab:used_arch}. \autoref{tab:tcn_parameters} shows the input, hidden and output dimensions of the different models. Here, G abbreviates the generator and D the discriminator.  Note that for all models, except the generator of the C-SVNN, the hidden dimension was set to eighty. The kernel size of each temporal block, except the first one, was two. Each TCN modeled a RFS of 127. 
\begin{table}[htp]
	\centering
	\begin{tabular}{|l|l|}
		\hline
		Module name & Arguments\\
		\hline
		Temporal block 1 & $(N_I, \ N_H, \ N_H, \ 1, \ 1)$ \\
		Temporal block 2 & $(N_H, \ N_H, \ N_H, \ 2, \ 1)$ \\
		Temporal block 3 & $(N_H, \ N_H, \ N_H, \ 2, \ 2)$ \\
		Temporal block 4 & $(N_H, \ N_H, \ N_H, \ 2, \ 4)$ \\
		Temporal block 5 & $(N_H, \ N_H, \ N_H, \ 2, \ 8)$ \\
		Temporal block 6 & $(N_H, \ N_H, \ N_H, \ 2, \ 16)$ \\
		Temporal block 7 & $(N_H, \ N_H, \ N_H, \ 2, \ 32)$ \\
		$1 \times 1$ Convolution    & $(N_H, \ N_O, \ 1, \ 1) $\\
		\hline
	\end{tabular}
	\caption{TCN architecture of the reported models. Note that the TCN architecture includes skip connections. }
	\label{tab:used_arch}
\end{table}

\begin{table}[htp]
	\centering
	\begin{tabular}{|c|c|c|c|c|}
		\hline
		Models & Pure TCN - G & Pure TCN - D & C-SVNN - G & C-SVNN - D \\ \hline
		$N_I$     & $3$          & $1$          & $3$        & $1$        \\ 
		$N_H$     & $80$         & $80$         & $50$       & $80$       \\ 
		$N_O$     & $1$          & $1$          & $2$        & $1$        \\ \hline
	\end{tabular}
	\caption{Configuration of the generator (G) and discriminator (D) of the \emph{pure TCN} and \emph{C-SVNN} model.}
	\label{tab:tcn_parameters}
\end{table}

\end{appendix}
\end{document}